\documentclass[twocolumn,secnumarabic,amssymb, nobibnotes, aps, prd, superscriptaddress,longbibliography]{revtex4-2}
\usepackage{booktabs}
\usepackage{xcolor}
\usepackage{tikz}
\usepackage{tikz-3dplot}
\usetikzlibrary{3d}
\usepackage{float}
\usepackage{microtype}
\usepackage{ragged2e}
\usepackage{amsmath,amssymb,bm,amsfonts}
\usepackage{amsthm}
\usepackage{bm}
\usepackage{braket}
\usepackage{placeins}
\usepackage{graphicx}
\usepackage{epstopdf}
\usepackage{booktabs}
\usepackage{enumitem}
\usepackage{xcolor}

\usepackage{float}
\usepackage[justification=centering]{caption}
\usepackage{subcaption}

\usepackage[ruled,vlined]{algorithm2e}
\usepackage[title]{appendix}

\usepackage{adjustbox}
\usepackage{textcomp}
\usepackage{mathrsfs}

\usepackage[
bookmarks=true,
colorlinks=true,
linkcolor=blue,
urlcolor=blue,
citecolor=blue,
bookmarks=true,
hyperindex=true
]{hyperref}

\makeatletter

\newcommand{\Rmnum}[1]{\expandafter\@slowromancap\romannumeral #1@}
\makeatother

\begin{document}
	
	\title{Local distinguishability of five quadripartite orthogonal product states}
	
  \author{Guang-Bao  \surname{Xu}}%
   \affiliation{College of Computer Science and Engineering, Shandong University of Science and Technology, Qingdao, 266590, China}
   \affiliation{Shandong Key Laboratory of Smart Mine Information Technology,  Qingdao, 266590, China}

	\author{Jie \surname{Zhan}}%
    \affiliation{College of Computer Science and Engineering, Shandong University of Science and Technology, Qingdao, 266590, China}
    
   \author{Yu-Guang \surname{Yang}}%
   \affiliation{College of Cyberspace Science and Technology, Beijing University of Technology, Beijing, 100124, China}

   \author{Dong-Huan \surname{Jiang}}
   \email{donghuan\_jiang@163.com}
   \affiliation{College of Mathematics and Systems Science, Shandong University of Science and Technology, Qingdao, 266590, China}
	
	\date{\today}
	
	\begin{abstract}
     The local distinguishability of five bipartite and tripartite orthogonal product states (OPSs) has been effectively resolved, while that of five quadripartite OPSs has not yet been effectively addressed. In this paper, we give the local distinguishability of five quadripartite OPSs, where any two of these five states are orthogonal on exactly one subsystem.
     Based on the vectors of the numbers of pairwise orthogonality relations, we classify different sets of five quadripartite OPSs into nine categories, which correspond to 371 distinct orthogonality graphs. We show that the sets of five quadripartite OPSs corresponding to 355 of the 371 cases can be perfectly distinguished by local operations and classical communication (LOCC), whereas the sets of five quadripartite OPSs for the remaining 16 cases cannot be exactly distinguished by LOCC. Our results enrich the theory of nonlocality for sets of OPSs.
     
     \end{abstract}
	
	\maketitle

    \section{\label{sec1}Introduction\protect}
    
    The local distinguishability of orthogonal product states (OPSs) constitutes a fundamental problem in quantum information theory. The core question is whether a set of OPSs can be perfectly distinguished by spatially separated parties under local operations and classical communication (LOCC). The phenomenon that a set of OPSs cannot be perfectly distinguished by LOCC was first revealed by Bennett \emph{et al.}  \cite{ref1} in 1999. This phenomenon, commonly known as nonlocality without entanglement, shows that entanglement is not the sole reason why a set of orthogonal quantum states cannot be locally discriminated. In the same year, unextendible product bases were proved that cannot be perfectly distinguished by LOCC \cite{ref2}. All these works have aroused considerable interest \cite{ref3,ref4,ref5}.
    
    Inspired by the works of Bennett et al., many researchers have started to investigate quantum nonlocality without entanglement, and a great number of results have been proposed. These research results fall mainly into the following categories. The first category concerns the theory of constructing sets of OPSs that exhibit nonlocality. Zhang \emph{et al.}  \cite{ref6, ref8} and Wang  \emph{et al.} \cite{ref7} proposed different methods to construct a nonlocal set of bipartite OPSs. Xu \emph{et al.} \cite{ref12} gave a method to construct a nonlocal set of multipartite OPSs. With further advances in research, an increasing number of results \cite{ref9,ref13,ref14,ref15,ref16,ref17,ref18} have been obtained. The second category focuses on research related to strong nonlocality. In 2019, Halder \emph{et al.} \cite{ref19} proposed the concept of strong quantum nonlocality. Subsequently, many results \cite{ref20,ref21,ref22,ref23,ref24,ref25} on the construction methods of strongly nonlocal sets have been proposed. The third category deals with the minimal number of elements required for a nonlocal set to retain its nonlocality, namely minimal nonlocality. Zhu \emph{et al.} \cite{ref26} proposed the concept of minimal nonlocality, and gave a method to construct a nonlocal set of OPSs with minimal nonlocality in bipartite system. Xu \emph{et al.} \cite{ref27} gave a method to construct a nonlocal set of orthogonal quantum states in any multipartite quantum system. All the above‑mentioned works have greatly enriched the theory of local distinguishability of quantum states.
    
    Recently, the local distinguishability of orthogonal product states with a given number has attracted considerable research interest. Walgate \emph{et al.}  \cite{ref3} proved that every pair of orthogonal pure states can be exactly distinguished by LOCC. DiVincenzo \emph{et al.} \cite{ref4} proved that any multipartite OPSs set with at most
    three members, as well as any bipartite OPSs set with at most four members, is locally distinguishable. Xu \emph{et al.} gave the local distinguishability of four OPSs in a multipartite quantum system \cite{ref31}, and that of five OPSs in bipartite and tripartite systems \cite{ref32}. These works have perfected and enriched the nonlocality theory of sets of OPSs. 
    
    Different from the previous results, we investigate the local distinguishability of five  OPSs in this paper, where any two of those five OPSs are orthogonal on exactly one subsystem. By analyzing the distribution of pairwise orthogonality relations across four subsystems, we classify the corresponding orthogonality structures into 371 cases. We prove that a set of five  OPSs corresponding to any of the 355 cases is locally distinguishable while that corresponding to the remaining 16 cases cannot be perfectly distinguished by LOCC. The rest of the paper is structured as follows. In Sec.~\ref{sec2}, some  necessary definitions and preliminary results are introduced. In Sec.~\ref{sec3} we give the local distinguishability of five  OPSs. In Sec.~\ref{sec4}, the 16 cases that cannot be perfectly distinguished by LOCC are discussed. A brief conclusion is given in Sec.~\ref{sec5}. 
    
    \section{\label{sec2}Preliminaries}
	\theoremstyle{remark}
	\newtheorem{definition}{\indent Definition}
	\newtheorem{lemma}{\indent Lemma}
	\newtheorem{theorem}{\indent Theorem}
	\newtheorem{corollary}{\indent Corollary}
	\def\QEDclosed{\mbox{\rule[0pt]{1.3ex}{1.3ex}}}
	\def\QED{\QEDclosed}
	\def\proof{\indent{\em Proof}.}
	\def\endproof{\hspace*{\fill}~\QED\par\endtrivlist\unskip}
	
	In this section, several notions used in the subsequent analysis are introduced.
	
	To characterize the orthogonality structure of a set of OPSs, we first recall the notion of an orthogonality graph.
	
	\begin{definition}
		\label{def1} \cite{ref4}
		Let $\mathcal{H}=\bigotimes^{m}_{i=1}\mathcal{H}_{i}$ be an $m$-partite Hilbert space with dim $\mathcal{H}_{i}=d_{i}$. Let 
$S=\{\vert \psi_{j}\rangle\equiv\bigotimes^{m}_{i=1}\vert \varphi_{i,j}\rangle|j=1,2,...,n\}$ be an orthogonal product basis in $\mathcal{H}$. We represent $S$ as a graph $G = (V, E_{1}\cup E_{2}\cup \cdots \cup E_{m})$, where the set of edges $E_{i}$ have color $i$. The states $\vert \psi_{j} \rangle \in S$ are represented as the vertices $V$. There exists an edge $e$ of color $i$ between the vertices $v_{k}$ and $v_{l}$, i.e., $e \in E_{i}$, when states $\vert \psi_{k}\rangle$ and $\vert \psi_{l}\rangle$ are orthogonal on $\mathcal{H}_{i}$. Since all the states in the product basis are mutually orthogonal, every vertex is connected to all the other vertices by at least one edge of some color. The graph $G$ is called the orthogonality graph of the product basis.
	\end{definition}
	
	In this paper, we investigate the local distinguishability of five  OPSs, where any two of those five states are orthogonal on exactly one subsystem. Hence, the corresponding orthogonality graph is a complete graph with five vertices, and every edge has a unique color.
	
	\begin{definition}
		\label{def2}
	For five  OPSs, a quadruple $(a,b,c,d)$ is called the vector of the numbers of pairwise orthogonality relations, where $a$ denotes the number of pairwise orthogonality relations on the first subsystem, while $b$, $c$, and $d$ denote the corresponding numbers on the second, third, and fourth subsystems, respectively.
	\end{definition}
	
	\begin{definition}
		\label{def3}\cite{ref33}  
    Let $G=(V,\,E_1\cup E_2\cup E_3\cup E_4)$ be the orthogonality graph of five  OPSs, and let $v_j\in V$ be a vertex of $G$. The degree of the $i$-th party of $v_j$, denoted 
  by $\deg_i(v_j)$, is the number of edges incident with $v_j$ that belong to $E_i$.
	\end{definition}
	
	Note that $\deg_i(v_j)$ represents the number of product states that are
	orthogonal to the state corresponding to $v_j$ on the $i$-th subsystem.
	
	\begin{definition}
		\label{def:local-neighborhood}
	Let $v$ be a vertex of the orthogonality graph $G=(V,\,E_1\cup E_2\cup E_3\cup E_4)$. For each subsystem $p\in\{1,2,3,4\}$, let $N_p(v)=\{u\in V\setminus\{v\}\mid \{u,v\}\in E_p\}$ denote the set of vertices whose corresponding states are orthogonal to the state corresponding to $v$ on the $p$-th subsystem. Clearly, $|N_p(v)|=\deg_p(v)$.
	\end{definition}
	
	\begin{definition} \cite{ref4,ref31}
		\label{def4}
		Let $\mathcal{H}=\bigotimes_{i=1}^{m}\mathcal{H}_i$ be an $m$-partite Hilbert
		space with $\dim\mathcal{H}_i=d_i$, and let
		$\mathcal{S}=\{\ket{\psi_j}\equiv\bigotimes_{i=1}^{m}\ket{\phi_{i,j}}
		\mid j=1,2,\ldots,n\}$ be a set of OPSs in $\mathcal{H}$. The local rank of
		$\mathcal{S}$ on the $i$-th subsystem, denoted by $r_i$, is defined as the rank
		of the vector set $\{\ket{\phi_{i,j}}\mid j=1,2,\ldots,n\}$.
	\end{definition}
	
	\begin{lemma}\cite{ref4}
		\label{lem:four-bipartite-ops}
		Let $\mathcal{S}$ be a set of bipartite OPSs with four or fewer members in any dimension where such a set can exist. Then $\mathcal{S}$ can be distinguished by local incomplete von Neumann measurements and classical communication.
	\end{lemma}	
	\begin{lemma}\cite{ref4}  
		\label{lem:three-ops}
		Any set of multipartite OPSs with three or fewer members is distinguishable by local incomplete von Neumann measurements and classical communication.
	\end{lemma}
	It should be noted that a von Neumann measurement is a projection measurement, which can be regarded as a special case of a POVM.
	
	According to the vector of the numbers of pairwise orthogonality relations, four tripartite OPSs, any two of which are orthogonal on exactly one subsystem, can
	be divided into three categories, i.e., $(4,1,1)$, $(3,2,1)$, and $(2,2,2)$. The following results will be used in what follows.
	\begin{lemma} \cite{ref31}
		\label{lem:four-tripartite-411-321}
		Four tripartite OPSs with the vector of the numbers of pairwise orthogonal
		relations $(4,1,1)$ or $(3,2,1)$ can be perfectly distinguished by LOCC.
	\end{lemma}
	For the category $(2,2,2)$, there are three different orthogonality graphs up to vertex relabeling and party permutation. We denote them by cases $(1\text{-}1)$, $(1\text{-}2)$, and $(1\text{-}3)$, as shown in Fig.~\ref{fig:four-tripartite-222}.
	\begin{figure}[t]
		\centering
		\includegraphics[width=0.7\columnwidth]{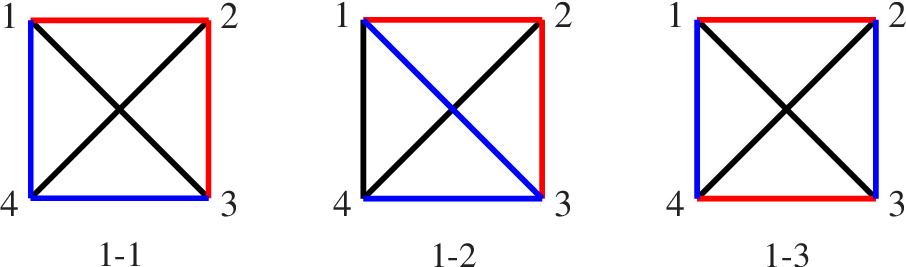}
		\caption{The feasible orthogonality graphs of four tripartite OPSs with the vector of the numbers of pairwise orthogonal relations $(2,2,2)$.}
		\label{fig:four-tripartite-222}
	\end{figure}
	\begin{lemma}\cite{ref31}
		\label{lem:four-tripartite-222-dist}
		Four tripartite OPSs with the vector of the numbers of pairwise orthogonality relations $(2,2,2)$ can be perfectly distinguished by LOCC if their orthogonality graph corresponds to case $(1\text{-}1)$ or case $(1\text{-}2)$ up to vertex relabeling and party permutation, as shown in Fig.~\ref{fig:four-tripartite-222}.
	\end{lemma}
	\begin{lemma} \cite{ref31}
		\label{lem:four-tripartite-222-exception}
		For four tripartite OPSs with the vector of the numbers of pairwise orthogonality relations $(2,2,2)$, if their orthogonality graph corresponds to case $(1\text{-}3)$ in 		Fig.~\ref{fig:four-tripartite-222} and the local ranks satisfy $r_1=r_2=r_3=2$, then they cannot be perfectly distinguished by LOCC. In the other local-rank cases, such a set can be distinguished by LOCC with a nonzero
		probability.
	\end{lemma}
	
	\begin{lemma}
		\label{lem:four-fourpartite-3111-2211}
		\cite{ref31} Four  OPSs with the vector of the numbers of pairwise orthogonality relations $(3,1,1,1)$ or $(2,2,1,1)$ can be perfectly distinguished by LOCC.
	\end{lemma}
	The following two reduction lemmas can be derived from the local measurement arguments in Ref.~\cite{ref35} and will be used repeatedly in what follows.
	\begin{lemma} ~\cite{ref35}
		\label{lem:degree4-removal}
	Let $\mathcal{S}$ be a set of five  OPSs, any two of which are orthogonal on exactly one subsystem, and the orthogonality graph of $S$ is $G=(V,\,E_1\cup E_2\cup E_3\cup E_4)$. Suppose that there exists a vertex $v\in V$ and a subsystem $p$ such that $\deg_p(v)=4$. Then the state corresponding to $v$ can be exactly identified by LOCC.
	\end{lemma}
	
	\begin{lemma} \cite{ref35} 
		\label{lem:31-removal}
		Let $\mathcal{S}$ be a set of five  OPSs, any two of which are orthogonal on exactly one subsystem, and the orthogonality graph of $S$ is $G=(V,\,E_1\cup E_2\cup E_3\cup E_4)$. Suppose that there exists a vertex $v\in V$ and two distinct subsystems $p$ and $q$ such that $\deg_p(v)=3$ and $\deg_q(v)=1$. Then the state corresponding to $v$ can be distinguished by LOCC.
	\end{lemma}	
	
	\section{\label{sec3}Local distinguishability of five quadripartite OPS\lowercase{s}}
	
	In this section, we consider the local distinguishability of five quadripartite OPSs, where any two of these states are orthogonal on exactly one subsystem. According to the vectors of pairwise orthogonal relations, the sets of five quadripartite OPSs can be divided into three types. The classification is shown in Table~\ref{tab:fourpartite-classification}. The cases in the first two columns can be treated using the known results for bipartite or tripartite systems. Therefore, in what follows, we mainly focus on the genuine  categories listed in the third column.
	
	\begin{table}[t]
		\caption{Categories of five quadripartite OPSs by the vectors of the numbers of pairwise orthogonal relations}
		\label{tab:fourpartite-classification}
		\centering
		\small
		\begin{ruledtabular}
			\begin{tabular}{ccc}
				\shortstack{Reducible to\\ bipartite categories}
				&
				\shortstack{Reducible to\\ tripartite categories}
				&
				\shortstack{Genuine categories}
				\\
				\hline
				$(10,0,0,0)$ & $(8,1,1,0)$ & $(7,1,1,1)$ \\
				$(9,1,0,0)$  & $(7,2,1,0)$ & $(6,2,1,1)$ \\
				$(8,2,0,0)$  & $(6,3,1,0)$ & $(5,3,1,1)$ \\
				$(7,3,0,0)$  & $(6,2,2,0)$ & $(5,2,2,1)$ \\
				$(6,4,0,0)$  & $(5,4,1,0)$ & $(4,4,1,1)$ \\
				$(5,5,0,0)$  & $(5,3,2,0)$ & $(4,3,2,1)$ \\
				& $(4,4,2,0)$ & $(4,2,2,2)$ \\
				& $(4,3,3,0)$ & $(3,3,3,1)$ \\
				&               & $(3,3,2,2)$ \\
			\end{tabular}
		\end{ruledtabular}
	\end{table}
	
	For each genuine category, we enumerate all feasible orthogonality graphs of five quadripartite OPSs. All feasible orthogonality graphs for categories $(7,1,1,1)$, $(6,2,1,1)$, $(5,3,1,1)$, $(5,2,2,1)$, $(4,4,1,1)$, $(4,3,2,1)$, $(4,2,2,2)$, $(3,3,3,1)$, and $(3,3,2,2)$ are given in Figs.~\ref{fig:7111}--\ref{fig:3322-part2}. Note that we have omitted configurations that coincide with listed ones either under party interchange or after state relabeling, as those cases will follow the same line of reasoning. The counts of locally distinguishable cases and of exceptional cases for five quadripartite OPSs are shown in Table~\ref{tab:genuine-fourpartite-counts}. From Table~\ref{tab:genuine-fourpartite-counts}, we know that the sets of five quadripartite OPSs corresponding to the categories $(7,1,1,1)$, $(6,2,1,1)$, $(5,3,1,1)$, and $(4,4,1,1)$ can be perfectly distinguished by LOCC. For the other categories, there exist some exceptional cases. 
		
	\begin{figure}[t]
		\centering
		\includegraphics[width=0.8\columnwidth]{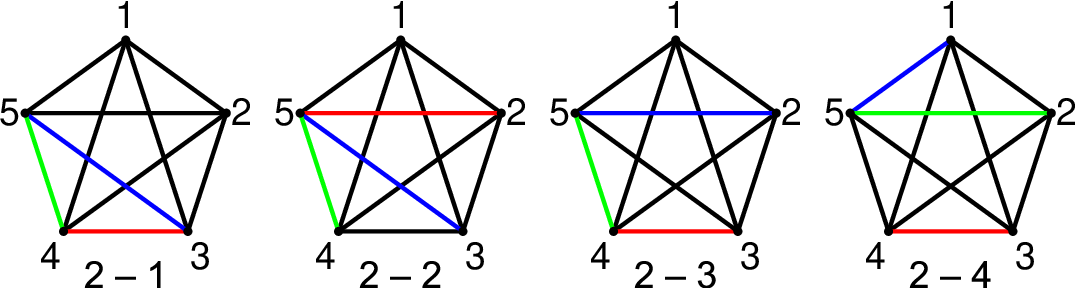}
		\caption{All feasible orthogonality graphs for category $(7,1,1,1)$.}
		\label{fig:7111}
	\end{figure}
	\begin{figure*}[t]
		\centering
		\includegraphics[width=0.6\textwidth,height=0.4\textheight,keepaspectratio]{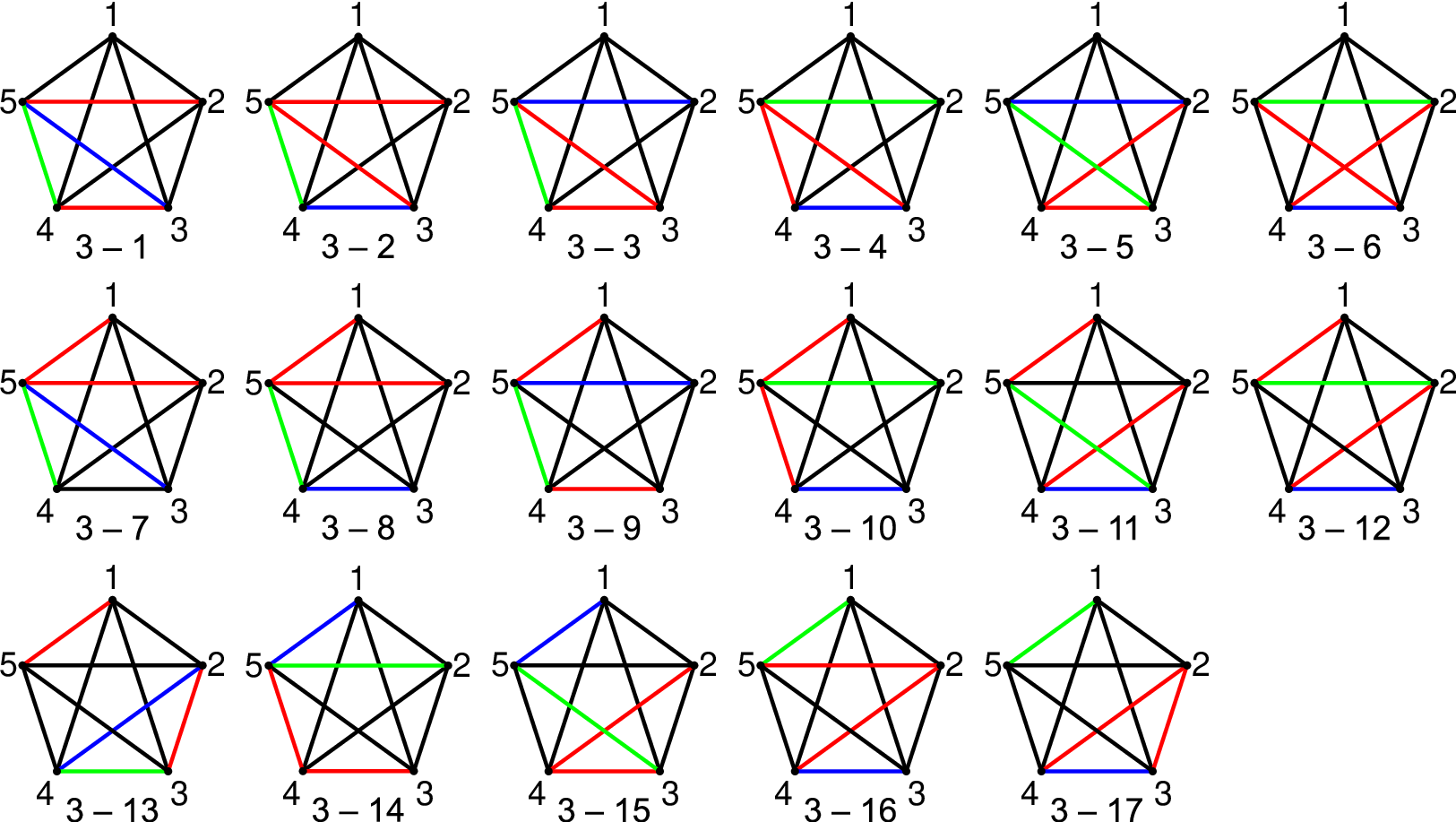}
		\caption{All feasible orthogonality graphs for category $(6,2,1,1)$.}
		\label{fig:6211}
	\end{figure*}
	\begin{figure*}[t]
		\centering
		\includegraphics[width=0.6\textwidth,height=0.4\textheight,keepaspectratio]{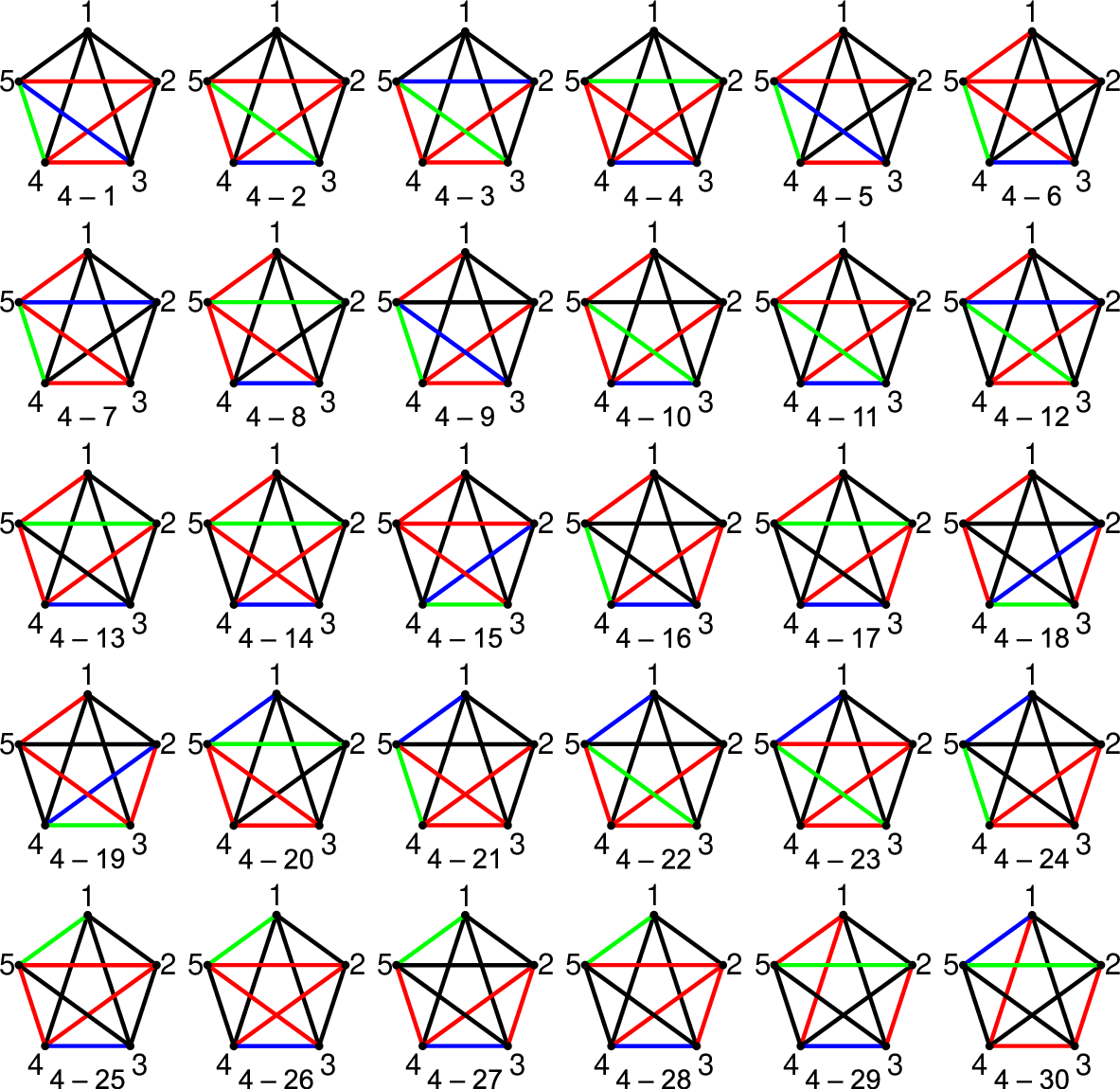}
		\caption{All feasible orthogonality graphs for category $(5,3,1,1)$.}
		\label{fig:5311}
	\end{figure*}
	\begin{figure*}[t]
		\centering
		\includegraphics[width=0.8\textwidth,height=0.4\textheight,keepaspectratio]{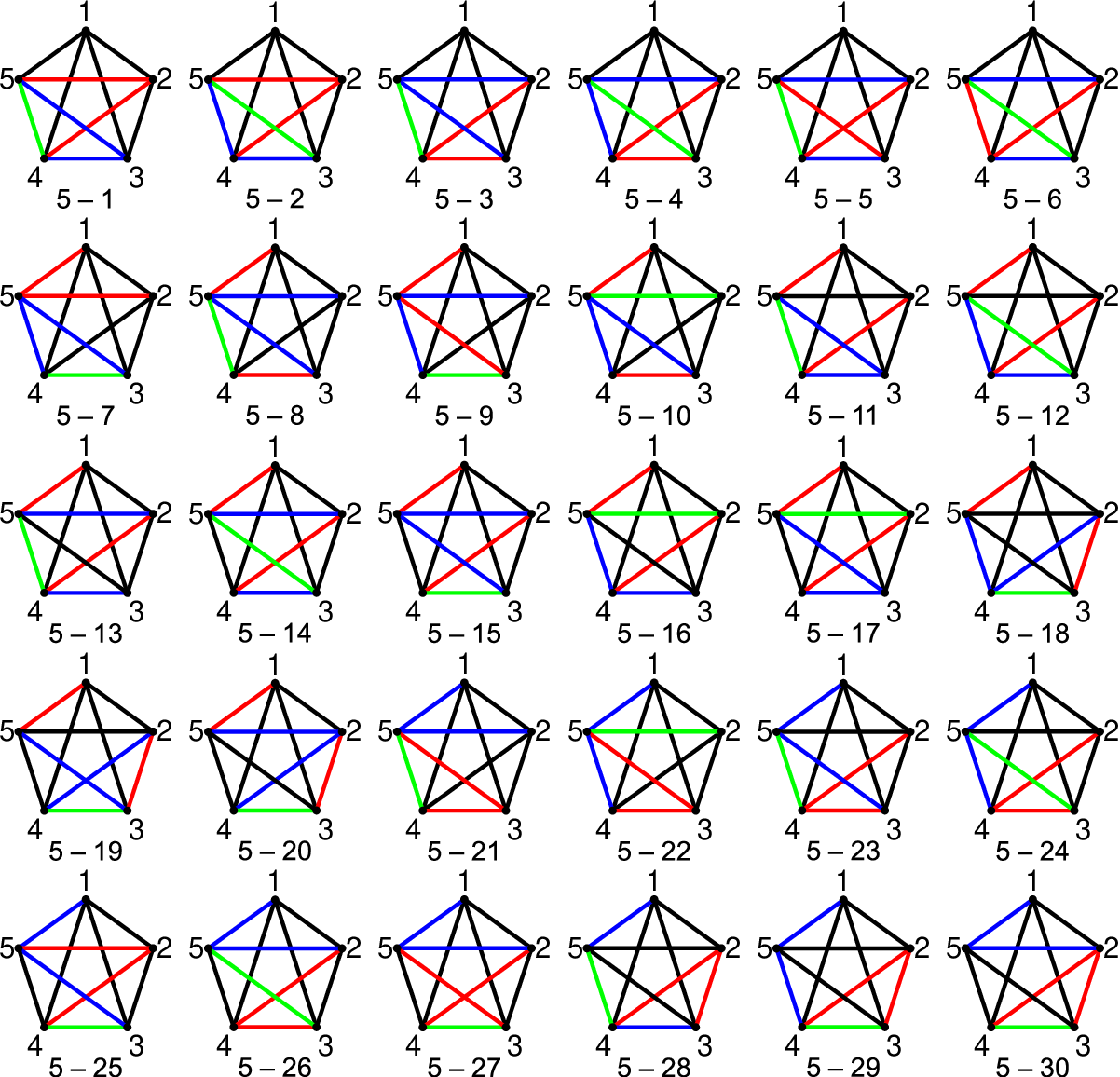}	
		\includegraphics[width=0.55\textwidth,height=0.3\textheight,keepaspectratio]{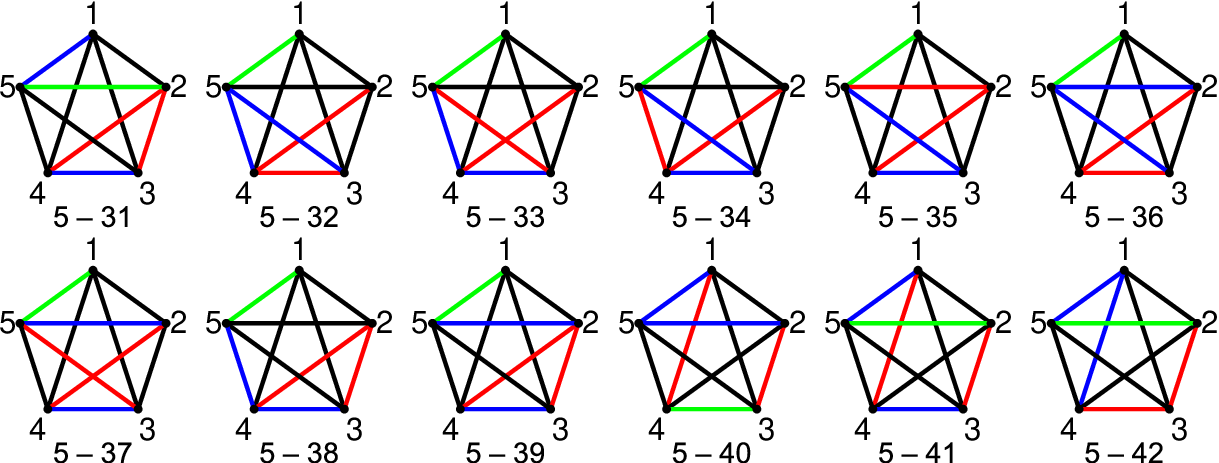}
		
		\caption{All feasible orthogonality graphs for category $(5,2,2,1)$.}
		\label{fig:5221}
	\end{figure*}
	\begin{figure*}[t]
		\centering
		\includegraphics[width=0.6\textwidth,height=0.4\textheight,keepaspectratio]{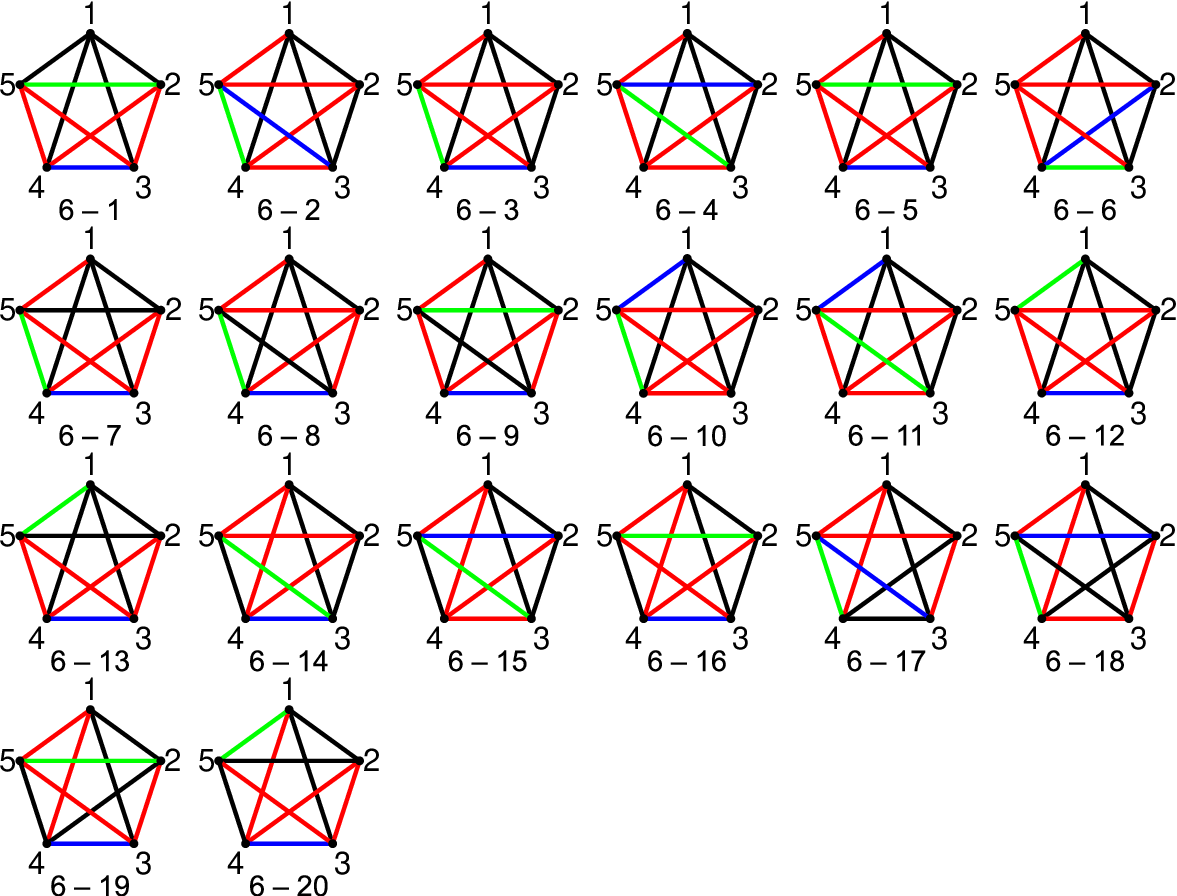}
		\caption{All feasible orthogonality graphs for category $(4,4,1,1)$.}
		\label{fig:4411}
	\end{figure*}
	\begin{figure*}[t]
		\centering
		\includegraphics[width=\textwidth,height=0.5\textheight,keepaspectratio]{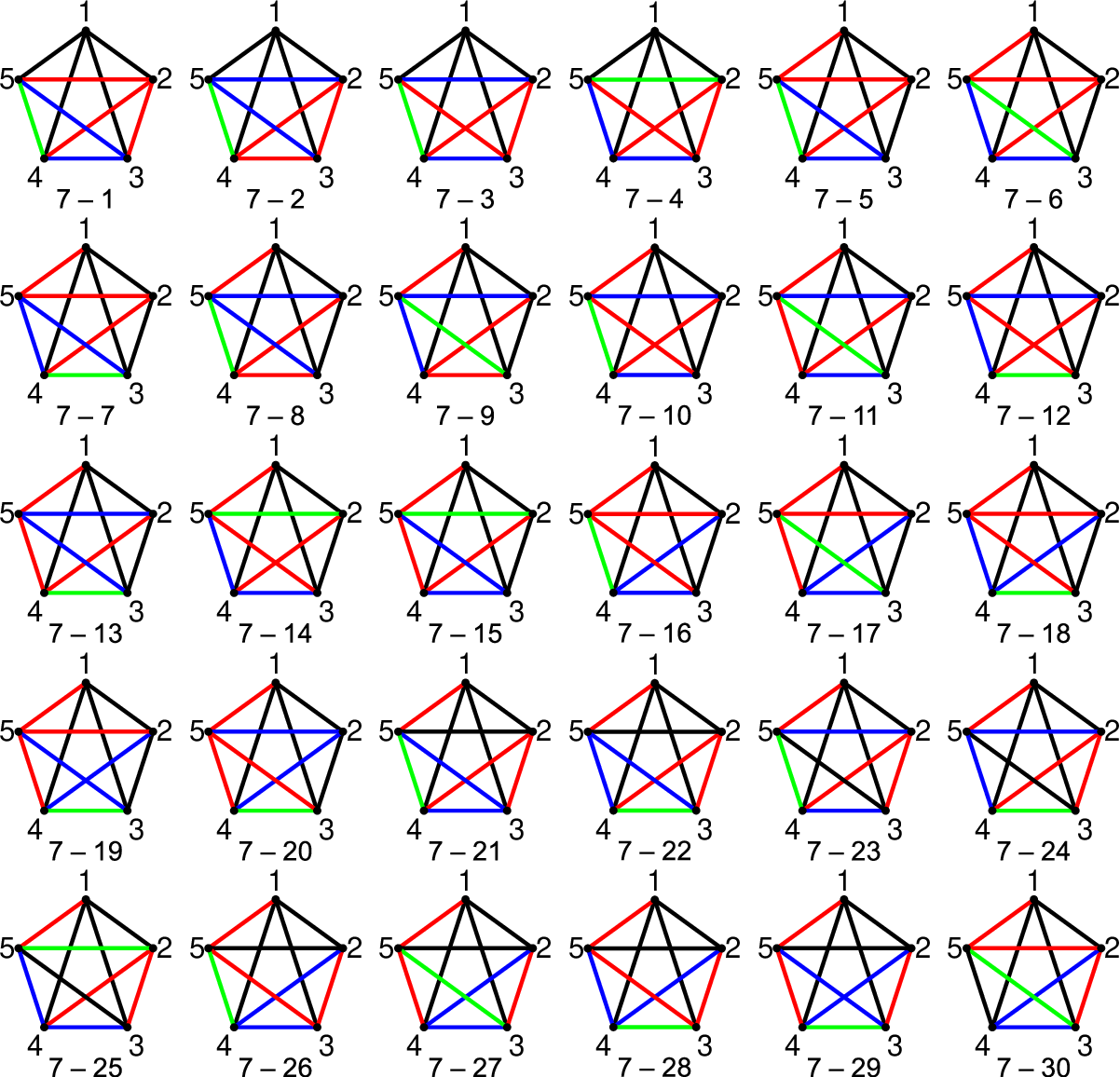}
		
		\vspace{2mm}
		
		\includegraphics[width=0.7\textwidth,height=0.5\textheight,keepaspectratio]{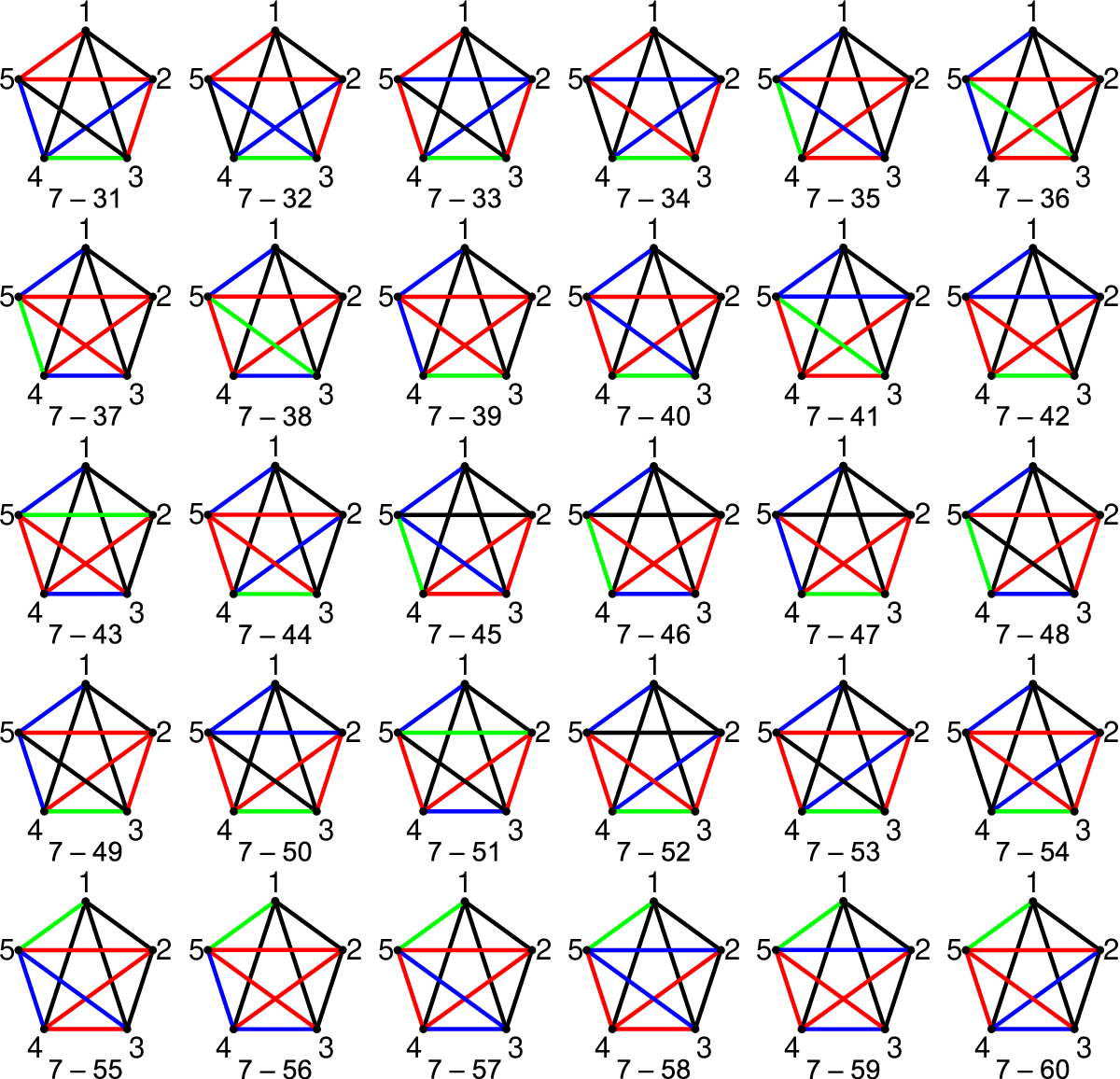}

		\label{fig:4321-part1}
	\end{figure*}
	
	\begin{figure*}[t]
		\centering
		\includegraphics[width=0.7\textwidth,height=0.5\textheight,keepaspectratio]{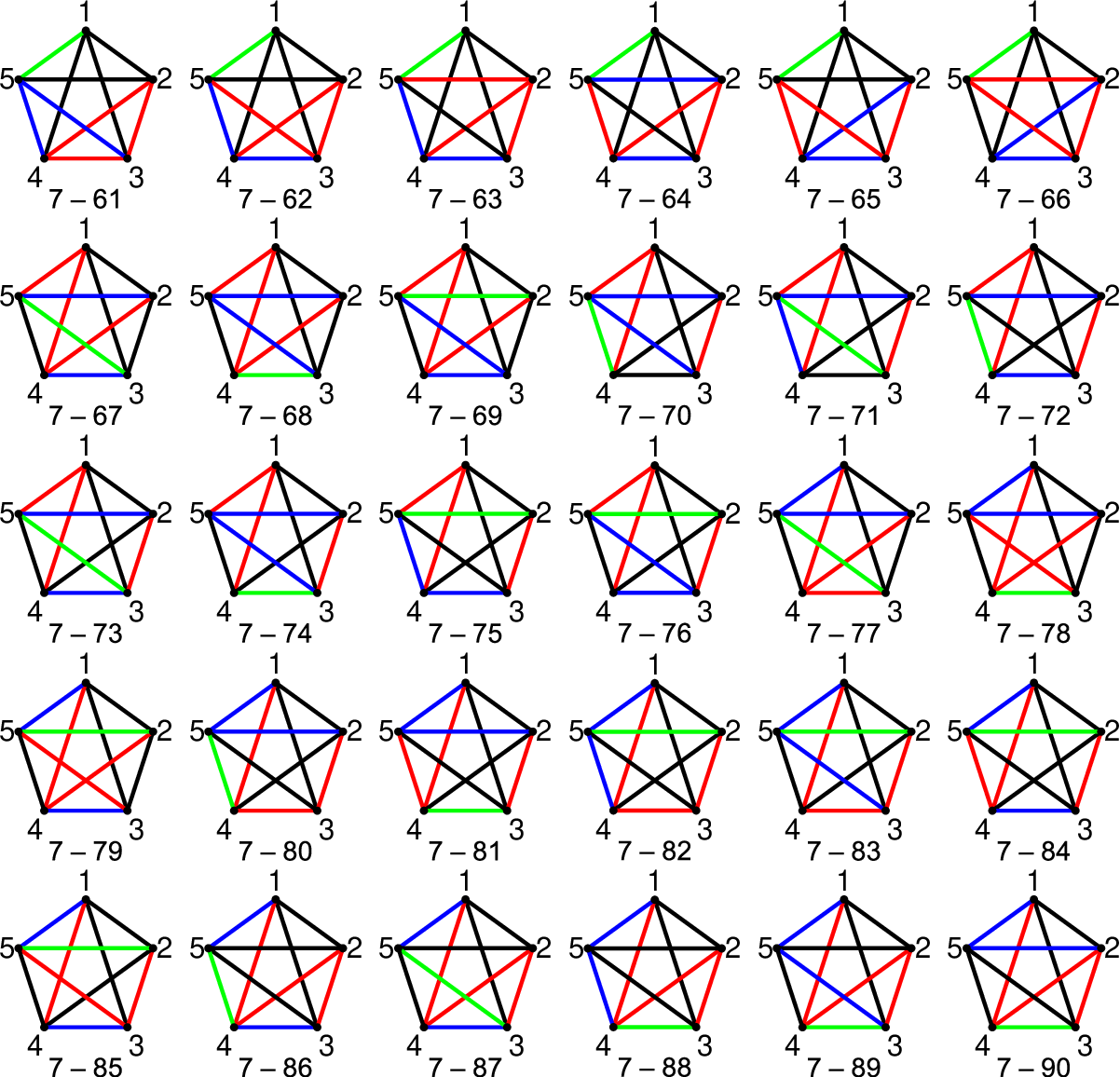}
		
		\vspace{2mm}
		
		\includegraphics[width=0.67\textwidth,height=0.5\textheight,keepaspectratio]{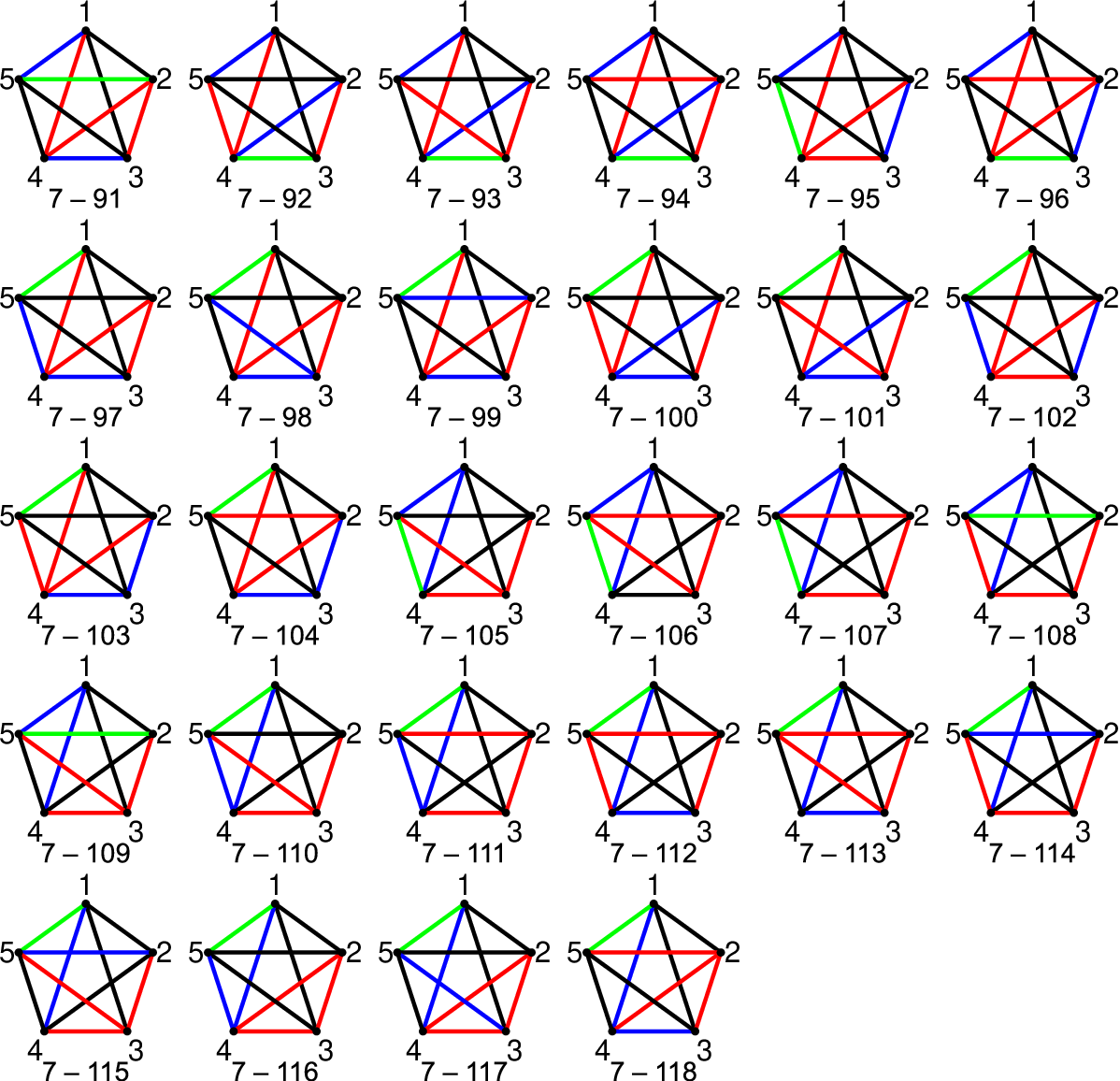}
		
		\caption{All feasible orthogonality graphs for category $(4,3,2,1)$.}
		\label{fig:4321-part2}
	\end{figure*}
	\begin{figure*}[t]
		\centering
		\includegraphics[width=0.7\textwidth,height=0.5\textheight,keepaspectratio]{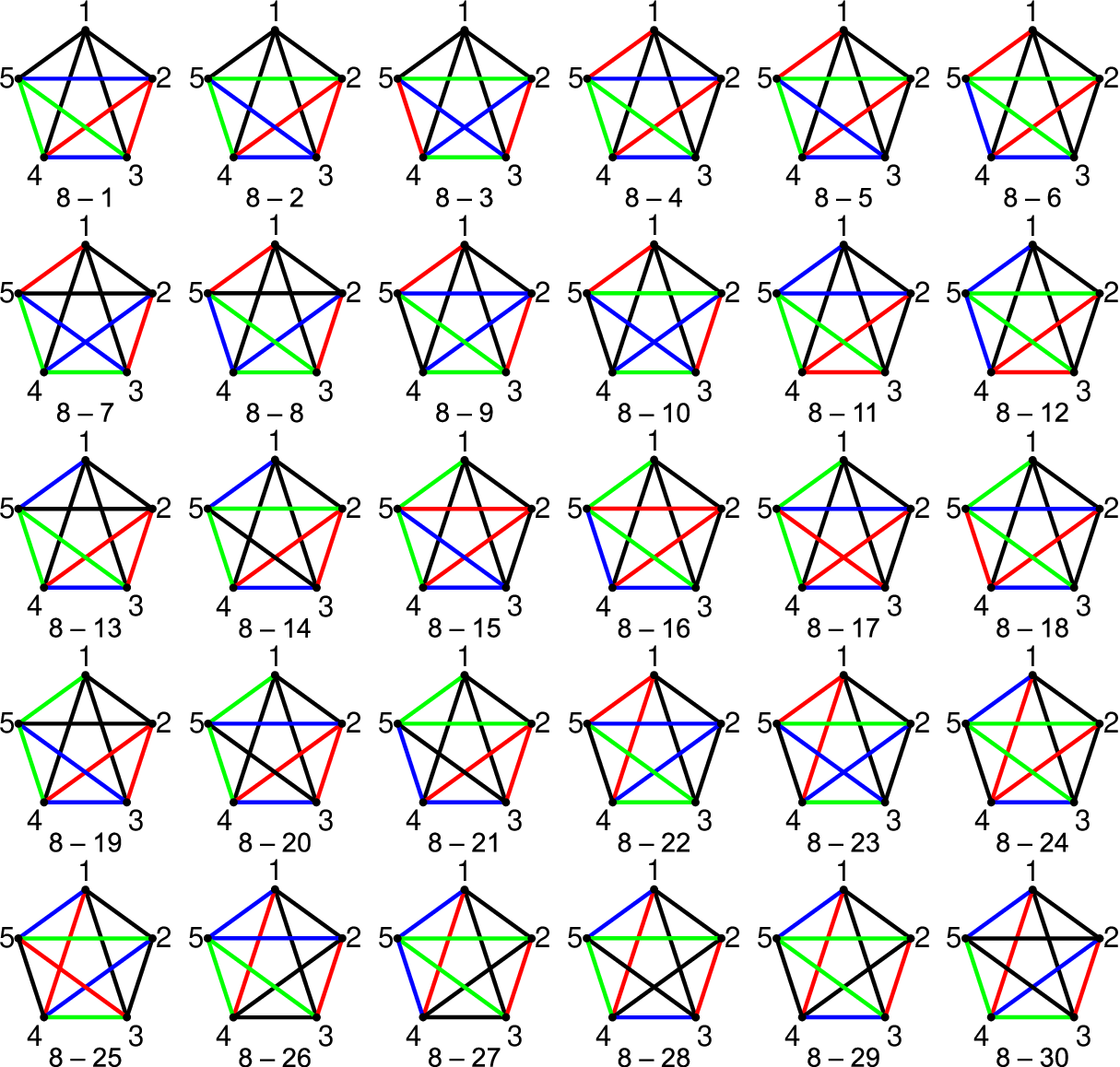}
		
		\vspace{2mm}
		
		\includegraphics[width=0.7\textwidth,height=0.21\textheight,keepaspectratio]{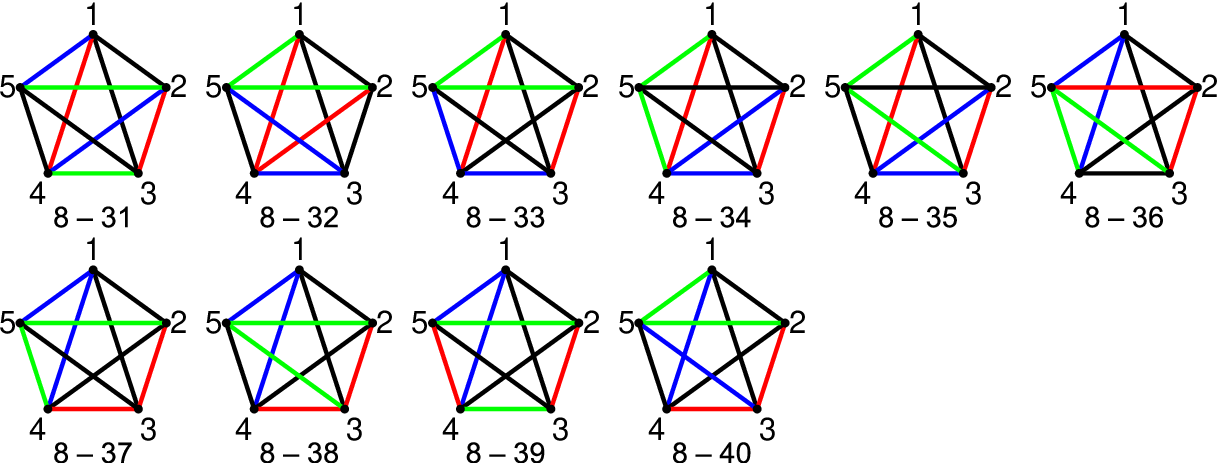}
		
		\caption{All feasible orthogonality graphs for category $(4,2,2,2)$.}
		\label{fig:4222}
	\end{figure*}
	\begin{figure*}[t]
		\centering
		\includegraphics[width=0.7\textwidth,height=0.62\textheight,keepaspectratio]{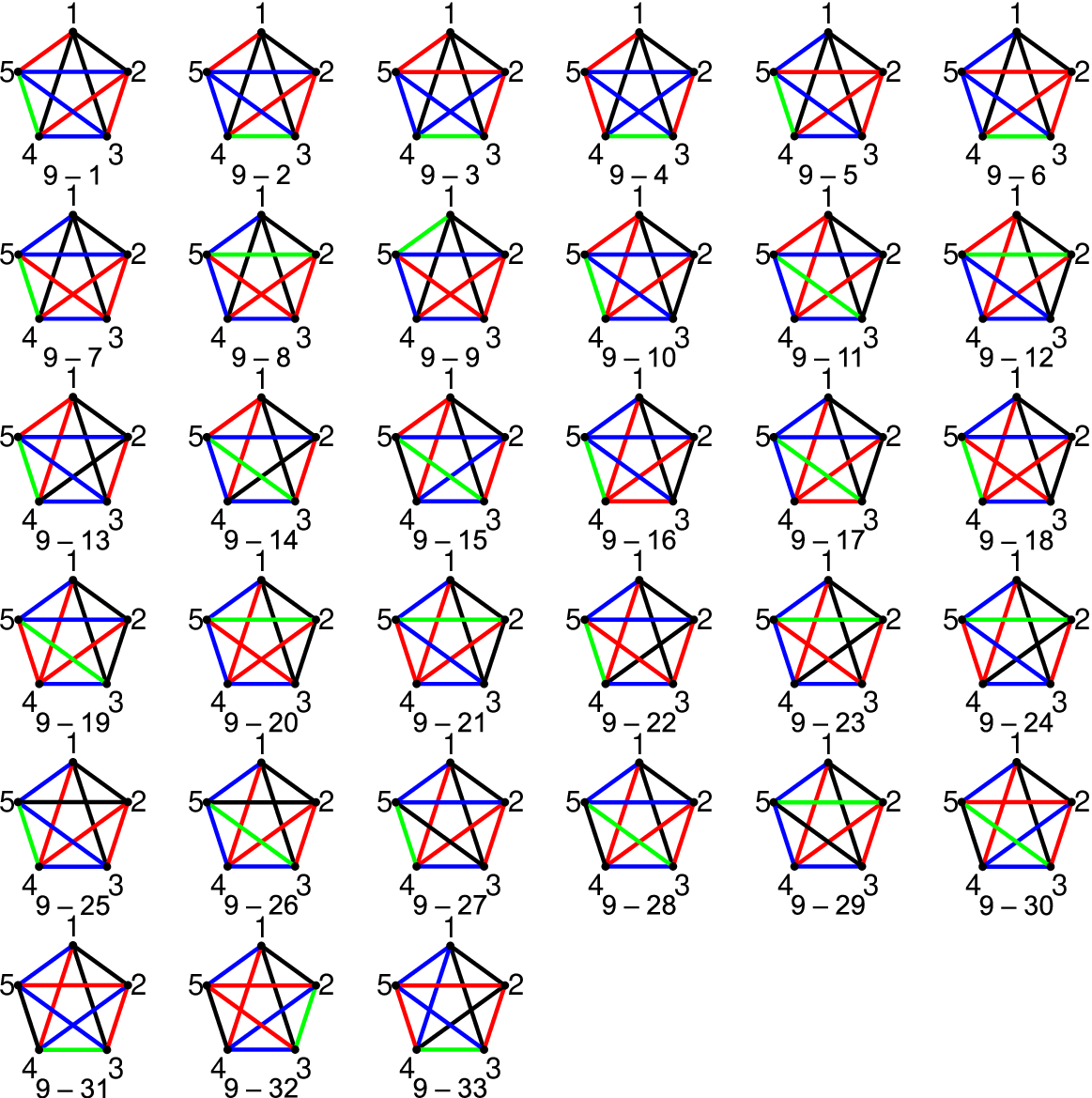}
		\caption{All feasible orthogonality graphs for category $(3,3,3,1)$.}
		\label{fig:3331}
	\end{figure*}
	\begin{figure*}[t]
		\centering
		\includegraphics[width=0.98\textwidth,height=0.5\textheight,keepaspectratio]{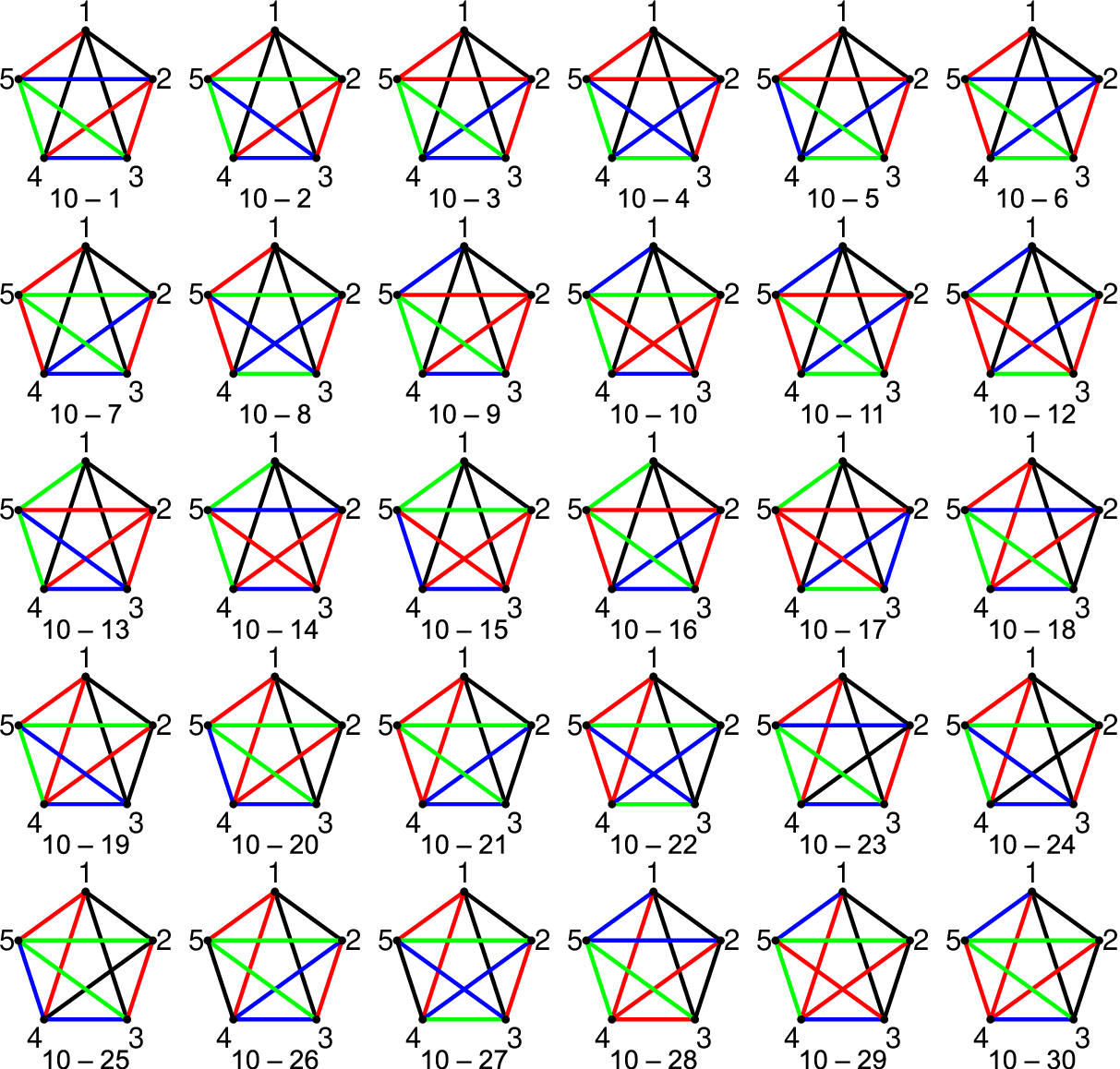}
		
		\includegraphics[width=0.98\textwidth,height=0.5\textheight,keepaspectratio]{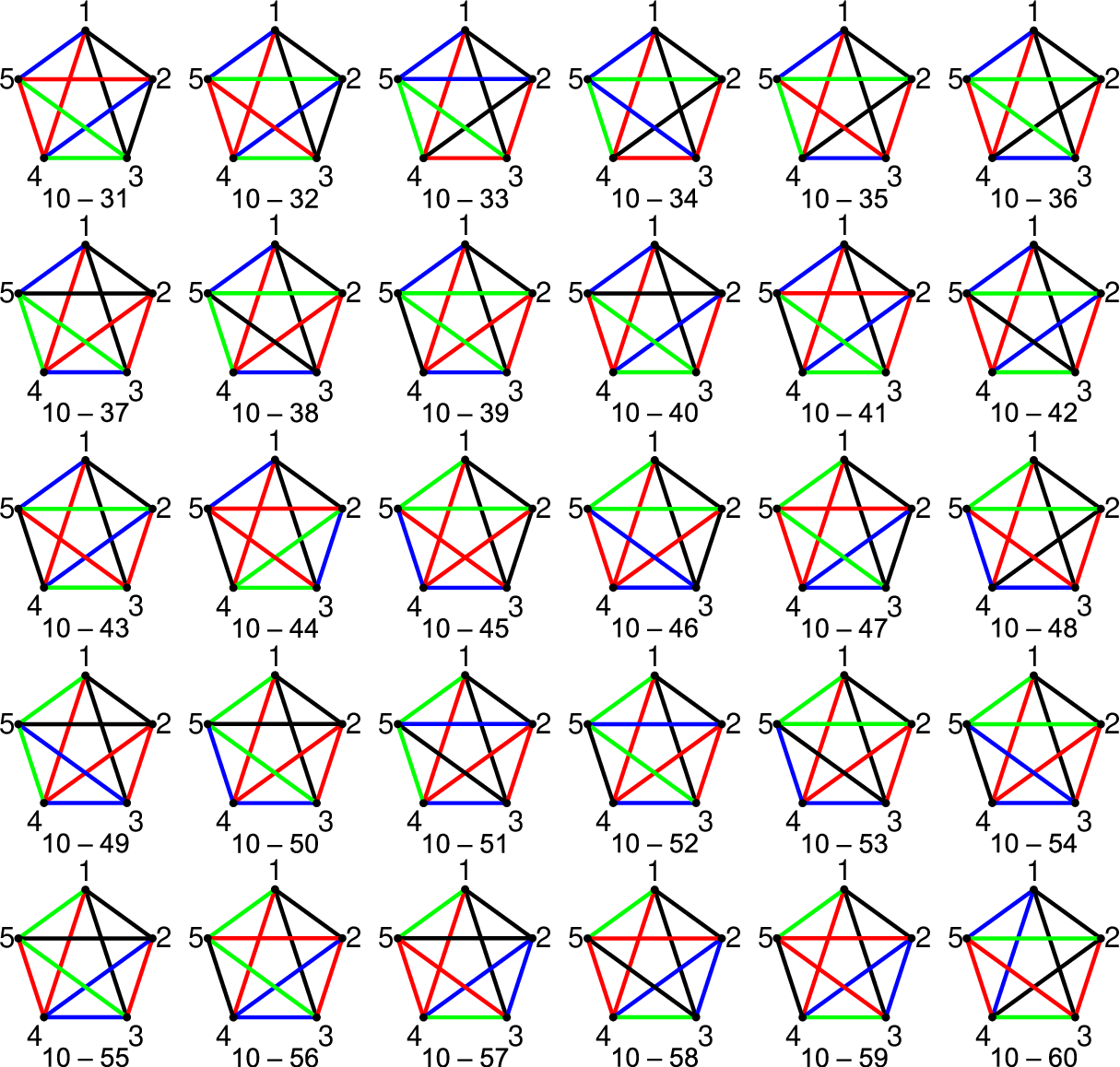}
		
		\label{fig:3322-part1}
	\end{figure*}
	
	\begin{figure*}[!t]
		\centering
		\includegraphics[
		width=0.98\textwidth,
		height=0.2\textheight,
		keepaspectratio
		]{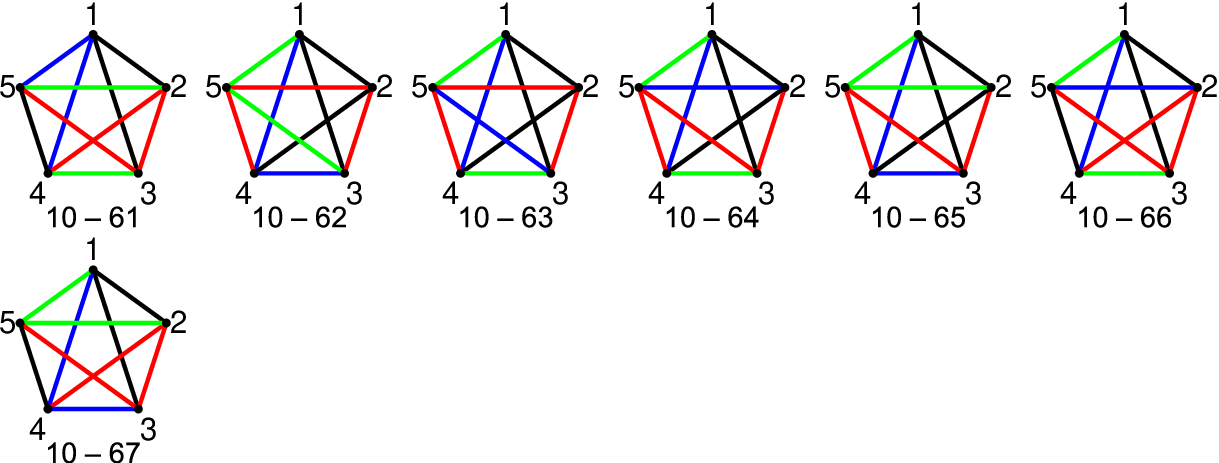}
		
		\caption{All feasible orthogonality graphs for category $(3,3,2,2)$.}
		\label{fig:3322-part2}
	\end{figure*}
	\begin{table}[t]
		\caption{All feasible configurations of five quadripartite OPSs}
		\label{tab:genuine-fourpartite-counts}
		\begin{ruledtabular}
			\begin{tabular}{cccc}
				Categories & \shortstack{Distinguishable\\cases} & \shortstack{Exceptional \\cases} & \shortstack{Total\\cases} \\
				\hline
				$(7,1,1,1)$ & 4   & 0 & 4   \\
				$(6,2,1,1)$ & 17  & 0 & 17  \\
				$(5,3,1,1)$ & 30  & 0 & 30  \\
				$(5,2,2,1)$ & 41  & 1 & 42  \\
				$(4,4,1,1)$ & 20  & 0 & 20  \\
				$(4,3,2,1)$ & 115 & 3 & 118 \\
				$(4,2,2,2)$ & 36  & 4 & 40  \\
				$(3,3,3,1)$ & 30  & 3 & 33  \\
				$(3,3,2,2)$ & 62  & 5 & 67  \\
			\end{tabular}
		\end{ruledtabular}
	\end{table}

	\begin{theorem} 
		\label{thm:22-reduction}
	Let $\mathcal{S}$ be a set of five  OPSs, in which any two states are orthogonal on exactly one subsystem. Suppose that there exists a vertex $v$ in the orthogonality graph of $S$ and two distinct subsystems $p$ and $q$ such that $\deg_p(v)=\deg_q(v)=2$. Let \(N_p(v)=\{u_1,u_2\}\) and \(N_q(v)=\{w_1,w_2\}\). If the edge between $w_1$ and $w_2$ is not of color $p$, then the $p$-th party can perform an orthogonality-preserving local measurement on the $p$-th subsystem such that the state corresponding to $v$ can be identified by LOCC. 
	\end{theorem}
	
	\begin{proof}
		For each vertex $x\in\{v,u_1,u_2,w_1,w_2\}$, write the corresponding state as $|\psi_x\rangle=|\phi_{p,x}\rangle_p|\eta_x\rangle_{\bar p}$, where $|\eta_x\rangle_{\bar p}=\bigotimes_{r\neq p}|\phi_{r,x}\rangle_r$. Here $\bar p$ denotes all subsystems except the $p$-th subsystem. All local states are assumed to be normalized.
		
		Since $N_p(v)=\{u_1,u_2\}$, we have $\langle\phi_{p,v}|\phi_{p,u_1}\rangle=\langle\phi_{p,v}|\phi_{p,u_2}\rangle=0$. The $p$-th party performs a two-outcome local measurement with measurement operators $M_1=|\phi_{p,v}\rangle\langle\phi_{p,v}|$ and $M_2=I-|\phi_{p,v}\rangle\langle\phi_{p,v}|$, where $M_1^\dagger M_1+M_2^\dagger M_2=I$.
		
		If the measurement outcome corresponds to $M_1$, then the measured state must correspond to $v$, $w_1$, or $w_2$. More precisely,
		\[
		\begin{aligned}
			(M_1\otimes I_{\bar p})|\psi_v\rangle
			&=|\phi_{p,v}\rangle_p|\eta_v\rangle_{\bar p},\\
			(M_1\otimes I_{\bar p})|\psi_{w_1}\rangle
			&=\langle\phi_{p,v}|\phi_{p,w_1}\rangle
			|\phi_{p,v}\rangle_p|\eta_{w_1}\rangle_{\bar p},\\
			(M_1\otimes I_{\bar p})|\psi_{w_2}\rangle
			&=\langle\phi_{p,v}|\phi_{p,w_2}\rangle
			|\phi_{p,v}\rangle_p|\eta_{w_2}\rangle_{\bar p}.
		\end{aligned}
		\]
		Since the states corresponding to $v$ and $w_1$, as well as those corresponding to $v$ and $w_2$, are orthogonal on the $q$-th subsystem, their orthogonality is preserved under the measurement on the $p$-th subsystem. It therefore suffices to verify that the postmeasurement states corresponding to $w_1$ and $w_2$ remain orthogonal. By the above relations, the inner product between the postmeasurement states corresponding to $w_1$ and $w_2$ is $\langle\psi_{w_1}|(M_1^\dagger M_1\otimes I_{\bar p})|\psi_{w_2}\rangle=\langle\phi_{p,w_1}|\phi_{p,v}\rangle\langle\phi_{p,v}|\phi_{p,w_2}\rangle\langle\eta_{w_1}|\eta_{w_2}\rangle_{\bar p}$. Since the edge between $w_1$ and $w_2$ is not of color $p$, the states corresponding to $w_1$ and $w_2$ are orthogonal on one of the subsystems other than the $p$-th subsystem, and hence $\langle\eta_{w_1}|\eta_{w_2}\rangle_{\bar p}=0$. Therefore, the above inner product vanishes, and the orthogonality between the states corresponding to $w_1$ and $w_2$ is preserved.
		
		If the measurement outcome corresponds to $M_2$, then the measured state must correspond to $u_1$, $u_2$, $w_1$, or $w_2$. For $i=1,2$, since $\langle\phi_{p,v}|\phi_{p,u_i}\rangle=0$, we have $M_2|\phi_{p,u_i}\rangle=(I-|\phi_{p,v}\rangle\langle\phi_{p,v}|)|\phi_{p,u_i}\rangle=|\phi_{p,u_i}\rangle$. Moreover, for $i=1,2$ and any $x\in\{u_1,u_2,w_1,w_2\}$, we have $\langle\phi_{p,u_i}|M_2^\dagger M_2|\phi_{p,x}\rangle=\langle\phi_{p,u_i}|\phi_{p,x}\rangle-\langle\phi_{p,u_i}|\phi_{p,v}\rangle\langle\phi_{p,v}|\phi_{p,x}\rangle=\langle\phi_{p,u_i}|\phi_{p,x}\rangle$. Hence all orthogonality relations involving $u_1$ or $u_2$ are preserved.
		
		It remains to verify that the postmeasurement states corresponding to $w_1$ and $w_2$ remain orthogonal. For the measurement outcome corresponding to $M_2$, their postmeasurement inner product is $\langle\psi_{w_1}|(M_2^\dagger M_2\otimes I_{\bar p})|\psi_{w_2}\rangle=\bigl(\langle\phi_{p,w_1}|\phi_{p,w_2}\rangle-\langle\phi_{p,w_1}|\phi_{p,v}\rangle\langle\phi_{p,v}|\phi_{p,w_2}\rangle\bigr)\langle\eta_{w_1}|\eta_{w_2}\rangle_{\bar p}$. Since the edge between $w_1$ and $w_2$ is not of color $p$, the states corresponding to $w_1$ and $w_2$ are orthogonal on one of the subsystems other than the $p$-th subsystem, which implies $\langle\eta_{w_1}|\eta_{w_2}\rangle_{\bar p}=0$. Therefore, their postmeasurement inner product vanishes, and the orthogonality between the states corresponding to $w_1$ and $w_2$ is preserved.
		
		Therefore, under the condition that the edge between $w_1$ and $w_2$ is not of color $p$, the above local measurement is orthogonality-preserving. If the measurement outcome corresponds to $M_1$, the measured state must be one of the states corresponding to the vertices $\{v,w_1,w_2\}$, which can be perfectly distinguished by LOCC according to Lemma~\ref{lem:three-ops}. If the measurement outcome corresponds to $M_2$, the state corresponding to $v$ is excluded. Hence, the state corresponding to $v$ can be identified by LOCC. Moreover, if the states corresponding to the vertices $\{u_1,u_2,w_1,w_2\}$ can also be perfectly distinguished by LOCC, then $S$ can be perfectly distinguished by LOCC.	\end{proof}
	
	\begin{theorem}
		\label{thm:211-reduction}
		Let $S$ be a set of five OPSs, where any two of which are orthogonal on exactly one subsystem. Suppose that there exists a vertex $v$ in the orthogonality graph of $S$ and three distinct subsystems $p$, $q$, and $r$ such that $\deg_p(v)=2$ and $\deg_q(v)=\deg_r(v)=1$. Let $N_p(v)=\{u_1,u_2\}$, $N_q(v)=\{w_1\}$, and $N_r(v)=\{w_2\}$. If the edge between $w_1$ and $w_2$ is not of color $p$, then the $p$-th party can perform an orthogonality-preserving local measurement on the $p$-th subsystem such that the state corresponding to $v$ can be identified by LOCC.
	\end{theorem}
	
	Here, we omit the proof of Theorem 2 since it is similar to that of Theorem 1.
	
	For a colored orthogonality graph $G=(V,E_1\cup E_2\cup E_3\cup E_4)$ with $V=\{1,2,3,4,5\}$, let $e_{ij}$ denote the undirected edge joining vertices $i$ and $j$, where $e_{ij}=e_{ji}$ for any distinct $i,j\in V$. The relation $e_{ij}\in E_k$ indicates that the edge between vertices $i$ and $j$ has color $k$. Two colored orthogonality graphs are regarded as identical if one can be obtained from the other by a relabeling of vertices and a permutation of parties. 
	\begin{theorem}
		\label{thm:main}
		Five quadripartite OPSs, any two of which are orthogonal on exactly one subsystem, can be perfectly distinguished by LOCC unless the orthogonality graph of the five states is of one of the following four structural types (up to relabeling of vertices and permutation of parties): 
		\begin{enumerate}[label=(\arabic*)]
			
			\item The orthogonality graph of five quadripartite OPSs satisfies
			\[
			\begin{aligned}
				E_1&=\{e_{12},e_{13},e_{23},e_{45}\},\\
				E_2&=\{e_{14},e_{35}\},\\
				E_3&=\{e_{15},e_{24}\},\\
				E_4&=\{e_{25},e_{34}\}.
			\end{aligned}
			\]
			
			\item The orthogonality graph of five quadripartite OPSs satisfies
			\[
			\begin{aligned}
				E_1&=\{e_{12},e_{13},e_{45}\},\\
				E_2&=\{e_{14},e_{15},e_{23}\},\\
				E_3&=\{e_{24},e_{35}\},\\
				E_4&=\{e_{25},e_{34}\}.
			\end{aligned}
			\]
			
			\item The orthogonality graph of five quadripartite OPSs satisfies
			\[
			\begin{aligned}
				E_1&=\{e_{12},e_{13},e_{45}\},\\
				E_2&=\{e_{14},e_{23},e_{25}\},\\
				E_3&=\{e_{15},e_{34}\},\\
				E_4&=\{e_{24},e_{35}\}.
			\end{aligned}
			\]
			
			\item There exist four states, an ordering of these four states, and three distinct subsystems $p$, $q$, and $r$ such that the vertices $v_j$ of the orthogonality graph of these four OPSs satisfy
			$\deg_p(v_j)=\deg_q(v_j)=\deg_r(v_j)=1$ for $j=1,2,3,4$.
		\end{enumerate}
	
	\end{theorem}
	
	By analyzing all feasible orthogonality graphs in Figs.~\ref{fig:7111}--\ref{fig:3322-part2}, every non-exceptional graph contains, up to relabeling of vertices and permutation of parties, one of the reducible structures described in Lemma~\ref{lem:degree4-removal}, Lemma~\ref{lem:31-removal}, Theorem~\ref{thm:22-reduction}, and Theorem~\ref{thm:211-reduction}. Therefore, a set of five quadripartite OPSs corresponding to one of 355 graphs can be perfectly distinguished by LOCC. The proof of Theorem~\ref{thm:main} is given in Appendix~\ref{app:A}.
	\section{\label{sec4}Local indistinguishability of the exceptional cases}
	\label{sec:exceptional}
    In this section, we investigate the local distinguishability of five quadripartite OPSs associated with each of the 16 exceptional orthogonality graphs excluded in Theorem~\ref{thm:main}.
	
    \subsection*{1. A general form for case $(8\text{-}25)$}
    
    Consider five quadripartite OPSs corresponding to case $(8\text{-}25)$, which belongs to exceptional case $1$ in Theorem~\ref{thm:main}. We assume that the five OPSs are denoted by $\ket{\psi_1}$, $\ket{\psi_2}$, $\ket{\psi_3}$, $\ket{\psi_4}$, and $\ket{\psi_5}$, and $\ket{\psi_i} $ can be written as $\ket{\alpha_i}_1\ket{\beta_i}_2 \ket{\gamma_i}_3\ket{\delta_i}_4$, where $i=1,\ldots,5$.
    
    On the first subsystem, states $1$, $2$, and $3$ are mutually orthogonal, and states $4$ and $5$ are orthogonal to each other. Suppose that the first subsystems of states $1$, $2$, and $3$ are $\ket{0}$, $\ket{1}$, and $\ket{2}$, respectively.
    
    Since state $4$ is not orthogonal to any of states $1$, $2$, and $3$ on the first subsystem, its components along $\ket{0}$, $\ket{1}$, and $\ket{2}$ must be nonzero. Thus, the first subsystem of state $4$ can be expressed as $\frac{1}{\sqrt{1+|a_1|^2+|b_1|^2+|c_1|^2}} (\ket{0}+a_1\ket{1}+b_1\ket{2}+c_1\ket{3}),$ where $a_1\neq0$ and $b_1\neq0$. Similarly, state $5$ is not orthogonal to states $1$, $2$, and $3$ on the first subsystem. Hence, its coefficients along $\ket{0}$, $\ket{1}$, and $\ket{2}$ are also nonzero, and its first subsystem can be written as $\frac{1}{\sqrt{1+|d_1|^2+|e_1|^2+|f_1|^2+|g_1|^2}}
    (\ket{0}+d_1\ket{1}+e_1\ket{2} +f_1\ket{3}+g_1\ket{4})$, where $d_1\neq0$ and $e_1\neq0$.
    
    Furthermore, we have $1+a_1^*d_1+b_1^*e_1+c_1^*f_1=0$ since states $4$ and $5$ are orthogonal on the first subsystem. Consequently, the general forms of the first subsystems of states $1$, $2$, $3$, $4$, and $5$ are given in Eq.~\eqref{eq:case825-first-subsystem}:
    \begin{equation}
    	\label{eq:case825-first-subsystem}
    	\begin{aligned}
    		\ket{\alpha_1}_1&=\ket{0},\\
    		\ket{\alpha_2}_1&=\ket{1},\\
    		\ket{\alpha_3}_1&=\ket{2},\\
      		\ket{\alpha_4}_1
    		&=
    		\frac{1}
    		{\sqrt{1+|a_1|^2+|b_1|^2+|c_1|^2}}
    		\\[-0.5ex]
    		&\quad\times
    		\bigl(
    		\ket{0}+a_1\ket{1}+b_1\ket{2}+c_1\ket{3}
    		\bigr),\\
    \end{aligned}
    \end{equation}
    \begin{equation}
    	\nonumber
    	\begin{aligned}
    		\ket{\alpha_5}_1
    		&=
    		\frac{1}
    		{\sqrt{1+|d_1|^2+|e_1|^2+|f_1|^2+|g_1|^2}}
    		\\[-0.5ex]
    		&\quad\times
    		\bigl(
    		\ket{0}+d_1\ket{1}+e_1\ket{2}
    		+f_1\ket{3}+g_1\ket{4}
    		\bigr),
    	\end{aligned}
    \end{equation}
    where $a_1\neq0$, $b_1\neq0$, $d_1\neq0$, $e_1\neq0$, and $1+a_1^*d_1+b_1^*e_1+c_1^*f_1=0$.
    
    On the second subsystem, states $1$ and $4$ are orthogonal, and states $3$ and $5$ are orthogonal. Suppose that the second subsystems of states $1$ and $4$ are $\ket{0}$ and $\ket{1}$, respectively.
    
    Since state $3$ is not orthogonal to states $1$ and $4$ on this subsystem, its coefficients along $\ket{0}$ and $\ket{1}$ must be nonzero. Thus, the second subsystem of state $3$ can be written as $\frac{1}{\sqrt{1+|a_2|^2+|b_2|^2}} (\ket{0}+a_2\ket{1}+b_2\ket{2})$, where $a_2\neq0$.
    
    Similarly, state $5$ is not orthogonal to states $1$ and $4$ on the second subsystem. Hence, its second subsystem can be expressed as $\frac{1}{\sqrt{1+|c_2|^2+|d_2|^2+|e_2|^2}}
    (\ket{0}+c_2\ket{1}+d_2\ket{2}+e_2\ket{3})$, where $c_2\neq0$. Since states $3$ and $5$ are orthogonal on this subsystem, the condition $\braket{\beta_3|\beta_5}=0$ gives $1+a_2^*c_2+b_2^*d_2=0$.
    
    State $2$ is not orthogonal to states $1$, $3$, $4$, or $5$ on the second subsystem. Therefore, its second subsystem can be written as $\frac{1}{\sqrt{1+|f_2|^2+|g_2|^2+|h_2|^2+|k_2|^2}}
    (\ket{0}+f_2\ket{1}+g_2\ket{2}+h_2\ket{3}+k_2\ket{4})$, where $f_2\neq0$. Consequently, the general forms of the second subsystems are given in Eq.~\eqref{eq:case825-second-subsystem}:
    \begin{equation}
    	\label{eq:case825-second-subsystem}
    	\begin{aligned}
    		\ket{\beta_1}_2&=\ket{0},\\
    		\ket{\beta_2}_2
    		&=
    		\frac{1}
    		{\sqrt{1+|f_2|^2+|g_2|^2+|h_2|^2+|k_2|^2}}
    		\\[-0.5ex]
    		&\quad\times
    		\bigl(
    		\ket{0}+f_2\ket{1}+g_2\ket{2}
    		\\[-0.5ex]
    		&\qquad
    		+h_2\ket{3}+k_2\ket{4}
    		\bigr),\\
    		\ket{\beta_3}_2
    		&=
    		\frac{1}
    		{\sqrt{1+|a_2|^2+|b_2|^2}}
    		\\[-0.5ex]
    		&\quad\times
    		\bigl(
    		\ket{0}+a_2\ket{1}+b_2\ket{2}
    		\bigr),\\
    		\ket{\beta_4}_2&=\ket{1},\\
    		\ket{\beta_5}_2
    		&=
    		\frac{1}
    		{\sqrt{1+|c_2|^2+|d_2|^2+|e_2|^2}}
    		\\[-0.5ex]
    		&\quad\times
    		\bigl(
    		\ket{0}+c_2\ket{1}+d_2\ket{2}
    		+e_2\ket{3}
    		\bigr),
    	\end{aligned}
    \end{equation}
    where $a_2\neq0$, $c_2\neq0$, $f_2\neq0$,
    $1+a_2^*c_2+b_2^*d_2=0$,
    $1+a_2^*f_2+b_2^*g_2\neq0$, and
    $1+c_2^*f_2+d_2^*g_2+e_2^*h_2\neq0$.
    
    The constructions on the third and fourth subsystems follow the same procedure as that on the second subsystem. On the third subsystem, states $1$ and $5$ are orthogonal, while states $2$ and $4$ are orthogonal. Hence, the corresponding third subsystems can be written as shown in Eq.~\eqref{eq:case825-third-subsystem}:
    \begin{equation}
    \nonumber
    	\begin{aligned}
    		\ket{\gamma_1}_3&=\ket{0},\\
    		\ket{\gamma_2}_3
    		&=
    		\frac{1}
    		{\sqrt{1+|a_3|^2+|b_3|^2}}
    		\times
    		\bigl(
    		\ket{0}+a_3\ket{1}+b_3\ket{2}
    		\bigr),\\
    \end{aligned}
    	\end{equation}
       \begin{equation}
    		\label{eq:case825-third-subsystem}
    		\begin{aligned}
      		\ket{\gamma_3}_3
    		&=
    		\frac{1}
    		{\sqrt{1+|f_3|^2+|g_3|^2+|h_3|^2+|k_3|^2}}
    		\\[-0.5ex]
    		&\quad\times
    		\bigl(
    		\ket{0}+f_3\ket{1}+g_3\ket{2}
    		\\[-0.5ex]
    		&\qquad
    		+h_3\ket{3}+k_3\ket{4}
    		\bigr),\\
     		\ket{\gamma_4}_3
    		&=
    		\frac{1}
    		{\sqrt{1+|c_3|^2+|d_3|^2+|e_3|^2}}
    		\\[-0.5ex]
    		&\quad\times
    		\bigl(
    		\ket{0}+c_3\ket{1}+d_3\ket{2}
    		+e_3\ket{3}
    		\bigr),\\
    		\ket{\gamma_5}_3&=\ket{1},
    	\end{aligned}
    \end{equation}
    where $a_3\neq0$, $c_3\neq0$, $f_3\neq0$,
    $1+a_3^*c_3+b_3^*d_3=0$,
    $1+a_3^*f_3+b_3^*g_3\neq0$, and
    $1+c_3^*f_3+d_3^*g_3+e_3^*h_3\neq0$.
    
    On the fourth subsystem, states $2$ and $5$ are orthogonal, and states $3$ and $4$ are orthogonal. The general forms of their fourth subsystems are given in Eq.~\eqref{eq:case825-fourth-subsystem}:
    \begin{equation}
    	\label{eq:case825-fourth-subsystem}
    	\begin{aligned}
    		\ket{\delta_1}_4
    		&=
    		\frac{1}
    		{\sqrt{1+|f_4|^2+|g_4|^2+|h_4|^2+|k_4|^2}}
    		\\[-0.5ex]
    		&\quad\times
    		\bigl(
    		\ket{0}+f_4\ket{1}+g_4\ket{2}
    		\\[-0.5ex]
    		&\qquad
    		+h_4\ket{3}+k_4\ket{4}
    		\bigr),\\
    		\ket{\delta_2}_4&=\ket{0},\\
    		\ket{\delta_3}_4
    		&=
    		\frac{1}
    		{\sqrt{1+|a_4|^2+|b_4|^2}}
    		\\[-0.5ex]
    		&\quad\times
    		\bigl(
    		\ket{0}+a_4\ket{1}+b_4\ket{2}
    		\bigr),\\
    		\ket{\delta_4}_4
    		&=
    		\frac{1}
    		{\sqrt{1+|c_4|^2+|d_4|^2+|e_4|^2}}
    		\\[-0.5ex]
    		&\quad\times
    		\bigl(
    		\ket{0}+c_4\ket{1}+d_4\ket{2}
    		+e_4\ket{3}
    		\bigr),\\
    		\ket{\delta_5}_4&=\ket{1},
    	\end{aligned}
    \end{equation}
    where $a_4\neq0$, $c_4\neq0$, $f_4\neq0$, $1+a_4^*c_4+b_4^*d_4=0$, $1+a_4^*f_4+b_4^*g_4\neq0$, and $1+c_4^*f_4+d_4^*g_4+e_4^*h_4\neq0$.

    \begin{theorem}
    	\label{thm:case825-indistinguishable}
    	Five quadripartite OPSs corresponding to case $(8\text{-}25)$, with the general forms given in Eqs.~\eqref{eq:case825-first-subsystem}--\eqref{eq:case825-fourth-subsystem}, cannot be perfectly distinguished by LOCC when $c_1=f_1=g_1=0$, $b_s=d_s=e_s=g_s=h_s=k_s=0$ for $s=2,3,4$, and $\operatorname{Im}(a_1^*d_1)\neq0$, where $\operatorname{Im}(z)$ denotes the imaginary part of a complex number $z$.
    \end{theorem}
    
    \begin{proof}
    	When $c_1=f_1=g_1=0$ and $b_s=d_s=e_s=g_s=h_s=k_s=0$ for $s=2,3,4$, the five states associated with case $(8\text{-}25)$ reduce to the forms shown in Eq.~\eqref{eq:case825-reduced-states}:
    	\begin{equation}
    		\nonumber
    		\begin{aligned}
    			\ket{\psi_1}
    			&=
    			\frac{1}{\sqrt{1+|f_4|^2}}
    			\ket{0}_1\ket{0}_2\ket{0}_3
    			\bigl(\ket{0}+f_4\ket{1}\bigr)_4,
    			\\
    	    	\ket{\psi_2}
    			&=
    			\frac{1}
    			{\sqrt{1+|f_2|^2}\sqrt{1+|a_3|^2}}
    			\ket{1}_1
    			\bigl(\ket{0}+f_2\ket{1}\bigr)_2
    			\\[-0.5ex]
    			&\quad\times
    			\bigl(\ket{0}+a_3\ket{1}\bigr)_3
    			\ket{0}_4,
    			\\
       			\ket{\psi_3}
    			&=
    			\frac{1}
    			{\sqrt{1+|a_2|^2}
    				\sqrt{1+|f_3|^2}
    				\sqrt{1+|a_4|^2}}
    			\ket{2}_1
    			\\[-0.5ex]
    			&\quad\times
    			\bigl(\ket{0}+a_2\ket{1}\bigr)_2
    			\bigl(\ket{0}+f_3\ket{1}\bigr)_3
    			\bigl(\ket{0}+a_4\ket{1}\bigr)_4,
    			\\
    \end{aligned}
    	\end{equation}
       \begin{equation}
    		\label{eq:case825-reduced-states}
    		\begin{aligned}
    			\ket{\psi_4}
    			&=
    			\frac{1}
    			{\sqrt{1+|a_1|^2+|b_1|^2}
    				\sqrt{1+|c_3|^2}
    				\sqrt{1+|c_4|^2}}
    			\\[-0.5ex]
    			&\quad\times
    			\bigl(\ket{0}+a_1\ket{1}+b_1\ket{2}\bigr)_1
    			\ket{1}_2
    			\bigl(\ket{0}+c_3\ket{1}\bigr)_3
    			\\[-0.5ex]
    			&\quad\times
    			\bigl(\ket{0}+c_4\ket{1}\bigr)_4,
    			\\
    			\ket{\psi_5}
    			&=
    			\frac{1}
    			{\sqrt{1+|d_1|^2+|e_1|^2}
    				\sqrt{1+|c_2|^2}}
    			\\[-0.5ex]
    			&\quad\times
    			\bigl(\ket{0}+d_1\ket{1}+e_1\ket{2}\bigr)_1
    			\bigl(\ket{0}+c_2\ket{1}\bigr)_2
    			\ket{1}_3\ket{1}_4 ,
    		\end{aligned}
    	\end{equation}
    	where $1+a_1^*d_1+b_1^*e_1=0$ and $1+a_s^*c_s=0$ for $s=2,3,4$. Under these conditions, the local ranks are $r_1=3$ and $r_2=r_3=r_4=2$.
    	
    	Suppose that the first party, say Alice, first performs an orthogonality-preserving measurement with POVM elements $M_j^\dagger M_j$, where
    	\[
    	M_j^\dagger M_j=
    	\begin{pmatrix}
    		m_{00}^{(j)} & m_{01}^{(j)} & m_{02}^{(j)}\\
    		m_{10}^{(j)} & m_{11}^{(j)} & m_{12}^{(j)}\\
    		m_{20}^{(j)} & m_{21}^{(j)} & m_{22}^{(j)}
    	\end{pmatrix}.
    	\]
    	Here the matrix is written under the basis $\{|0\rangle,|1\rangle,|2\rangle\}$ of the local support, and $j=1,2,\ldots,l$. To preserve the orthogonality of the postmeasurement states, every pair that is orthogonal only on Alice's subsystem must remain orthogonal after the measurement.
    	
    	For states $|\psi_1\rangle$ and $|\psi_2\rangle$, orthogonality preservation gives $\langle\psi_1|(M_j^\dagger M_j\otimes I_{234})|\psi_2\rangle=0$ and $\langle\psi_2|(M_j^\dagger M_j\otimes I_{234})|\psi_1\rangle=0$. Thus, $m_{01}^{(j)}=m_{10}^{(j)}=0$.
    	
    	For states $|\psi_1\rangle$ and $|\psi_3\rangle$, the same requirement gives $m_{02}^{(j)}=m_{20}^{(j)}=0$. For states $|\psi_2\rangle$ and $|\psi_3\rangle$, we similarly obtain $m_{12}^{(j)}=m_{21}^{(j)}=0$.
    	
    	It remains to consider states $|\psi_4\rangle$ and $|\psi_5\rangle$. Since their first subsystems must remain orthogonal after Alice's measurement, we have $m_{00}^{(j)}+a_1^*d_1m_{11}^{(j)}+b_1^*e_1m_{22}^{(j)}=0$. Their original orthogonality relation is $1+a_1^*d_1+b_1^*e_1=0$. Subtracting $m_{22}^{(j)}(1+a_1^*d_1+b_1^*e_1)=0$ from the preceding equality gives $(m_{00}^{(j)}-m_{22}^{(j)})+a_1^*d_1(m_{11}^{(j)}-m_{22}^{(j)})=0$.
    	
    	Since $m_{00}^{(j)}$, $m_{11}^{(j)}$, and $m_{22}^{(j)}$ are real and $\operatorname{Im}(a_1^*d_1)\neq 0$, the imaginary part of this equality yields $m_{11}^{(j)}=m_{22}^{(j)}$. Substituting this relation back gives $m_{00}^{(j)}=m_{22}^{(j)}$. Consequently, every positive operator-valued measure (POVM) element performed by Alice must have the form
    	\[
    	M_j^\dagger M_j=
    	\begin{pmatrix}
    		m_{00}^{(j)} & 0 & 0\\
    		0 & m_{00}^{(j)} & 0\\
    		0 & 0 & m_{00}^{(j)}
    	\end{pmatrix}.
    	\]
    	
    	Therefore, if Alice obtains measurement outcome $j$, its probability is $p_A(j\mid k)=\langle\psi_k|(M_j^\dagger M_j\otimes I_{234})|\psi_k\rangle=m_{00}^{(j)}$ for every $k\in\{1,2,3,4,5\}$. Thus, Alice obtains the same outcome distribution for all five possible input states and cannot acquire any useful information.

    	Suppose that the second party, say Bob, performs the first orthogonality-preserving measurement. Since the local rank on the second subsystem is $r_2=2$, Bob's POVM element on the local support can be written as
    	\[
    	N_j^\dagger N_j=
    	\begin{pmatrix}
    		n_{00}^{(j)} & n_{01}^{(j)}\\
    		n_{10}^{(j)} & n_{11}^{(j)}
    	\end{pmatrix}.
    	\]Here the matrix is written in the basis $\{|0\rangle,|1\rangle\}$ of Bob's local support.
    	
    	Since states 1 and 4 are orthogonal on the second subsystem, we obtain $n_{01}^{(j)}=n_{10}^{(j)}=0$. The orthogonality of states 3 and 5 after the measurement gives $n_{00}^{(j)}+a_2^*c_2n_{11}^{(j)}=0$. Since $1+a_2^*c_2=0$, it follows that $n_{00}^{(j)}=n_{11}^{(j)}$. Hence, Bob's POVM element is proportional to the identity operator on the local support and cannot reveal any information about the measured state.
    	
    	For the third party, the orthogonality of states 1 and 5 on the third subsystem requires the off-diagonal entries of each orthogonality-preserving POVM element to vanish. For states 2 and 4, preservation of their orthogonality gives a relation between the two diagonal entries. Together with $1+a_3^*c_3=0$, this relation implies that the two diagonal entries are equal. Hence, every such POVM element is proportional to the identity operator on the relevant local subspace, and the third party cannot obtain any information for distinguishing the five states.
    	
    	For the fourth party, the orthogonality of states 2 and 5 on the fourth subsystem similarly implies that the off-diagonal entries vanish. The orthogonality-preserving condition for states 3 and 4, together with $1+a_4^*c_4=0$, further gives equality of the two diagonal entries. Therefore, every orthogonality-preserving POVM element of the fourth party is proportional to the identity operator on the relevant local subspace and provides no information about the state being measured.
    	
    	Consequently, none of the four parties can obtain useful information from an orthogonality-preserving first local measurement. Hence, no nontrivial first measurement can be performed while preserving the orthogonality of the five states, and the five OPSs corresponding to case (8-25) cannot be perfectly distinguished by LOCC.
    \end{proof}
    
    \subsection*{2. A general form for case $(10\text{-}27)$}
    
    Consider five quadripartite OPSs corresponding to case $(10\text{-}27)$, which belongs to exceptional case $2$ in Theorem~\ref{thm:main}. We assume that the five OPSs are denoted by $\ket{\psi_1}$, $\ket{\psi_2}$, $\ket{\psi_3}$, $\ket{\psi_4}$, and $\ket{\psi_5}$, and $\ket{\psi_i} $ can be written as $\ket{\alpha_i}_1\ket{\beta_i}_2 \ket{\gamma_i}_3\ket{\delta_i}_4$, where $i=1,\ldots,5$.
    
    On the first subsystem, state $1$ is orthogonal to states $2$ and $3$, while states $4$ and $5$ are orthogonal to each other. Suppose that the first subsystems of states $1$ and $2$ are $\ket{0}$ and $\ket{1}$, respectively.
    
    Since state $3$ is orthogonal to state $1$ but not to state $2$, its first subsystem can be expressed as $\frac{1}{\sqrt{1+|a_1|^2}} (\ket{1}+a_1\ket{2})$.
    
    State $4$ is not orthogonal to states $1$, $2$, or $3$ on the first subsystem. Hence, its first subsystem can be written as $\frac{1}{\sqrt{1+|b_1|^2+|c_1|^2+|d_1|^2}} (\ket{0}+b_1\ket{1}+c_1\ket{2}+d_1\ket{3})$, where $b_1\neq0$ and $b_1+a_1^*c_1\neq0$.
    
    Similarly, state $5$ is not orthogonal to states $1$, $2$, or $3$ on the first subsystem. Therefore, its first subsystem can be represented as $\frac{1}{\sqrt{1+|e_1|^2+|f_1|^2+|g_1|^2+|h_1|^2}}
    (\ket{0}+e_1\ket{1}+f_1\ket{2}+g_1\ket{3}+h_1\ket{4})$, where $e_1\neq0$ and $e_1+a_1^*f_1\neq0$.
    
    Furthermore, we have $1+b_1^{*}e_1+c_1^{*}f_1+d_1^{*}g_1=0$ since states $4$ and $5$ are orthogonal on the first subsystem. Consequently, the general forms of these first subsystems are given in Eq.~\ref{eq:case1027-first-subsystem}:
    \begin{equation}
    	\label{eq:case1027-first-subsystem}
    	\begin{aligned}
    		\ket{\alpha_1}_1&=\ket{0},\\
    		\ket{\alpha_2}_1&=\ket{1},\\
    		\ket{\alpha_3}_1
    		&=
    		\frac{1}{\sqrt{1+|a_1|^2}}
    		\times
    		\bigl(\ket{1}+a_1\ket{2}\bigr),\\
    		\ket{\alpha_4}_1
    		&=
    		\frac{1}
    		{\sqrt{1+|b_1|^2+|c_1|^2+|d_1|^2}}
    		\\[-0.5ex]
    		&\quad\times
    		\bigl(
    		\ket{0}+b_1\ket{1}+c_1\ket{2}+d_1\ket{3}
    		\bigr),\\
    		\ket{\alpha_5}_1
    		&=
    		\frac{1}
    		{\sqrt{1+|e_1|^2+|f_1|^2+|g_1|^2+|h_1|^2}}
    		\\[-0.5ex]
    		&\quad\times
    		\bigl(
    		\ket{0}+e_1\ket{1}+f_1\ket{2}
    		+g_1\ket{3}+h_1\ket{4}
    		\bigr),
    	\end{aligned}
    \end{equation}
    where $b_1\neq0$, $e_1\neq0$, $b_1+a_1^*c_1\neq0$, $e_1+a_1^*f_1\neq0$, and $1+b_1^*e_1+c_1^*f_1+d_1^*g_1=0$.
    
    On the second subsystem, state $1$ is orthogonal to states $4$ and $5$, and states $2$ and $3$ are orthogonal to each other. Suppose that the second subsystems of states $1$ and $4$ are $\ket{0}$ and $\ket{1}$, respectively.
    
    Since state $5$ is orthogonal to state $1$ but not to state $4$, its second subsystem can be expressed as $\frac{1}{\sqrt{1+|a_2|^2}} (\ket{1}+a_2\ket{2})$.
    
    State $2$ is not orthogonal to states $1$, $4$, or $5$ on the second subsystem. Thus, its second subsystem can be written as $\frac{1}{\sqrt{1+|b_2|^2+|c_2|^2+|d_2|^2}}
    (\ket{0}+b_2\ket{1}+c_2\ket{2}+d_2\ket{3})$, where $b_2\neq0$ and $b_2+a_2^*c_2\neq0$.
    
    Likewise, the second subsystem of state $3$ can be represented as $\frac{1}{\sqrt{1+|e_2|^2+|f_2|^2+|g_2|^2+|h_2|^2}} (\ket{0}+e_2\ket{1}+f_2\ket{2}+g_2\ket{3}+h_2\ket{4})$, where $e_2\neq0$ and $e_2+a_2^*f_2\neq0$.
    
    Since states $2$ and $3$ are orthogonal on this subsystem, $\braket{\beta_2|\beta_3}=0$ implies
    $1+b_2^*e_2+c_2^*f_2+d_2^*g_2=0$. Therefore, the general forms of the second subsystems are given in Eq.~\eqref{eq:case1027-second-subsystem}:
    \begin{equation}
    	\label{eq:case1027-second-subsystem}
    	\begin{aligned}
    		\ket{\beta_1}_2&=\ket{0},\\
    		\ket{\beta_2}_2
    		&=
    		\frac{1}
    		{\sqrt{1+|b_2|^2+|c_2|^2+|d_2|^2}}
    		\\[-0.5ex]
    		&\quad\times
    		\bigl(
    		\ket{0}+b_2\ket{1}+c_2\ket{2}+d_2\ket{3}
    		\bigr),\\
    		\ket{\beta_3}_2
    		&=
    		\frac{1}
    		{\sqrt{1+|e_2|^2+|f_2|^2+|g_2|^2+|h_2|^2}}
    		\\[-0.5ex]
    		&\quad\times
    		\bigl(
    		\ket{0}+e_2\ket{1}+f_2\ket{2}
    		+g_2\ket{3}+h_2\ket{4}
    		\bigr),\\
    		\ket{\beta_4}_2&=\ket{1},\\
    		\ket{\beta_5}_2
    		&=
    		\frac{1}{\sqrt{1+|a_2|^2}}
    		\times
    		\bigl(\ket{1}+a_2\ket{2}\bigr),
    	\end{aligned}
    \end{equation}
    where $b_2\neq0$, $e_2\neq0$, $b_2+a_2^*c_2\neq0$, $e_2+a_2^*f_2\neq0$, and $1+b_2^*e_2+c_2^*f_2+d_2^*g_2=0$.
    
    The constructions on the third and fourth subsystems follow the same procedure. On the third subsystem, states $2$ and $4$ are orthogonal, and states $3$ and $5$ are orthogonal. Hence, the corresponding subsystems are given in Eq.~\eqref{eq:case1027-third-subsystem}:
    \begin{equation}
    	\label{eq:case1027-third-subsystem}
    	\begin{aligned}
    		\ket{\gamma_1}_3
    		&=
    		\frac{1}
    		{\sqrt{1+|p_3|^2+|q_3|^2+|r_3|^2+|s_3|^2}}
    		\\[-0.5ex]
    		&\quad\times
    		\bigl(
    		\ket{0}+p_3\ket{1}+q_3\ket{2}
    		+r_3\ket{3}+s_3\ket{4}
    		\bigr),\\
    		\ket{\gamma_2}_3&=\ket{0},\\
    		\ket{\gamma_3}_3
    		&=
    		\frac{1}
    		{\sqrt{1+|a_3|^2+|b_3|^2}}
    		\\[-0.5ex]
    		&\quad\times
    		\bigl(
    		\ket{0}+a_3\ket{1}+b_3\ket{2}
    		\bigr),\\
    		\ket{\gamma_4}_3&=\ket{1},\\
    		\ket{\gamma_5}_3
    		&=
    		\frac{1}
    		{\sqrt{1+|c_3|^2+|d_3|^2+|e_3|^2}}
    		\\[-0.5ex]
    		&\quad\times
    		\bigl(
    		\ket{0}+c_3\ket{1}+d_3\ket{2}
    		+e_3\ket{3}
    		\bigr),
    	\end{aligned}
    \end{equation}
    where $a_3\neq0$, $c_3\neq0$, $p_3\neq0$,
    $1+a_3^*c_3+b_3^*d_3=0$,
    $1+a_3^*p_3+b_3^*q_3\neq0$, and
    $1+c_3^*p_3+d_3^*q_3+e_3^*r_3\neq0$.
    
    On the fourth subsystem, states $2$ and $5$ are orthogonal, and states $3$ and $4$ are orthogonal. Thus, the general forms of the fourth subsystems are given in Eq.~\eqref{eq:case1027-fourth-subsystem}:
    \begin{equation}
    	\label{eq:case1027-fourth-subsystem}
    	\begin{aligned}
    		\ket{\delta_1}_4
    		&=
    		\frac{1}
    		{\sqrt{1+|p_4|^2+|q_4|^2+|r_4|^2+|s_4|^2}}
    		\\[-0.5ex]
    		&\quad\times
    		\bigl(
    		\ket{0}+p_4\ket{1}+q_4\ket{2}
    		+r_4\ket{3}+s_4\ket{4}
    		\bigr),\\
    		\ket{\delta_2}_4&=\ket{0},\\
    		\ket{\delta_3}_4
    		&=
    		\frac{1}
    		{\sqrt{1+|a_4|^2+|b_4|^2}}
    		\\[-0.5ex]
    		&\quad\times
    		\bigl(
    		\ket{0}+a_4\ket{1}+b_4\ket{2}
    		\bigr),\\
    		\ket{\delta_4}_4
    		&=
    		\frac{1}
    		{\sqrt{1+|c_4|^2+|d_4|^2+|e_4|^2}}
    		\\[-0.5ex]
    		&\quad\times
    		\bigl(
    		\ket{0}+c_4\ket{1}+d_4\ket{2}
    		+e_4\ket{3}
    		\bigr),\\
    		\ket{\delta_5}_4&=\ket{1},
    	\end{aligned}
    \end{equation}
    where $a_4\neq0$, $c_4\neq0$, $p_4\neq0$,
    $1+a_4^*c_4+b_4^*d_4=0$,
    $1+a_4^*p_4+b_4^*q_4\neq0$, and
    $1+c_4^*p_4+d_4^*q_4+e_4^*r_4\neq0$.

    \begin{theorem}
    	\label{thm:case1027-indistinguishable}
    	Five quadripartite OPSs corresponding to case $(10\text{-}27)$, with the general forms given in Eqs.~\eqref{eq:case1027-first-subsystem}--%
    	\eqref{eq:case1027-fourth-subsystem}, cannot be perfectly distinguished by LOCC when
    	$a_u=c_u=d_u=f_u=g_u=h_u=0$ for $u=1,2$, and $b_v=d_v=e_v=q_v=r_v=s_v=0$ for $v=3,4$.
    \end{theorem}
    
    \begin{proof}
    	When $a_u=c_u=d_u=f_u=g_u=h_u=0$ for $u=1,2$, and $b_v=d_v=e_v=q_v=r_v=s_v=0$ for $v=3,4$, the five states associated with case $(10\text{-}27)$ reduce to the forms shown in Eq.~\eqref{eq:case1027-reduced-states}:
    	\begin{equation}
    		\label{eq:case1027-reduced-states}
    		\begin{aligned}
    			\ket{\psi_1}
    			&=
    			\frac{1}
    			{\sqrt{1+|p_3|^2}\sqrt{1+|p_4|^2}}
    			\ket{0}_1\ket{0}_2
    			\\[-0.5ex]
    			&\quad\times
    			\bigl(\ket{0}+p_3\ket{1}\bigr)_3
    			\bigl(\ket{0}+p_4\ket{1}\bigr)_4,
    			\\
    			\ket{\psi_2}
    			&=
    			\frac{1}{\sqrt{1+|b_2|^2}}
    			\ket{1}_1
    			\bigl(\ket{0}+b_2\ket{1}\bigr)_2
    			\ket{0}_3\ket{0}_4,
    			\\
    			\ket{\psi_3}
    			&=
    			\frac{1}
    			{\sqrt{1+|e_2|^2}
    				\sqrt{1+|a_3|^2}
    				\sqrt{1+|a_4|^2}}
    			\ket{1}_1
    			\\[-0.5ex]
    			&\quad\times
    			\bigl(\ket{0}+e_2\ket{1}\bigr)_2
    			\bigl(\ket{0}+a_3\ket{1}\bigr)_3
    			\bigl(\ket{0}+a_4\ket{1}\bigr)_4,
    			\\
    			\ket{\psi_4}
    			&=
    			\frac{1}
    			{\sqrt{1+|b_1|^2}\sqrt{1+|c_4|^2}}
    			\bigl(\ket{0}+b_1\ket{1}\bigr)_1
    			\ket{1}_2\ket{1}_3
    			\\[-0.5ex]
    			&\quad\times
    			\bigl(\ket{0}+c_4\ket{1}\bigr)_4,
    			\\
    			\ket{\psi_5}
    			&=
    			\frac{1}
    			{\sqrt{1+|e_1|^2}\sqrt{1+|c_3|^2}}
    			\bigl(\ket{0}+e_1\ket{1}\bigr)_1
    			\ket{1}_2
    			\\[-0.5ex]
    			&\quad\times
    			\bigl(\ket{0}+c_3\ket{1}\bigr)_3
    			\ket{1}_4 ,
    		\end{aligned}
    	\end{equation}
    	where $1+b_1^*e_1=0$, $1+b_2^*e_2=0$, and $1+a_v^*c_v=0$ for $v=3,4$. Under these parameter conditions, the local ranks satisfy $r_1=r_2=r_3=r_4=2$.

    	Suppose that the first party, say Alice, first performs an orthogonality-preserving measurement with POVM elements $M_j^\dagger M_j$, where
    	\[
    	M_j^\dagger M_j=
    	\begin{pmatrix}
    		m_{00}^{(j)} & m_{01}^{(j)}\\
    		m_{10}^{(j)} & m_{11}^{(j)}
    	\end{pmatrix}.
    	\]
    	Here the matrix is written under the basis $\{|0\rangle,|1\rangle\}$ of the local support, and $j=1,2,\ldots,l$. To preserve the orthogonality of the postmeasurement states, any two states whose orthogonality is provided by Alice's subsystem must remain orthogonal after Alice's measurement.
    	
    	For states $\ket{\psi_1}$ and $\ket{\psi_2}$, we have $\bra{\psi_1}(M_j^\dagger M_j\otimes I_{234})\ket{\psi_2}=0$ and $\bra{\psi_2}(M_j^\dagger M_j\otimes I_{234})\ket{\psi_1}=0$. Thus, $m_{01}^{(j)}=m_{10}^{(j)}=0$. The pair formed by states $1$ and $3$ gives the same conditions.
    	
    	For states $\ket{\psi_4}$ and $\ket{\psi_5}$, orthogonality preservation gives $m_{00}^{(j)}+b_1^*e_1m_{11}^{(j)}=0$. Since their original orthogonality relation is $1+b_1^*e_1=0$, we obtain $m_{00}^{(j)}=m_{11}^{(j)}$. Consequently, every POVM element performed by Alice must have the form
    	\[
    	M_j^\dagger M_j=
    	\begin{pmatrix}
    		m_{00}^{(j)} & 0\\
    		0 & m_{00}^{(j)}
    	\end{pmatrix}.
    	\]

    	If Alice obtains measurement outcome $j$, its probability is $p_A(j\mid k)= \bra{\psi_k}(M_j^\dagger M_j\otimes I_{234})\ket{\psi_k} =m_{00}^{(j)}$ for every $k\in\{1,2,3,4,5\}$. Hence, the probability of outcome $j$ is identical for all possible input states, and Alice cannot obtain useful information.

    	Suppose next that the second party, say Bob, performs the first orthogonality-preserving measurement. Bob's POVM element on the local support can be
    	written as
    	\[
    	N_j^\dagger N_j=
    	\begin{pmatrix}
    		n_{00}^{(j)} & n_{01}^{(j)}\\
    		n_{10}^{(j)} & n_{11}^{(j)}
    	\end{pmatrix}.
    	\]Here the matrix is written in the basis $\{|0\rangle,|1\rangle\}$ of Bob's local support.
    	
    	 Since states 1 and 4 are orthogonal on the second subsystem, orthogonality preservation gives $n_{01}^{(j)}=n_{10}^{(j)}=0$. The orthogonality between states $1$ and $5$ gives the same restriction. For states $2$ and $3$, preservation of their orthogonality gives $n_{00}^{(j)}+b_2^*e_2n_{11}^{(j)}=0$. Together with $1+b_2^*e_2=0$, this relation implies $n_{00}^{(j)}=n_{11}^{(j)}$. Hence, every such POVM element is proportional to the identity operator on the corresponding local subspace, and the second party cannot obtain any information for distinguishing the five states.
    	
    	For the third party, the orthogonality of states $2$ and $4$ forces the off-diagonal entries of each orthogonality-preserving POVM element to vanish, while that of states $3$ and $5$, together with $1+a_3^*c_3=0$, implies equality of the two diagonal entries. Hence, every such POVM element is proportional to the identity operator on the corresponding local subspace.
    	
    	Similarly, for the fourth party, the orthogonality relations between states $2$ and $5$ and between states $3$ and $4$, together with $1+a_4^*c_4=0$, imply that every orthogonality-preserving POVM element is proportional to the identity operator on the corresponding local subspace.
    	
    	Consequently, none of the four parties can obtain useful information from an orthogonality-preserving first local measurement. Hence, no nontrivial first local measurement can be performed while preserving the orthogonality of the five states. Therefore, the five quadripartite OPSs corresponding to case $(10\text{-}27)$ cannot be perfectly distinguished by LOCC.
    \end{proof}
    
    \subsection*{3. A general form for case $(10\text{-}43)$}
    
    Consider five quadripartite OPSs corresponding to case $(10\text{-}43)$, which belongs to exceptional case $3$ in Theorem~\ref{thm:main}. We assume that the five quadripartite OPSs are denoted as $\ket{\psi_1}$, $\ket{\psi_2}$, $\ket{\psi_3}$, $\ket{\psi_4}$, and $\ket{\psi_5}$, and $\ket{\psi_i}$ can be written as $\ket{\alpha_i}_1\ket{\beta_i}_2 \ket{\gamma_i}_3\ket{\delta_i}_4$, where $i=1,\ldots,5$.
    
    On the first subsystem, state $1$ is orthogonal to states $2$ and $3$, while states $4$ and $5$ are orthogonal to each other. Suppose that the first subsystems of states $1$ and $2$ are $\ket{0}$ and $\ket{1}$, respectively.
    
    Since state $3$ is orthogonal to state $1$ but is not orthogonal to state $2$, its first subsystem can be expressed as $\frac{1}{\sqrt{1+|a_1|^2}}(\ket{1}+a_1\ket{2})$.
    
    State $4$ is not orthogonal to states $1$, $2$, or $3$ on the first subsystem. Hence, its first subsystem can be written as $\frac{1}{\sqrt{1+|b_1|^2+|c_1|^2+|d_1|^2}}
    (\ket{0}+b_1\ket{1}+c_1\ket{2}+d_1\ket{3})$, where $b_1\neq0$ and $b_1+a_1^*c_1\neq0$.
    
    Similarly, state $5$ is not orthogonal to states $1$, $2$, or $3$ on this subsystem. Thus, its first subsystem can be represented as $\frac{1}{\sqrt{1+|e_1|^2+|f_1|^2+|g_1|^2+|h_1|^2}}
    (\ket{0}+e_1\ket{1}+f_1\ket{2}+g_1\ket{3}+h_1\ket{4})$, where $e_1\neq0$ and $e_1+a_1^*f_1\neq0$.
    
    Furthermore, we have $1+b_1^{*}e_1+c_1^{*}f_1+d_1^{*}g_1=0$ since states $4$ and $5$ are orthogonal on the first subsystem. Consequently, the general forms of these first subsystems are given in Eq.~(\ref{eq:case1043-first-subsystem}):
    \begin{equation}
    	\label{eq:case1043-first-subsystem}
    	\begin{aligned}
    		\ket{\alpha_1}_1&=\ket{0},\\
    		\ket{\alpha_2}_1&=\ket{1},\\
    		\ket{\alpha_3}_1
    		&=
    		\frac{1}{\sqrt{1+|a_1|^2}}
    		\times
    		\bigl(\ket{1}+a_1\ket{2}\bigr),\\
    		\ket{\alpha_4}_1
    		&=
    		\frac{1}
    		{\sqrt{1+|b_1|^2+|c_1|^2+|d_1|^2}}
    		\\[-0.5ex]
    		&\quad\times
    		\bigl(
    		\ket{0}+b_1\ket{1}+c_1\ket{2}+d_1\ket{3}
    		\bigr),\\
    		\ket{\alpha_5}_1
    		&=
    		\frac{1}
    		{\sqrt{1+|e_1|^2+|f_1|^2+|g_1|^2+|h_1|^2}}
    		\\[-0.5ex]
    		&\quad\times
    		\bigl(
    		\ket{0}+e_1\ket{1}+f_1\ket{2}
    		+g_1\ket{3}+h_1\ket{4}
    		\bigr),
    	\end{aligned}
    \end{equation}
    where $b_1\neq0$, $e_1\neq0$, $b_1+a_1^*c_1\neq0$, $e_1+a_1^*f_1\neq0$, and $1+b_1^*e_1+c_1^*f_1+d_1^*g_1=0$.
    
    On the second subsystem, states $3$ and $5$ are orthogonal, state $2$ is orthogonal to state $3$, and states $1$ and $4$ are orthogonal. Suppose that the second subsystems of states $3$ and $5$ are $\ket{0}$ and $\ket{1}$, respectively.
    
    Since state $2$ is orthogonal to state $3$ but is not orthogonal to state $5$, its second subsystem can be expressed as $\frac{1}{\sqrt{1+|a_2|^2}}(\ket{1}+a_2\ket{2})$.
    
    State $1$ is not orthogonal to states $2$, $3$, or $5$ on this subsystem. Hence, its second subsystem can be written as $\frac{1}{\sqrt{1+|b_2|^2+|c_2|^2+|d_2|^2}} (\ket{0}+b_2\ket{1}+c_2\ket{2}+d_2\ket{3})$, where $b_2\neq0$ and $b_2+a_2^*c_2\neq0$.
    
    Similarly, state $4$ is not orthogonal to states $2$, $3$, or $5$. Its second subsystem can therefore be represented as $\frac{1}{\sqrt{1+|e_2|^2+|f_2|^2+|g_2|^2+|h_2|^2}}
    (\ket{0}+e_2\ket{1}+f_2\ket{2}+g_2\ket{3}+h_2\ket{4})$, where $e_2\neq0$ and $e_2+a_2^*f_2\neq0$.
    
    Since states $1$ and $4$ are orthogonal on the second subsystem, $\braket{\beta_1|\beta_4}=0$ implies $1+b_2^*e_2+c_2^*f_2+d_2^*g_2=0$. Thus, the general forms of the second subsystems are shown in Eq.~\eqref{eq:case1043-second-subsystem}:
    \begin{equation}
    	\label{eq:case1043-second-subsystem}
    	\begin{aligned}
    		\ket{\beta_1}_2
    		&=
    		\frac{1}
    		{\sqrt{1+|b_2|^2+|c_2|^2+|d_2|^2}}
    		\\[-0.5ex]
    		&\quad\times
    		\bigl(
    		\ket{0}+b_2\ket{1}+c_2\ket{2}+d_2\ket{3}
    		\bigr),\\
    		\ket{\beta_2}_2
    		&=
    		\frac{1}{\sqrt{1+|a_2|^2}}
    		\times
    		\bigl(\ket{1}+a_2\ket{2}\bigr),\\
    		\ket{\beta_3}_2&=\ket{0},\\
    		\ket{\beta_4}_2
    		&=
    		\frac{1}
    		{\sqrt{1+|e_2|^2+|f_2|^2+|g_2|^2+|h_2|^2}}
    		\\[-0.5ex]
    		&\quad\times
    		\bigl(
    		\ket{0}+e_2\ket{1}+f_2\ket{2}
    		+g_2\ket{3}+h_2\ket{4}
    		\bigr),\\
    		\ket{\beta_5}_2&=\ket{1},
    	\end{aligned}
    \end{equation}
    where $b_2\neq0$, $e_2\neq0$, $b_2+a_2^*c_2\neq0$, $e_2+a_2^*f_2\neq0$, and $1+b_2^*e_2+c_2^*f_2+d_2^*g_2=0$.
    
    The constructions on the third and fourth subsystems follow the same procedure. On the third subsystem, states $1$ and $5$ are orthogonal, and states $2$ and $4$ are orthogonal. Hence, the corresponding third subsystems are shown in Eq.~\eqref{eq:case1043-third-subsystem}:
    \begin{equation}
    	\label{eq:case1043-third-subsystem}
    	\begin{aligned}
    		\ket{\gamma_1}_3&=\ket{0},\\
    		\ket{\gamma_2}_3
    		&=
    		\frac{1}
    		{\sqrt{1+|a_3|^2+|b_3|^2}}
    		\\[-0.5ex]
    		&\quad\times
    		\bigl(
    		\ket{0}+a_3\ket{1}+b_3\ket{2}
    		\bigr),\\
    		\ket{\gamma_3}_3
    		&=
    		\frac{1}
    		{\sqrt{1+|f_3|^2+|g_3|^2+|h_3|^2+|i_3|^2}}
    		\\[-0.5ex]
    		&\quad\times
    		\bigl(
    		\ket{0}+f_3\ket{1}+g_3\ket{2}
    		\\[-0.5ex]
    		&\qquad
    		+h_3\ket{3}+i_3\ket{4}
    		\bigr),\\
    		\ket{\gamma_4}_3
    		&=
    		\frac{1}
    		{\sqrt{1+|c_3|^2+|d_3|^2+|e_3|^2}}
    		\\[-0.5ex]
    		&\quad\times
    		\bigl(
    		\ket{0}+c_3\ket{1}+d_3\ket{2}
    		+e_3\ket{3}
    		\bigr),\\
    		\ket{\gamma_5}_3&=\ket{1},
    	\end{aligned}
    \end{equation}
    where $a_3\neq0$, $c_3\neq0$, $f_3\neq0$,
    $1+a_3^*c_3+b_3^*d_3=0$,
    $1+a_3^*f_3+b_3^*g_3\neq0$, and
    $1+c_3^*f_3+d_3^*g_3+e_3^*h_3\neq0$.
    
    On the fourth subsystem, states $2$ and $5$ are orthogonal, and states $3$ and $4$ are orthogonal. Accordingly, the general forms of the fourth subsystems are given in
    Eq.~\eqref{eq:case1043-fourth-subsystem}:
    \begin{equation}
    	\label{eq:case1043-fourth-subsystem}
    	\begin{aligned}
    		\ket{\delta_1}_4
    		&=
    		\frac{1}
    		{\sqrt{1+|f_4|^2+|g_4|^2+|h_4|^2+|i_4|^2}}
    		\\[-0.5ex]
    		&\quad\times
    		\bigl(
    		\ket{0}+f_4\ket{1}+g_4\ket{2}
    		\\[-0.5ex]
    		&\qquad
    		+h_4\ket{3}+i_4\ket{4}
    		\bigr),\\
    		\ket{\delta_2}_4&=\ket{0},\\
    		\ket{\delta_3}_4
    		&=
    		\frac{1}
    		{\sqrt{1+|a_4|^2+|b_4|^2}}
    		\\[-0.5ex]
    		&\quad\times
    		\bigl(
    		\ket{0}+a_4\ket{1}+b_4\ket{2}
    		\bigr),\\
    		\ket{\delta_4}_4
    		&=
    		\frac{1}
    		{\sqrt{1+|c_4|^2+|d_4|^2+|e_4|^2}}
    		\\[-0.5ex]
    		&\quad\times
    		\bigl(
    		\ket{0}+c_4\ket{1}+d_4\ket{2}
    		+e_4\ket{3}
    		\bigr),\\
    		\ket{\delta_5}_4&=\ket{1},
    	\end{aligned}
    \end{equation}
    where $a_4\neq0$, $c_4\neq0$, $f_4\neq0$,
    $1+a_4^*c_4+b_4^*d_4=0$,
    $1+a_4^*f_4+b_4^*g_4\neq0$, and
    $1+c_4^*f_4+d_4^*g_4+e_4^*h_4\neq0$.

    \begin{theorem}
    	\label{thm:case1043-indistinguishable}
    	Five quadripartite OPSs corresponding to case $(10\text{-}43)$, with the general forms given in Eqs.~\eqref{eq:case1043-first-subsystem}--%
    	\eqref{eq:case1043-fourth-subsystem}, cannot be perfectly distinguished by LOCC when
    	$a_u=c_u=d_u=f_u=g_u=h_u=0$ for $u=1,2$, and
    	$b_v=d_v=e_v=g_v=h_v=i_v=0$ for $v=3,4$.
    \end{theorem}
    
    \begin{proof}
    	When $a_u=c_u=d_u=f_u=g_u=h_u=0$ for $u=1,2$, and $b_v=d_v=e_v=g_v=h_v=i_v=0$ for $v=3,4$, the five states associated with case $(10\text{-}43)$ reduce to the forms shown in Eq.~\eqref{eq:case1043-reduced-states}:
    	\begin{equation}
    		\label{eq:case1043-reduced-states}
    		\begin{aligned}
    			\ket{\psi_1}
    			&=
    			\frac{1}
    			{\sqrt{1+|b_2|^2}\sqrt{1+|f_4|^2}}
    			\ket{0}_1
    			\bigl(\ket{0}+b_2\ket{1}\bigr)_2
    			\\[-0.5ex]
    			&\quad\times
    			\ket{0}_3
    			\bigl(\ket{0}+f_4\ket{1}\bigr)_4,
    			\\
    			\ket{\psi_2}
    			&=
    			\frac{1}{\sqrt{1+|a_3|^2}}
    			\ket{1}_1\ket{1}_2
    			\bigl(\ket{0}+a_3\ket{1}\bigr)_3
    			\ket{0}_4,
    			\\
    			\ket{\psi_3}
    			&=
    			\frac{1}
    			{\sqrt{1+|f_3|^2}\sqrt{1+|a_4|^2}}
    			\ket{1}_1\ket{0}_2
    			\\[-0.5ex]
    			&\quad\times
    			\bigl(\ket{0}+f_3\ket{1}\bigr)_3
    			\bigl(\ket{0}+a_4\ket{1}\bigr)_4,
    			\\
    			\ket{\psi_4}
    			&=
    			\frac{1}
    			{\sqrt{1+|b_1|^2}
    				\sqrt{1+|e_2|^2}
    				\sqrt{1+|c_3|^2}
    				\sqrt{1+|c_4|^2}}
    			\\[-0.5ex]
    			&\quad\times
    			\bigl(\ket{0}+b_1\ket{1}\bigr)_1
    			\bigl(\ket{0}+e_2\ket{1}\bigr)_2
    			\\[-0.5ex]
    			&\quad\times
    			\bigl(\ket{0}+c_3\ket{1}\bigr)_3
    			\bigl(\ket{0}+c_4\ket{1}\bigr)_4,
    			\\
    			\ket{\psi_5}
    			&=
    			\frac{1}{\sqrt{1+|e_1|^2}}
    			\bigl(\ket{0}+e_1\ket{1}\bigr)_1
    			\ket{1}_2\ket{1}_3\ket{1}_4 ,
    		\end{aligned}
    	\end{equation}
    	where $1+b_1^*e_1=0$, $1+b_2^*e_2=0$, and $1+a_v^*c_v=0$ for $v=3,4$. Under these conditions, the local ranks satisfy $r_1=r_2=r_3=r_4=2$.
    	
    	Suppose that the first party, say Alice, first performs an orthogonality-preserving measurement with POVM elements $M_j^\dagger M_j$, where
    	\[
    	M_j^\dagger M_j=
    	\begin{pmatrix}
    		m_{00}^{(j)} & m_{01}^{(j)}\\
    		m_{10}^{(j)} & m_{11}^{(j)}
    	\end{pmatrix}.
    	\]
    	Here the matrix is written in the basis $\{|0\rangle,|1\rangle\}$ of the local support, and $j=1,2,\ldots,l$. To preserve the orthogonality of the postmeasurement states, every pair whose orthogonality is supplied by the first subsystem must remain orthogonal after Alice's measurement.
    	
    	For states $\ket{\psi_1}$ and $\ket{\psi_2}$, we have $\bra{\psi_1}(M_j^\dagger M_j\otimes I_{234})\ket{\psi_2}=0$ and $\bra{\psi_2}(M_j^\dagger M_j\otimes I_{234})\ket{\psi_1}=0$. Thus, $m_{01}^{(j)}=m_{10}^{(j)}=0$. The pair formed by states $1$ and $3$ yields the same conditions.
    	
    	For states $\ket{\psi_4}$ and $\ket{\psi_5}$, orthogonality preservation gives $m_{00}^{(j)}+b_1^*e_1m_{11}^{(j)}=0$. Since their original orthogonality condition is $1+b_1^*e_1=0$, we obtain $m_{00}^{(j)}=m_{11}^{(j)}$. Consequently, every POVM element
    	performed by Alice must have the form
    	\[
    	M_j^\dagger M_j=
    	\begin{pmatrix}
    		m_{00}^{(j)} & 0\\
    		0 & m_{00}^{(j)}
    	\end{pmatrix}.
    	\]
    	
        When Alice obtains outcome $j$, its probability is $p_A(j\mid k)= \bra{\psi_k}(M_j^\dagger M_j\otimes I_{234})\ket{\psi_k}=m_{00}^{(j)}$ for every $k\in\{1,2,3,4,5\}$. Hence, Alice obtains the same outcome distribution for all five possible input states and cannot acquire useful information.
        
        Suppose next that the second party, say Bob, performs the first orthogonality-preserving
        measurement. Bob's POVM element on the local support can be written as
        \[
        N_j^\dagger N_j=
        \begin{pmatrix}
        	n_{00}^{(j)} & n_{01}^{(j)}\\
        	n_{10}^{(j)} & n_{11}^{(j)}
        \end{pmatrix}.
        \]
        Here the matrix is written in the basis $\{|0\rangle,|1\rangle\}$ of Bob's local support.
        
        Since states 3 and 5 are orthogonal on the second subsystem, orthogonality preservation
        gives $n_{01}^{(j)}=n_{10}^{(j)}=0$. The orthogonality between states $2$ and $3$ gives the same restriction. For states $1$ and $4$, preservation of their orthogonality gives $n_{00}^{(j)}+b_2^*e_2n_{11}^{(j)}=0$. Together with $1+b_2^*e_2=0$, this relation implies $n_{00}^{(j)}=n_{11}^{(j)}$. Hence, every such POVM element is proportional to the identity operator on the corresponding local subspace, and the second party cannot obtain any information for distinguishing the five states.
        
        For the third party, the orthogonality of states $1$ and $5$ forces the off-diagonal entries of each orthogonality-preserving POVM element to vanish, while that of states $2$ and $4$, together with $1+a_3^*c_3=0$, implies equality of the two diagonal entries. Hence, every such POVM element is proportional to the identity operator on the corresponding local subspace.
        
        Similarly, for the fourth party, the orthogonality relations between states $2$ and $5$ and between states $3$ and $4$, together with $1+a_4^*c_4=0$, imply that every orthogonality-preserving POVM element is proportional to the identity operator on the corresponding local subspace.
        
        Consequently, none of the four parties can obtain useful information from an orthogonality-preserving first local measurement. Hence, no nontrivial first local measurement can be performed while preserving the orthogonality of the five states. Therefore, the five quadripartite OPSs corresponding to case $(10\text{-}43)$ cannot be perfectly distinguished by LOCC.
    \end{proof}
    
   For exceptional case $4$ in Theorem~\ref{thm:main}, the corresponding orthogonality graphs of ﬁve quadripartite OPSs are cases $(5\text{-}37)$, $(7\text{-}79)$, $(7\text{-}85)$, $(7\text{-}94)$, $(8\text{-}3)$, $(8\text{-}23)$, $(8\text{-}39)$, $(9\text{-}30)$, $(9\text{-}31)$, $(9\text{-}32)$, $(10\text{-}8)$, $(10\text{-}41)$, and $(10\text{-}59)$, as shown in Fig.~\ref{fig:first-exceptional-type}. We discuss separately the local indistinguishability of five quadripartite OPSs corresponding to each graph.

   \begin{figure}[H]
   	\centering
   	\includegraphics[
   	width=0.8\columnwidth,
   	keepaspectratio
   	]{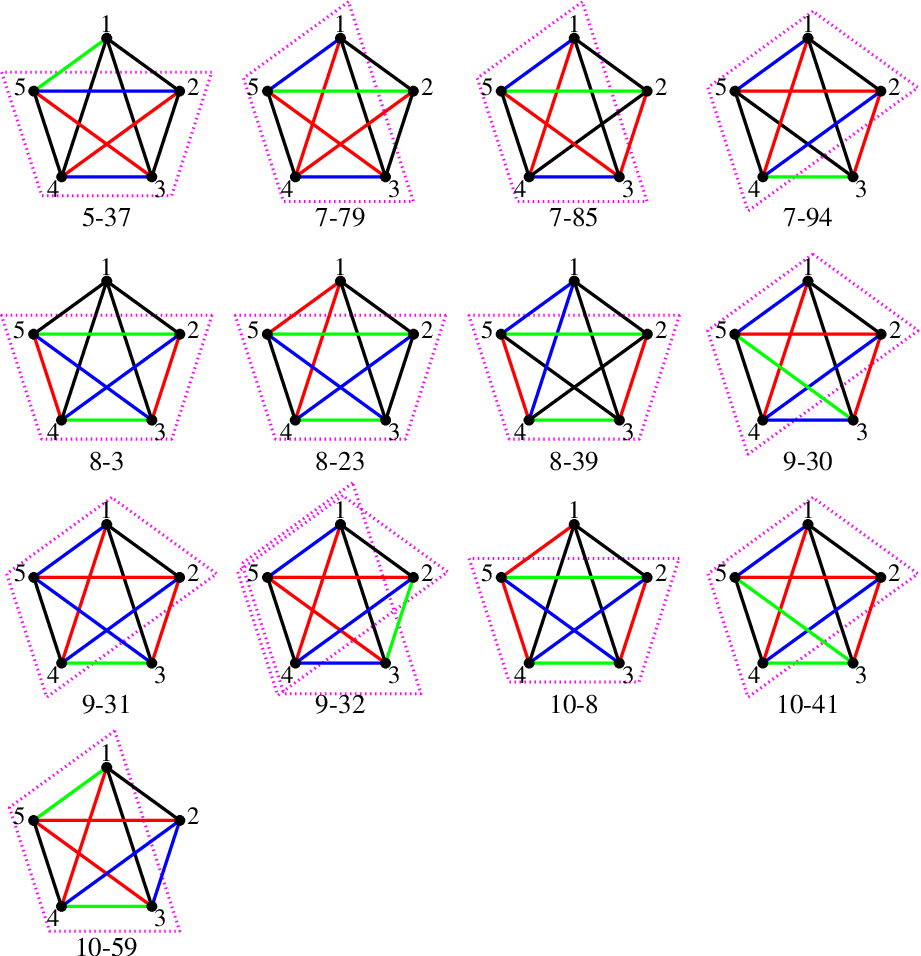}
   	\caption{Excluded case 4 in Theorem~\ref{thm:main}.}
   	\label{fig:first-exceptional-type}
   \end{figure}
   
   (1) For case $(5\text{-}37)$, state $1$ is orthogonal to states $\{2,3,4\}$ on the first subsystem and to state $5$ on the fourth subsystem. Hence, the conditions of Lemma~\ref{lem:31-removal} are satisfied, and state $1$ can be identified by LOCC.
   
   The remaining four states are $\{2,3,4,5\}$. The orthogonality graph induced by these states is identical to case $(1\text{-}3)$ in Fig.~\ref{fig:four-tripartite-222} after an appropriate relabeling of the vertices. By Lemma~\ref{lem:four-tripartite-222-exception}, these four states are locally indistinguishable when their three local ranks are all equal to $2$, and are distinguishable by LOCC with nonzero probability otherwise.
   
   (2) For case $(7\text{-}79)$, the only orthogonality relation on the fourth subsystem is between states $2$ and $5$. Let $|\delta_2\rangle$ denote the normalized fourth local component of state $2$. The fourth party performs the measurement $\{P,I-P\}$, where $P=|\delta_2\rangle\langle\delta_2|$.
   
   For the outcome corresponding to $P$, state $5$ is excluded, leaving the states $\{1,2,3,4\}$. Their vector of the numbers of pairwise orthogonal relations is $(3,2,1,0)$. Since there is no orthogonality relation on the fourth subsystem, these four states can be regarded as tripartite OPSs with the vector $(3,2,1)$. By Lemma~\ref{lem:four-tripartite-411-321}, they can be perfectly distinguished by LOCC.
   
   For the outcome corresponding to $I-P$, state $2$ is excluded, leaving the states $\{1,3,4,5\}$. The orthogonality graph induced by these four states is identical to case $(1\text{-}3)$ in Fig.~\ref{fig:four-tripartite-222} after an appropriate relabeling of the vertices. By Lemma~\ref{lem:four-tripartite-222-exception}, these four states are locally indistinguishable when their three local ranks are all equal to $2$, and are distinguishable by LOCC with nonzero probability otherwise.
   
   (3) For case $(7\text{-}85)$, state $2$ is orthogonal to states $\{1,4\}$ on the first subsystem, to state $3$ on the second subsystem, and to state $5$ on the fourth subsystem. Hence, the conditions of Theorem~\ref{thm:211-reduction} are satisfied, and state $2$ can be identified by LOCC.
   
   The remaining four states are $\{1,3,4,5\}$. The orthogonality graph induced by these states is identical to case $(1\text{-}3)$ in Fig.~\ref{fig:four-tripartite-222} after an appropriate relabeling of the vertices. By Lemma~\ref{lem:four-tripartite-222-exception}, these four states cannot be perfectly distinguished by LOCC when their three local ranks are all equal to $2$, and are distinguishable by LOCC with nonzero probability otherwise.
   
   (4) For case $(7\text{-}94)$, state $3$ is orthogonal to states $\{1,5\}$ on the first subsystem, to state $2$ on the second subsystem, and to state $4$ on the fourth subsystem. Hence, the conditions of Theorem~\ref{thm:211-reduction} are satisfied, and state $3$ can be identified by LOCC.
   
   The remaining four states are $\{1,2,4,5\}$. The orthogonality graph induced by these states is identical to case $(1\text{-}3)$ in Fig.~\ref{fig:four-tripartite-222} after an appropriate relabeling of the vertices. By Lemma~\ref{lem:four-tripartite-222-exception}, these four states are locally indistinguishable when their three local ranks are all equal to 2, and are distinguishable by LOCC with nonzero probability otherwise.
   
   (5) For case $(8\text{-}3)$, state $1$ is orthogonal to states $\{2,3,4,5\}$ on the first subsystem. Hence, the conditions of Lemma~\ref{lem:degree4-removal} are satisfied, and state $1$ can be identified by LOCC.
   
   The remaining four states are $\{2,3,4,5\}$. The orthogonality graph induced by these states is identical to case $(1\text{-}3)$ in Fig.~\ref{fig:four-tripartite-222} after an appropriate relabeling of the vertices. By Lemma~\ref{lem:four-tripartite-222-exception}, these four states are locally indistinguishable when their three local ranks are all equal to $2$, and are distinguishable by LOCC with nonzero probability otherwise.
   
   (6) For case $(8\text{-}23)$, state $1$ is orthogonal to states $\{2,3\}$ on the first subsystem and to states $\{4,5\}$ on the second subsystem. Since states $2$ and $3$ are not orthogonal on the second subsystem, the conditions of Theorem~\ref{thm:22-reduction} are satisfied, and state $1$ can be identified by LOCC.
   
   The remaining four states are $\{2,3,4,5\}$. The orthogonality graph induced by these states is identical to case $(1\text{-}3)$ in Fig.~\ref{fig:four-tripartite-222} after an appropriate relabeling of the vertices. By Lemma~\ref{lem:four-tripartite-222-exception}, these four states are locally indistinguishable when their three local ranks are all equal to $2$, and are distinguishable by LOCC with nonzero probability otherwise.
   
   (7) For case $(8\text{-}39)$, state $1$ is orthogonal to states $\{2,3\}$ on the first subsystem and to states $\{4,5\}$ on the third subsystem. Since states $4$ and $5$ are orthogonal on the second subsystem rather than on the first subsystem, the conditions of Theorem~\ref{thm:22-reduction} are satisfied, and state $1$ can be identified by LOCC.
   
   The remaining four states are $\{2,3,4,5\}$. The orthogonality graph induced by these states is identical to case $(1\text{-}3)$ in Fig.~\ref{fig:four-tripartite-222} after an appropriate relabeling of the vertices. By Lemma~\ref{lem:four-tripartite-222-exception}, these four states are locally indistinguishable when their three local ranks are all equal to $2$, and are distinguishable by LOCC with nonzero probability otherwise.
   
   (8) For case $(9\text{-}30)$, the only orthogonality relation on the fourth subsystem is between states $3$ and $5$. Let $|\delta_3\rangle$ denote the normalized fourth local component of state $3$. The fourth party performs the measurement $\{P,I-P\}$, where $P=|\delta_3\rangle\langle\delta_3|$.
   
   For the outcome corresponding to $P$, state $5$ is excluded, leaving the states $\{1,2,3,4\}$. The orthogonality graph induced by these four states is identical to case $(1\text{-}1)$ in Fig.~\ref{fig:four-tripartite-222} after an appropriate relabeling of the vertices. By Lemma~\ref{lem:four-tripartite-222-dist}, they can be perfectly distinguished by LOCC.
   
   For the outcome corresponding to $I-P$, state $3$ is excluded, leaving the states $\{1,2,4,5\}$. The orthogonality graph induced by these four states is identical to case $(1\text{-}3)$ in Fig.~\ref{fig:four-tripartite-222} after an appropriate relabeling of the vertices. By Lemma~\ref{lem:four-tripartite-222-exception}, these four states are locally indistinguishable when their three local ranks are all equal to $2$, and are distinguishable by LOCC with nonzero probability otherwise.
   
   (9) For case $(9\text{-}31)$, the only orthogonality relation on the fourth subsystem is between states $3$ and $4$. Let $|\delta_3\rangle$ denote the normalized fourth local component of state $3$. The fourth party performs the measurement $\{P,I-P\}$, where $P=|\delta_3\rangle\langle\delta_3|$.
   
   For the outcome corresponding to $P$, state $4$ is excluded, leaving the states $\{1,2,3,5\}$. The orthogonality graph induced by these four states is identical to case $(1\text{-}2)$ in Fig.~\ref{fig:four-tripartite-222} after an appropriate relabeling of the vertices. By Lemma~\ref{lem:four-tripartite-222-dist}, they can be perfectly distinguished by LOCC.
   
   For the outcome corresponding to $I-P$, state $3$ is excluded, leaving the states $\{1,2,4,5\}$. The orthogonality graph induced by these four states is identical to case $(1\text{-}3)$ in Fig.~\ref{fig:four-tripartite-222} after an appropriate relabeling of the vertices. By Lemma~\ref{lem:four-tripartite-222-exception}, these four states are locally indistinguishable when their three local ranks are all equal to $2$, and are distinguishable by LOCC with nonzero probability otherwise.
   
   (10) For case $(9\text{-}32)$, the only orthogonality relation on the fourth subsystem is between states $2$ and $3$. Let $|\delta_2\rangle$ denote the normalized fourth local component of state $2$. The fourth party performs the measurement $\{P,I-P\}$, where $P=|\delta_2\rangle\langle\delta_2|$.
   
   For the outcome corresponding to $P$, state $3$ is excluded, leaving the states $\{1,2,4,5\}$. For the outcome corresponding to $I-P$, state $2$ is excluded, leaving the states $\{1,3,4,5\}$. In either case, the orthogonality graph induced by the remaining four states is identical to case $(1\text{-}3)$ in Fig.~\ref{fig:four-tripartite-222} after an appropriate relabeling of the vertices. By Lemma~\ref{lem:four-tripartite-222-exception}, the remaining four states are locally indistinguishable when their three local ranks are all equal to $2$, and are distinguishable by LOCC with nonzero probability otherwise.
   
   (11) For case $(10\text{-}8)$, state $1$ is orthogonal to states $\{2,3,4\}$ on the first subsystem and to state $5$ on the second subsystem. Hence, the conditions of Lemma~\ref{lem:31-removal} are satisfied, and state $1$ can be identified by LOCC.
   
   The remaining four states are $\{2,3,4,5\}$. The orthogonality graph induced by these states is identical to case $(1\text{-}3)$ in Fig.~\ref{fig:four-tripartite-222} after an appropriate relabeling of the vertices. By Lemma~\ref{lem:four-tripartite-222-exception}, these four states are locally indistinguishable when their three local ranks are all equal to $2$, and are distinguishable by LOCC with nonzero probability otherwise.
   
   (12) For case $(10\text{-}41)$, state $3$ is orthogonal to states $\{4,5\}$ on the fourth subsystem, to state $1$ on the first subsystem, and to state $2$ on the second subsystem. Since states $1$ and $2$ are orthogonal on the first subsystem rather than on the fourth subsystem, the conditions of Theorem~\ref{thm:211-reduction} are satisfied, and state $3$ can be identified by LOCC.
   
   The remaining four states are $\{1,2,4,5\}$. The orthogonality graph induced by these states is identical to case $(1\text{-}3)$ in Fig.~\ref{fig:four-tripartite-222} after an appropriate relabeling of the vertices. By Lemma~\ref{lem:four-tripartite-222-exception}, these four states are locally indistinguishable when their three local ranks are all equal to $2$, and are distinguishable by LOCC with nonzero probability otherwise.
   
   (13) For case $(10\text{-}59)$, state $2$ is orthogonal to states $\{3,4\}$ on the third subsystem, to state $1$ on the first subsystem, and to state $5$ on the second subsystem. Since states $1$ and $5$ are orthogonal on the fourth subsystem rather than on the third subsystem, the conditions of Theorem~\ref{thm:211-reduction} are satisfied, and state $2$ can be identified by LOCC.
   
   The remaining four states are $\{1,3,4,5\}$. The orthogonality graph induced by these states is identical to case $(1\text{-}3)$ in Fig.~\ref{fig:four-tripartite-222} after an appropriate relabeling of the vertices. By Lemma~\ref{lem:four-tripartite-222-exception}, these four states are locally indistinguishable when their three local ranks are all equal to $2$, and are distinguishable by LOCC with nonzero probability otherwise.
   
   Combining the classification in Sec.~III with the exceptional case analysis above, and identifying graphs related by a relabeling of vertices or a permutation of parties, the genuine quadripartite categories contain $371$ distinct orthogonality graphs. Among
   them, $355$ graphs correspond to five OPSs that can be perfectly distinguished by LOCC, while the remaining $16$ graphs belong to the four exceptional structural types stated in
   Theorem~\ref{thm:main}. For these exceptional graphs, local distinguishability depends further on the local ranks or on the specific parameters of the local states, as discussed in this
   section.
   
	\section{\label{sec5}Conclusion}

	In this paper, we studied the local distinguishability of five quadripartite OPSs under the assumption that every pair of OPSs is orthogonal on exactly one subsystem. By using the vectors of the numbers of pairwise orthogonality relations, we divided the feasible sets of five quadripartite OPSs into different categories and described their structures by colored orthogonality graphs. After identifying graphs up to relabeling of states and permutation of subsystems, this enumeration gives $371$ distinct orthogonality graphs. We showed that a set of five quadripartite OPSs corresponding to one of  $355$ orthogonality graphs can be perfectly distinguished by LOCC. The remaining $16$ graphs fall into four exceptional structural types, represented by the orthogonality graphs $(5\text{-}37)$, $(7\text{-}79)$, $(7\text{-}85)$, $(7\text{-}94)$, $(8\text{-}3)$, $(8\text{-}23)$, $(8\text{-}25)$, $(8\text{-}39)$, $(9\text{-}30)$, $(9\text{-}31)$, $(9\text{-}32)$, $(10\text{-}8)$, $(10\text{-}27)$, $(10\text{-}41)$, $(10\text{-}43)$, and $(10\text{-}59)$. For sets of five quadripartite OPSs corresponding to these exceptional cases, we further discussed their local indistinguishability and gave representative analyses. Thus, our work extends the known results on five OPSs from bipartite and tripartite systems to quadripartite systems, and provides a classification framework for studying small sets of OPSs in multipartite quantum systems.
\clearpage

\begin{acknowledgments}
	\vspace{-10pt}
    This work is supported by Natural Science Foundation of Shandong Province of China (Grant No. ZR2026MS1078 and ZR2023MF080), Beijing Natural Science Foundation (Grant No. 4252014) and 
Shandong Key Laboratory of Smart Mine Information Technology.
\end{acknowledgments}
	\appendix
	\renewcommand{\appendixname}{APPENDIX}
	\section{PROOF OF THEOREM~\ref{thm:main}}
	\label{app:A}
	
	Now, we prove that five quadripartite OPSs, any two of which are orthogonal only on one subsystem, can be perfectly distinguished by LOCC, except the five quadripartite OPSs whose
	orthogonality graphs are $(5\text{-}37)$, $(7\text{-}79)$,
	$(7\text{-}85)$, $(7\text{-}94)$, $(8\text{-}3)$, $(8\text{-}23)$,
	$(8\text{-}25)$, $(8\text{-}39)$, $(9\text{-}30)$, $(9\text{-}31)$,
	$(9\text{-}32)$, $(10\text{-}8)$, $(10\text{-}27)$, $(10\text{-}41)$,
	$(10\text{-}43)$, or $(10\text{-}59)$.
	
	\subsection*{1. Category \((7,1,1,1)\)}
	
	The orthogonality graphs of five quadripartite OPSs with the vector of the numbers of pairwise orthogonal relations \((7,1,1,1)\) are shown in Fig.~\ref{fig:7111}. There are four different cases in this category, i.e., \((2\text{-}1)\)--\((2\text{-}4)\).
	
	(1) In cases \((2\text{-}1)\), \((2\text{-}2)\), and \((2\text{-}3)\), state \(1\) is orthogonal to all the other states on the first subsystem. This means that the degree of state \(1\) on the first subsystem is \(4\). According to Lemma~\ref{lem:degree4-removal}, the first party can identify state \(1\) from all the other states. The result is that four states \(\{2,3,4,5\}\) are left to be distinguished.
	
	For these three cases, the vector of the numbers of pairwise orthogonal relations of states \(\{2,3,4,5\}\) is \((3,1,1,1)\), up to a permutation of parties. By Lemma~\ref{lem:four-fourpartite-3111-2211}, states \(\{2,3,4,5\}\) can be perfectly distinguished by LOCC. Therefore, five quadripartite OPSs in cases \((2\text{-}1)\), \((2\text{-}2)\), and
	\((2\text{-}3)\) can be perfectly distinguished by LOCC.
	
	(2) In case \((2\text{-}4)\), state \(1\) is orthogonal to states \(\{2,3,4\}\) on the first subsystem and is orthogonal to state \(5\) on the third subsystem. Hence the degree of state \(1\) is \(3\) on the first subsystem and \(1\) on the third subsystem. According to
	Lemma~\ref{lem:31-removal}, state \(1\) can be identified by LOCC. The result is that four states \(\{2,3,4,5\}\) are left to be distinguished.
	
	For states \(\{2,3,4,5\}\), the vector of the numbers of pairwise orthogonal relations is \((4,1,0,1)\). Since there is no orthogonal relation on the third subsystem, these four states can be regarded as four tripartite OPSs with the vector \((4,1,1)\), up to a permutation of
	parties. By Lemma~\ref{lem:four-tripartite-411-321}, states \(\{2,3,4,5\}\) can be perfectly distinguished by LOCC. Therefore, five quadripartite OPSs in case \((2\text{-}4)\) can be perfectly distinguished by LOCC.
	
	Therefore, five quadripartite OPSs with the vector \((7,1,1,1)\) can be perfectly distinguished by LOCC.
	
	\subsection*{2. Category \((6,2,1,1)\)}
	
	The orthogonality graphs of five quadripartite OPSs with the vector of the numbers of pairwise orthogonal relations \((6,2,1,1)\) are shown in Fig.~\ref{fig:6211}. There are seventeen different cases in this category, labeled \((3\text{-}1)\)--\((3\text{-}17)\).
	
	(1) In cases \((3\text{-}1)\), \((3\text{-}2)\), \((3\text{-}3)\), \((3\text{-}4)\), \((3\text{-}5)\), and \((3\text{-}6)\), state \(1\) is orthogonal to all the other states on the first subsystem. This means that the degree of state \(1\) on the first subsystem is \(4\). According to Lemma~\ref{lem:degree4-removal}, the first party can identify state \(1\) from all the other states. The result is that four states \(\{2,3,4,5\}\) are left to be distinguished.
	
	For these six cases, the vector of the numbers of pairwise orthogonal relations of states \(\{2,3,4,5\}\) is \((2,2,1,1)\), up to a permutation of parties. By Lemma~\ref{lem:four-fourpartite-3111-2211}, states \(\{2,3,4,5\}\) can be perfectly distinguished by LOCC. Therefore, five quadripartite OPSs in cases \((3\text{-}1)\)--\((3\text{-}6)\) can be
	perfectly distinguished by LOCC.
	
	(2) In cases \((3\text{-}7)\), \((3\text{-}8)\), \((3\text{-}9)\), \((3\text{-}10)\), \((3\text{-}11)\), \((3\text{-}12)\), and \((3\text{-}13)\), state \(1\) is orthogonal to states \(\{2,3,4\}\) on the first subsystem and is orthogonal to state \(5\) on the second subsystem. Hence the degree of state \(1\) is \(3\) on the first subsystem and \(1\) on the second subsystem. According to Lemma~\ref{lem:31-removal}, state \(1\) can be identified by LOCC. The
	result is that four states \(\{2,3,4,5\}\) are left to be distinguished.
	
	For these seven cases, the vector of the numbers of pairwise orthogonal relations of states \(\{2,3,4,5\}\) is \((3,1,1,1)\), up to a permutation of parties. By Lemma~\ref{lem:four-fourpartite-3111-2211}, states \(\{2,3,4,5\}\) can be perfectly distinguished by LOCC. Therefore, five quadripartite OPSs in cases \((3\text{-}7)\)--\((3\text{-}13)\) can be
	perfectly distinguished by LOCC.
	
	(3) In cases \((3\text{-}14)\) and \((3\text{-}15)\), state \(1\) is orthogonal to states \(\{2,3,4\}\) on the first subsystem and is orthogonal to state \(5\) on the third subsystem. Hence the degree of state \(1\) is \(3\) on the first subsystem and \(1\) on the third
	subsystem. According to Lemma~\ref{lem:31-removal}, state \(1\) can be identified by LOCC. The result is that four states \(\{2,3,4,5\}\) are left to be distinguished.
	
	For states \(\{2,3,4,5\}\), the vector of the numbers of pairwise orthogonal relations is \((3,2,0,1)\). Since there is no orthogonal relation on the third subsystem, these four states can be regarded as four tripartite OPSs with the vector \((3,2,1)\), up to a permutation of
	parties. By Lemma~\ref{lem:four-tripartite-411-321}, states \(\{2,3,4,5\}\) can be perfectly distinguished by LOCC. Therefore, five quadripartite OPSs in cases \((3\text{-}14)\) and \((3\text{-}15)\) can be perfectly distinguished by LOCC.
	
	(4) In cases \((3\text{-}16)\) and \((3\text{-}17)\), state \(1\) is orthogonal to states \(\{2,3,4\}\) on the first subsystem and is orthogonal to state \(5\) on the fourth subsystem. Hence the degree of state \(1\) is \(3\) on the first subsystem and \(1\) on the fourth
	subsystem. According to Lemma~\ref{lem:31-removal}, state \(1\) can be identified by LOCC. The result is that four states \(\{2,3,4,5\}\) are left to be distinguished.
	
	For states \(\{2,3,4,5\}\), the vector of the numbers of pairwise orthogonal relations is \((3,2,1,0)\). Since there is no orthogonal relation on the fourth subsystem, these four states can be regarded as four tripartite OPSs with the vector \((3,2,1)\), up to a permutation of
	parties. By Lemma~\ref{lem:four-tripartite-411-321}, states \(\{2,3,4,5\}\) can be perfectly distinguished by LOCC. Therefore, five quadripartite OPSs in cases \((3\text{-}16)\) and \((3\text{-}17)\) can be perfectly distinguished by LOCC.
	
	Therefore, five quadripartite OPSs with the vector \((6,2,1,1)\) can be perfectly distinguished by LOCC.
	
	\subsection*{3. Category \((5,3,1,1)\)}
	
	The orthogonality graphs of five quadripartite OPSs with the vector of the numbers of pairwise orthogonal relations \((5,3,1,1)\) are shown in Fig.~\ref{fig:5311}. There are thirty different cases in this category, labeled \((4\text{-}1)\)--\((4\text{-}30)\).
	
	(1) In cases \((4\text{-}1)\), \((4\text{-}2)\), \((4\text{-}3)\), and \((4\text{-}4)\), state \(1\) is orthogonal to all the other states on the first subsystem. This means that the degree of state \(1\) on the first subsystem is \(4\). According to Lemma~\ref{lem:degree4-removal}, the
	first party can identify state \(1\) from all the other states. The result is that four states \(\{2,3,4,5\}\) are left to be distinguished.
	
	For these four cases, the vector of the numbers of pairwise orthogonal relations of states \(\{2,3,4,5\}\) is \((1,3,1,1)\), up to a permutation of parties. Hence these four states can be regarded as four quadripartite OPSs with the vector \((3,1,1,1)\). By Lemma~\ref{lem:four-fourpartite-3111-2211}, states \(\{2,3,4,5\}\) can be perfectly distinguished by LOCC. Therefore, five quadripartite OPSs in cases \((4\text{-}1)\)--\((4\text{-}4)\) can be perfectly distinguished by LOCC.
	
	(2) In cases \((4\text{-}5)\)--\((4\text{-}19)\), state \(1\) is orthogonal to states \(\{2,3,4\}\) on the first subsystem and is orthogonal to state \(5\) on the second subsystem. Hence the degree of state \(1\) is \(3\) on the first subsystem and \(1\) on the second
	subsystem. According to Lemma~\ref{lem:31-removal}, state \(1\) can be identified by LOCC. The result is that four states \(\{2,3,4,5\}\) are left to be distinguished.
	
	For these cases, the vector of the numbers of pairwise orthogonal relations of states \(\{2,3,4,5\}\) is \((2,2,1,1)\), up to a permutation of parties. By Lemma~\ref{lem:four-fourpartite-3111-2211}, states \(\{2,3,4,5\}\) can be perfectly distinguished by LOCC. Therefore, five quadripartite OPSs in cases \((4\text{-}5)\)--\((4\text{-}19)\) can be
	perfectly distinguished by LOCC.
	(3) In cases \((4\text{-}20)\)--\((4\text{-}24)\), state \(1\) is orthogonal to states \(\{2,3,4\}\) on the first subsystem and is orthogonal to state \(5\) on the third subsystem. Hence the degree of state \(1\) is \(3\) on the first subsystem and \(1\) on the third subsystem. According to Lemma~\ref{lem:31-removal}, state \(1\) can be identified by LOCC. The result is that four states \(\{2,3,4,5\}\) are left to be distinguished.
	
	For states \(\{2,3,4,5\}\), the vector of the numbers of pairwise orthogonal relations is \((2,3,0,1)\). Since there is no orthogonal relation on the third subsystem, these four states can be regarded as four tripartite OPSs with the vector \((3,2,1)\), up to a permutation of
	parties. By Lemma~\ref{lem:four-tripartite-411-321}, states \(\{2,3,4,5\}\) can be perfectly distinguished by LOCC. Therefore, five quadripartite OPSs in cases \((4\text{-}20)\)--\((4\text{-}24)\) can be perfectly distinguished by LOCC.
	
	(4) In cases \((4\text{-}25)\)--\((4\text{-}28)\), state \(1\) is orthogonal to states \(\{2,3,4\}\) on the first subsystem and is orthogonal to state \(5\) on the fourth subsystem. Hence the degree of state \(1\) is \(3\) on the first subsystem and \(1\) on the fourth
	subsystem. According to Lemma~\ref{lem:31-removal}, state \(1\) can be identified by LOCC. The result is that four states \(\{2,3,4,5\}\) are left to be distinguished.
	
	For states \(\{2,3,4,5\}\), the vector of the numbers of pairwise orthogonal relations is \((2,3,1,0)\). Since there is no orthogonal relation on the fourth subsystem, these four states can be regarded as four tripartite OPSs with the vector \((3,2,1)\), up to a permutation of
	parties. By Lemma~\ref{lem:four-tripartite-411-321}, states \(\{2,3,4,5\}\) can be perfectly distinguished by LOCC. Therefore, five quadripartite OPSs in cases \((4\text{-}25)\)--\((4\text{-}28)\) can be perfectly distinguished by LOCC.
	
	(5) In case \((4\text{-}29)\), state \(3\) is orthogonal to state \(4\) on the third subsystem, which can be used to identify states \(\{1,2,3,5\}\) and \(\{1,2,4,5\}\) by the third party.
	
	For states \(\{1,2,3,5\}\), the vector of the numbers of pairwise orthogonal relations is \((3,2,0,1)\). This means that there is no orthogonal relation among states \(\{1,2,3,5\}\) on the third subsystem. Thus, states \(\{1,2,3,5\}\) can be seen as a set of four tripartite OPSs
	with the vector \((3,2,1)\) when the third subsystem is omitted. By Lemma~\ref{lem:four-tripartite-411-321}, states \(\{1,2,3,5\}\) can be perfectly distinguished by LOCC.
	
	For states \(\{1,2,4,5\}\), the vector of the numbers of pairwise orthogonal relations is \((3,2,0,1)\). This means that there is no orthogonal relation among states \(\{1,2,4,5\}\) on the third subsystem. Thus, states \(\{1,2,4,5\}\) can be seen as a set of four tripartite OPSs
	with the vector \((3,2,1)\) when the third subsystem is omitted. By Lemma~\ref{lem:four-tripartite-411-321}, states \(\{1,2,4,5\}\) can be perfectly distinguished by LOCC. Therefore, five quadripartite OPSs in case \((4\text{-}29)\) can be perfectly distinguished by LOCC.

	(6) In case \((4\text{-}30)\), state \(3\) is orthogonal to states \(\{2,4\}\) on the second subsystem and is orthogonal to states \(\{1,5\}\) on the first subsystem. Moreover, the edge between states \(1\) and \(5\) is on the third subsystem, and hence it is not on the second subsystem. Therefore, by Theorem~\ref{thm:22-reduction}, the second party can identify state \(3\) from all the other states.
	
	The result is that four states \(\{1,2,4,5\}\) are left to be distinguished. For states \(\{1,2,4,5\}\), the vector of the numbers of pairwise orthogonal relations is \((3,1,1,1)\), up to a permutation of parties. By Lemma~\ref{lem:four-fourpartite-3111-2211}, states \(\{1,2,4,5\}\) can be perfectly distinguished by LOCC. Therefore, five quadripartite OPSs in case \((4\text{-}30)\) can be perfectly distinguished by LOCC.

	\subsection*{4. Category \((5,2,2,1)\)}
	
	The orthogonality graphs of five quadripartite OPSs with the vector of the numbers of pairwise orthogonal relations \((5,2,2,1)\) are shown in Fig.~\ref{fig:5221}. There are forty-two different cases in this category, labeled \((5\text{-}1)\)--\((5\text{-}42)\).
	
	(1) In cases \((5\text{-}1)\), \((5\text{-}2)\), \((5\text{-}3)\), \((5\text{-}4)\), \((5\text{-}5)\), and \((5\text{-}6)\), state \(1\) is orthogonal to states \(\{2,3,4,5\}\) on the first subsystem. This means that the degree of state \(1\) on the first subsystem is \(4\). According to Lemma~\ref{lem:degree4-removal}, the first party can identify state \(1\) from all the other states. The result is that four states \(\{2,3,4,5\}\) are left to be distinguished.
	
	For these cases, the vector of the numbers of pairwise orthogonal relations of states \(\{2,3,4,5\}\) is \((1,2,2,1)\), up to a permutation of parties. Hence these four states can be regarded as four quadripartite OPSs with the vector \((2,2,1,1)\). By Lemma~\ref{lem:four-fourpartite-3111-2211}, states \(\{2,3,4,5\}\) can be perfectly distinguished by LOCC. Therefore, five quadripartite OPSs in cases \((5\text{-}1)\), \((5\text{-}2)\), \((5\text{-}3)\), \((5\text{-}4)\), \((5\text{-}5)\), and \((5\text{-}6)\) can be perfectly distinguished by LOCC.
	
	(2) In cases \((5\text{-}7)\)--\((5\text{-}20)\), state \(1\) is orthogonal to states \(\{2,3,4\}\) on the first subsystem and is orthogonal to state \(5\) on the second subsystem. Hence the degree of state \(1\) is \(3\) on the first subsystem and \(1\) on the second
	subsystem. According to Lemma~\ref{lem:31-removal}, state \(1\) can be identified by LOCC. The result is that four states \(\{2,3,4,5\}\) are left to be distinguished.
	
	For these cases, the vector of the numbers of pairwise orthogonal relations of states \(\{2,3,4,5\}\) is \((2,1,2,1)\), up to a permutation of parties. Hence these four states can be regarded as four quadripartite OPSs with the vector \((2,2,1,1)\). By
	Lemma~\ref{lem:four-fourpartite-3111-2211}, states \(\{2,3,4,5\}\) can be perfectly distinguished by LOCC. Therefore, five quadripartite OPSs in cases \((5\text{-}7)\)--\((5\text{-}20)\) can be perfectly distinguished by LOCC.
	
	(3) In cases \((5\text{-}21)\)--\((5\text{-}31)\), state \(1\) is orthogonal to states \(\{2,3,4\}\) on the first subsystem and is orthogonal to state \(5\) on the third subsystem. Hence the degree of state \(1\) is \(3\) on the first subsystem and \(1\) on the third subsystem. According to Lemma~\ref{lem:31-removal}, state \(1\) can be identified by LOCC. The result is that four states \(\{2,3,4,5\}\) are left to be distinguished.
	
	For these cases, the vector of the numbers of pairwise orthogonal relations of states \(\{2,3,4,5\}\) is \((2,2,1,1)\), up to a permutation of parties. By Lemma~\ref{lem:four-fourpartite-3111-2211}, states \(\{2,3,4,5\}\) can be perfectly distinguished by LOCC. Therefore, five quadripartite OPSs in cases \((5\text{-}21)\)--\((5\text{-}31)\) can be
	perfectly distinguished by LOCC.
	
	(4) In cases \((5\text{-}32)\), \((5\text{-}33)\), \((5\text{-}34)\), \((5\text{-}35)\), \((5\text{-}36)\), \((5\text{-}38)\), and \((5\text{-}39)\), state \(1\) is orthogonal to states \(\{2,3,4\}\) on the first subsystem and is orthogonal to state \(5\) on the fourth
	subsystem. Hence the degree of state \(1\) is \(3\) on the first subsystem and \(1\) on the fourth subsystem. According to Lemma~\ref{lem:31-removal}, state \(1\) can be identified by LOCC. The result is that four states \(\{2,3,4,5\}\) are left to be distinguished.
	
	For these cases, the vector of the numbers of pairwise orthogonal relations of states \(\{2,3,4,5\}\) is \((2,2,2,0)\). This means that there is no orthogonal relation among states \(\{2,3,4,5\}\) on the fourth subsystem. Thus, states \(\{2,3,4,5\}\) can be regarded as four
	tripartite OPSs with the vector \((2,2,2)\) when the fourth subsystem is omitted. For the listed cases \((5\text{-}32)\), \((5\text{-}33)\), \((5\text{-}34)\), \((5\text{-}35)\), \((5\text{-}36)\), \((5\text{-}38)\), and \((5\text{-}39)\), the remaining four-state graphs correspond to the locally distinguishable graphs of type \((2,2,2)\). By Lemma~\ref{lem:four-tripartite-222-dist}, states \(\{2,3,4,5\}\) can be perfectly distinguished by LOCC. Therefore, five quadripartite OPSs in these cases can be perfectly distinguished by LOCC.
	
	(5) In case \((5\text{-}40)\), state \(5\) is orthogonal to states \(\{1,2\}\) on the third subsystem and is orthogonal to states \(\{3,4\}\) on the first subsystem. Moreover, the edge between states \(3\) and \(4\) is not on the third subsystem. Therefore, by Theorem~\ref{thm:22-reduction}, the third party can identify state \(5\) from all the other states.
	
	The result is that four states \(\{1,2,3,4\}\) are left to be distinguished. For states \(\{1,2,3,4\}\), the vector of the numbers of pairwise orthogonal relations is \((3,2,0,1)\). This means that there is no orthogonal relation among states \(\{1,2,3,4\}\) on the third
	subsystem. Thus, states \(\{1,2,3,4\}\) can be regarded as four tripartite OPSs with the vector \((3,2,1)\) when the third subsystem is omitted. By Lemma~\ref{lem:four-tripartite-411-321}, states \(\{1,2,3,4\}\) can be perfectly distinguished by LOCC. Therefore, five quadripartite OPSs in case \((5\text{-}40)\) can be perfectly distinguished by LOCC.
	
	(6) In case \((5\text{-}41)\), state \(2\) is orthogonal to state \(5\) on the fourth subsystem, which can be used to identify states \(\{1,2,3,4\}\) and \(\{1,3,4,5\}\) by the fourth party.
	
	For states \(\{1,2,3,4\}\), the vector of the numbers of pairwise orthogonal relations is \((3,2,1,0)\). This means that there is no orthogonal relation among states \(\{1,2,3,4\}\) on the fourth subsystem. Thus, states \(\{1,2,3,4\}\) can be regarded as four tripartite OPSs with
	the vector \((3,2,1)\) when the fourth subsystem is omitted. By Lemma~\ref{lem:four-tripartite-411-321}, states \(\{1,2,3,4\}\) can be perfectly distinguished by LOCC.
	
	For states \(\{1,3,4,5\}\), the vector of the numbers of pairwise orthogonal relations is \((3,1,2,0)\). This means that there is no orthogonal relation among states \(\{1,3,4,5\}\) on the fourth subsystem. Thus, states \(\{1,3,4,5\}\) can be regarded as four tripartite OPSs with
	the vector \((3,2,1)\), up to a permutation of parties, when the fourth subsystem is omitted. By Lemma~\ref{lem:four-tripartite-411-321}, states \(\{1,3,4,5\}\) can be perfectly distinguished by LOCC. Therefore, five quadripartite OPSs in case \((5\text{-}41)\) can be perfectly distinguished
	by LOCC.
	
	(7) In case \((5\text{-}42)\), state \(1\) is orthogonal to states \(\{4,5\}\) on the third subsystem and is orthogonal to states \(\{2,3\}\) on the first subsystem. Moreover, the edge between states \(2\) and \(3\) is not on the third subsystem. Therefore, by Theorem~\ref{thm:22-reduction}, the third party can identify state \(1\) from all the other states.
	
	The result is that four states \(\{2,3,4,5\}\) are left to be distinguished. For states \(\{2,3,4,5\}\), the vector of the numbers of pairwise orthogonal relations is \((3,2,0,1)\). This means that there is no orthogonal relation among states \(\{2,3,4,5\}\) on the third
	subsystem. Thus, states \(\{2,3,4,5\}\) can be regarded as four tripartite OPSs with the vector \((3,2,1)\) when the third subsystem is omitted. By Lemma~\ref{lem:four-tripartite-411-321}, states \(\{2,3,4,5\}\) can be perfectly distinguished by LOCC. Therefore, five quadripartite OPSs in case \((5\text{-}42)\) can be perfectly distinguished by LOCC.
	
	\subsection*{5. Category \((4,4,1,1)\)}
	
	The orthogonality graphs of five quadripartite OPSs with the vector of the numbers of pairwise orthogonal relations \((4,4,1,1)\) are shown in Fig.~\ref{fig:4411}. There are twenty different cases in this category, labeled \((6\text{-}1)\)--\((6\text{-}20)\).
	
	(1) In case \((6\text{-}1)\), state \(1\) is orthogonal to states \(\{2,3,4,5\}\) on the first subsystem. This means that the degree of state \(1\) on the first subsystem is \(4\). According to
	Lemma~\ref{lem:degree4-removal}, the first party can identify state \(1\) from all the other states. The result is that four states \(\{2,3,4,5\}\) are left to be distinguished.
	
	For states \(\{2,3,4,5\}\), the vector of the numbers of pairwise orthogonal relations is \((0,4,1,1)\). This means that there is no orthogonal relation among states \(\{2,3,4,5\}\) on the first subsystem. Thus, states \(\{2,3,4,5\}\) can be regarded as four tripartite OPSs with
	the vector \((4,1,1)\) when the first subsystem is omitted. By Lemma~\ref{lem:four-tripartite-411-321}, states \(\{2,3,4,5\}\) can be perfectly distinguished by LOCC. Therefore, five quadripartite OPSs in case \((6\text{-}1)\) can be perfectly distinguished by LOCC.
	
	(2) In cases \((6\text{-}2)\)--\((6\text{-}9)\), state \(1\) is orthogonal to states \(\{2,3,4\}\) on the first subsystem and is orthogonal to state \(5\) on the second subsystem. Hence the degree of state \(1\) is \(3\) on the first subsystem and \(1\) on the second
	subsystem. According to Lemma~\ref{lem:31-removal}, state \(1\) can be identified by LOCC. The result is that four states \(\{2,3,4,5\}\) are left to be distinguished.
	
	For these cases, the vector of the numbers of pairwise orthogonal relations of states \(\{2,3,4,5\}\) is \((1,3,1,1)\), up to a permutation of parties. Hence these four states can be regarded as four quadripartite OPSs with the vector \((3,1,1,1)\). By Lemma~\ref{lem:four-fourpartite-3111-2211}, states \(\{2,3,4,5\}\) can be perfectly distinguished by LOCC. Therefore, five quadripartite OPSs in cases \((6\text{-}2)\)--\((6\text{-}9)\) can be perfectly distinguished by LOCC.
	
	(3) In cases \((6\text{-}10)\) and \((6\text{-}11)\), state \(1\) is orthogonal to states \(\{2,3,4\}\) on the first subsystem and is orthogonal to state \(5\) on the third subsystem. Hence the degree of state \(1\) is \(3\) on the first subsystem and \(1\) on the third
	subsystem. According to Lemma~\ref{lem:31-removal}, state \(1\) can be identified by LOCC. The result is that four states \(\{2,3,4,5\}\) are left to be distinguished.
	
	For states \(\{2,3,4,5\}\), the vector of the numbers of pairwise orthogonal relations is \((1,4,0,1)\). This means that there is no orthogonal relation among states \(\{2,3,4,5\}\) on the third subsystem. Thus, states \(\{2,3,4,5\}\) can be regarded as four tripartite OPSs with
	the vector \((4,1,1)\) when the third subsystem is omitted. By Lemma~\ref{lem:four-tripartite-411-321}, states \(\{2,3,4,5\}\) can be perfectly distinguished by LOCC. Therefore, five quadripartite OPSs in cases \((6\text{-}10)\) and \((6\text{-}11)\) can be perfectly distinguished by LOCC.
	
	(4) In cases \((6\text{-}12)\) and \((6\text{-}13)\), state \(1\) is orthogonal to states \(\{2,3,4\}\) on the first subsystem and is orthogonal to state \(5\) on the fourth subsystem. Hence the degree of state \(1\) is \(3\) on the first subsystem and \(1\) on the fourth subsystem. According to Lemma~\ref{lem:31-removal}, state \(1\) can be identified by LOCC. The result is that four states \(\{2,3,4,5\}\) are left to be distinguished.
	
	For states \(\{2,3,4,5\}\), the vector of the numbers of pairwise orthogonal relations is \((1,4,1,0)\). This means that there is no orthogonal relation among states \(\{2,3,4,5\}\) on the fourth subsystem. Thus, states \(\{2,3,4,5\}\) can be regarded as four tripartite OPSs with the vector \((4,1,1)\) when the fourth subsystem is omitted. By Lemma~\ref{lem:four-tripartite-411-321}, states \(\{2,3,4,5\}\) can be perfectly distinguished by LOCC. Therefore, five quadripartite OPSs in cases \((6\text{-}12)\) and \((6\text{-}13)\) can be perfectly distinguished by LOCC.
	
	(5) In cases \((6\text{-}14)\), \((6\text{-}15)\), and \((6\text{-}16)\), state \(1\) is orthogonal to states \(\{2,3\}\) on the first subsystem and is orthogonal to states \(\{4,5\}\) on the second subsystem. Moreover, the edge between states \(2\) and \(3\) is not on the second subsystem. Therefore, by Theorem~\ref{thm:22-reduction}, the second party can identify state \(1\) from all the other states.
	
	The result is that four states \(\{2,3,4,5\}\) are left to be distinguished. For these cases, the vector of the numbers of pairwise orthogonal relations of states \(\{2,3,4,5\}\) is \((2,2,1,1)\), up to a permutation of parties. By Lemma~\ref{lem:four-fourpartite-3111-2211}, states \(\{2,3,4,5\}\) can be perfectly distinguished by LOCC. Therefore, five quadripartite OPSs in cases \((6\text{-}14)\), \((6\text{-}15)\), and \((6\text{-}16)\) can be perfectly distinguished by LOCC.
	
	(6) In cases \((6\text{-}17)\) and \((6\text{-}18)\), state \(1\) is orthogonal to states \(\{2,3\}\) on the first subsystem and is orthogonal to states \(\{4,5\}\) on the second subsystem. Moreover, the edge between states \(4\) and \(5\) is not on the first subsystem. Therefore, by Theorem~\ref{thm:22-reduction}, the first party can identify state \(1\) from all the other states.
	
	The result is that four states \(\{2,3,4,5\}\) are left to be distinguished. For these cases, the vector of the numbers of pairwise orthogonal relations of states \(\{2,3,4,5\}\) is \((2,2,1,1)\), up to a permutation of parties. By Lemma~\ref{lem:four-fourpartite-3111-2211}, states \(\{2,3,4,5\}\) can be perfectly distinguished by LOCC. Therefore, five quadripartite OPSs in cases \((6\text{-}17)\) and \((6\text{-}18)\) can be perfectly distinguished by LOCC.
	
	(7) In case \((6\text{-}19)\), state \(5\) is orthogonal to states \(\{1,3\}\) on the second subsystem, to state \(4\) on the first subsystem, and to state \(2\) on the fourth subsystem. Moreover, the edge between states \(4\) and \(2\) is not on the second subsystem. Therefore, by Theorem~\ref{thm:211-reduction}, the second party can identify state \(5\) from all the other states.
	
	The result is that four states \(\{1,2,3,4\}\) are left to be distinguished. For states \(\{1,2,3,4\}\), the vector of the numbers of pairwise orthogonal relations is \((3,2,1,0)\). This means that there is no orthogonal relation among states \(\{1,2,3,4\}\) on the fourth subsystem. Thus, states \(\{1,2,3,4\}\) can be regarded as four tripartite OPSs with the vector \((3,2,1)\) when the fourth subsystem is omitted. By Lemma~\ref{lem:four-tripartite-411-321}, states \(\{1,2,3,4\}\) can be perfectly distinguished by LOCC. Therefore, five quadripartite OPSs in case \((6\text{-}19)\) can be perfectly distinguished by LOCC.
	
	(8) In case \((6\text{-}20)\), state \(2\) is orthogonal to states \(\{1,5\}\) on the first subsystem and is orthogonal to states \(\{3,4\}\) on the second subsystem. Moreover, the edge between states \(3\) and \(4\) is not on the first subsystem. Therefore, by Theorem~\ref{thm:22-reduction}, the first party can identify state \(2\) from all the other states.
	
	The result is that four states \(\{1,3,4,5\}\) are left to be distinguished. For states \(\{1,3,4,5\}\), the vector of the numbers of pairwise orthogonal relations is \((2,2,1,1)\), up to a permutation of parties. By Lemma~\ref{lem:four-fourpartite-3111-2211}, states \(\{1,3,4,5\}\) can be perfectly distinguished by LOCC. Therefore, five quadripartite OPSs in case \((6\text{-}20)\) can be perfectly distinguished by LOCC.
	
	\subsection*{6. Category \((4,3,2,1)\)}
	The orthogonality graphs of five quadripartite OPSs with the vector of the numbers of pairwise orthogonal relations \((4,3,2,1)\) are shown in Fig.~\ref{fig:4321-part2}. There are one hundred and eighteen different cases in this category, labeled \((7\text{-}1)\)--\((7\text{-}118)\).
	
	(1) In cases \((7\text{-}1)\)--\((7\text{-}4)\), state \(1\) is orthogonal to states \(\{2,3,4,5\}\) on the first subsystem. This means that the degree of state \(1\) on the first subsystem is \(4\). According to Lemma~\ref{lem:degree4-removal}, the first party can identify state \(1\) from all the other states. The result is that four states \(\{2,3,4,5\}\) are left to be distinguished.
	
	For states \(\{2,3,4,5\}\), the vector of the numbers of pairwise orthogonal relations is \((0,3,2,1)\). This means that there is no orthogonal relation among states \(\{2,3,4,5\}\) on the first subsystem. Thus, states \(\{2,3,4,5\}\) can be regarded as four tripartite OPSs with the vector \((3,2,1)\) when the first subsystem is omitted. By Lemma~\ref{lem:four-tripartite-411-321}, states \(\{2,3,4,5\}\) can be perfectly distinguished by LOCC. Therefore, five quadripartite OPSs in cases \((7\text{-}1)\)--\((7\text{-}4)\) can be perfectly distinguished by LOCC.
	
	(2) In cases \((7\text{-}5)\)--\((7\text{-}34)\), state \(1\) is orthogonal to states \(\{2,3,4\}\) on the first subsystem and is orthogonal to state \(5\) on the second subsystem. Hence the degree of state \(1\) is \(3\) on the first subsystem and \(1\) on the second subsystem. According to Lemma~\ref{lem:31-removal}, state \(1\) can be identified by LOCC. The result is that four states \(\{2,3,4,5\}\) are left to be distinguished.
	
	For these cases, the vector of the numbers of pairwise orthogonal relations of states \(\{2,3,4,5\}\) is \((1,2,2,1)\). Hence these four states belong to the quadripartite type \((2,2,1,1)\). By Lemma~\ref{lem:four-fourpartite-3111-2211}, states \(\{2,3,4,5\}\) can be perfectly distinguished by LOCC. Therefore, five quadripartite OPSs in cases \((7\text{-}5)\)--\((7\text{-}34)\) can be perfectly distinguished by LOCC.
	
	(3) In cases \((7\text{-}35)\)--\((7\text{-}54)\), state \(1\) is orthogonal to states \(\{2,3,4\}\) on the first subsystem and is orthogonal to state \(5\) on the third subsystem. Hence the degree of state \(1\) is \(3\) on the first subsystem and \(1\) on the third subsystem. According to Lemma~\ref{lem:31-removal}, state \(1\) can be identified by LOCC. The result is that four states \(\{2,3,4,5\}\) are left to be distinguished.
	
	For these cases, the vector of the numbers of pairwise orthogonal relations of states \(\{2,3,4,5\}\) is \((1,3,1,1)\). Hence these four states belong to the quadripartite type \((3,1,1,1)\). By Lemma~\ref{lem:four-fourpartite-3111-2211}, states \(\{2,3,4,5\}\) can be perfectly distinguished by LOCC. Therefore, five quadripartite OPSs in cases \((7\text{-}35)\)--\((7\text{-}54)\) can be perfectly distinguished by LOCC.
	
	(4) In cases \((7\text{-}55)\)--\((7\text{-}66)\), state \(1\) is orthogonal to states \(\{2,3,4\}\) on the first subsystem and is orthogonal to state \(5\) on the fourth subsystem. Hence the degree of state \(1\) is \(3\) on the first subsystem and \(1\) on the fourth subsystem. According to Lemma~\ref{lem:31-removal}, state \(1\) can be identified by LOCC. The result is that four states \(\{2,3,4,5\}\) are left to be distinguished.
	
	For these cases, the vector of the numbers of pairwise orthogonal relations of states \(\{2,3,4,5\}\) is \((1,3,2,0)\). This means that there is no orthogonal relation among states \(\{2,3,4,5\}\) on the fourth subsystem. Thus, states \(\{2,3,4,5\}\) can be regarded as four tripartite OPSs with the vector \((3,2,1)\) when the fourth subsystem is omitted. By Lemma~\ref{lem:four-tripartite-411-321}, states \(\{2,3,4,5\}\) can be perfectly distinguished by LOCC. Therefore, five quadripartite OPSs in cases \((7\text{-}55)\)--\((7\text{-}66)\) can be perfectly distinguished by LOCC.
	
	(5) In cases \((7\text{-}67)\), \((7\text{-}68)\), and \((7\text{-}69)\), state \(1\) is orthogonal to states \(\{4,5\}\) on the second subsystem and is orthogonal to states \(\{2,3\}\) on the first subsystem. Moreover, the edge between states \(2\) and \(3\) is not on the second subsystem. Therefore, by Theorem~\ref{thm:22-reduction}, the second party can identify state \(1\) from all the other states.
	
	The result is that four states \(\{2,3,4,5\}\) are left to be distinguished. For states \(\{2,3,4,5\}\), the vector of the numbers of pairwise orthogonal relations is \((2,1,2,1)\). Hence these four states belong to the quadripartite type \((2,2,1,1)\). By Lemma~\ref{lem:four-fourpartite-3111-2211}, states \(\{2,3,4,5\}\) can be perfectly distinguished by LOCC. Therefore, five quadripartite OPSs in cases \((7\text{-}67)\), \((7\text{-}68)\), and \((7\text{-}69)\) can be perfectly distinguished by LOCC.
	
	(6) In cases \((7\text{-}70)\), \((7\text{-}71)\), \((7\text{-}72)\), and \((7\text{-}75)\), state \(1\) is orthogonal to states \(\{2,3\}\) on the first subsystem and is orthogonal to states \(\{4,5\}\) on the second subsystem. Moreover, the edge between states \(4\) and \(5\) is not on the first subsystem. Therefore, by Theorem~\ref{thm:22-reduction}, the first party can identify state \(1\) from all the other states.
	
	The result is that four states \(\{2,3,4,5\}\) are left to be distinguished. For states \(\{2,3,4,5\}\), the vector of the numbers of pairwise orthogonal relations is \((2,1,2,1)\). Hence these four states belong to the quadripartite type \((2,2,1,1)\). By Lemma~\ref{lem:four-fourpartite-3111-2211}, states \(\{2,3,4,5\}\) can be perfectly distinguished by LOCC. Therefore, five quadripartite OPSs in cases \((7\text{-}70)\), \((7\text{-}71)\), \((7\text{-}72)\), and \((7\text{-}75)\) can be perfectly distinguished by LOCC.
	
	(7) In cases \((7\text{-}73)\), \((7\text{-}74)\), and \((7\text{-}115)\), state \(2\) is orthogonal to states \(\{1,4\}\) on the first subsystem, is orthogonal to state \(3\) on the second subsystem, and is orthogonal to state \(5\) on the fourth subsystem. Moreover, the edge between states \(3\) and \(5\) is not on the first subsystem. Therefore, by Theorem~\ref{thm:211-reduction}, the first party can identify state \(2\) from all the other states. 
	
	The result is that four states \(\{1,3,4,5\}\) are left to be distinguished. For states \(\{1,3,4,5\}\), the vector of the numbers of pairwise orthogonal relations is \((2,2,2,0)\). This means that there is no orthogonal relation among states \(\{1,3,4,5\}\) on the fourth subsystem. Thus, states \(\{1,3,4,5\}\) can be regarded as four tripartite OPSs with the vector \((2,2,2)\) when the fourth subsystem is omitted. By Lemma~\ref{lem:four-tripartite-222-dist}, states \(\{1,3,4,5\}\) can be perfectly distinguished by LOCC. Therefore, five quadripartite OPSs in cases \((7\text{-}73)\), \((7\text{-}74)\), and \((7\text{-}115)\) can be perfectly distinguished by LOCC.
	
	(8) In cases \((7\text{-}76)\) and \((7\text{-}83)\), state \(2\) is orthogonal to states \(\{1,4\}\) on the first subsystem, is orthogonal to state \(3\) on the second subsystem, and is orthogonal to state \(5\) on the fourth subsystem. Moreover, the edge between states \(3\) and \(5\) is not on the first subsystem. Therefore, by Theorem~\ref{thm:211-reduction}, the first party can identify state \(2\) from all the other states. 
	
	The result is that four states \(\{1,3,4,5\}\) are left to be distinguished. For states \(\{1,3,4,5\}\), the vector of the numbers of pairwise orthogonal relations is \((2,2,2,0)\). This means that there is no orthogonal relation among states \(\{1,3,4,5\}\) on the fourth subsystem. Thus, states \(\{1,3,4,5\}\) can be regarded as four tripartite OPSs with the vector \((2,2,2)\) when the fourth subsystem is omitted. By Lemma~\ref{lem:four-tripartite-222-dist}, states \(\{1,3,4,5\}\) can be perfectly distinguished by LOCC. Therefore, five quadripartite OPSs in cases \((7\text{-}76)\) and \((7\text{-}83)\) can be perfectly distinguished by LOCC.
	
	(9) In case \((7\text{-}77)\), state \(5\) is orthogonal to states \(\{1,2\}\) on the third subsystem, is orthogonal to state \(4\) on the first subsystem, and is orthogonal to state \(3\) on the fourth subsystem. Moreover, the edge between states \(3\) and \(4\) is not on the third subsystem. Therefore, by Theorem~\ref{thm:211-reduction}, the third party can identify state \(5\) from all the other states.
	
	The result is that four states \(\{1,2,3,4\}\) are left to be distinguished. For states \(\{1,2,3,4\}\), the vector of the numbers of pairwise orthogonal relations is \((3,3,0,0)\). This means that there is no orthogonal relation among states \(\{1,2,3,4\}\) on the third and fourth subsystems. Thus, states \(\{1,2,3,4\}\) can be regarded as four bipartite OPSs when the third and fourth subsystems are omitted. By Lemma~\ref{lem:four-bipartite-ops}, states \(\{1,2,3,4\}\) can be perfectly distinguished by LOCC. Therefore, five quadripartite OPSs in case \((7\text{-}77)\) can be perfectly distinguished by LOCC. 
	
	(10) In case \((7\text{-}78)\), state \(5\) is orthogonal to states \(\{1,2\}\) on the third subsystem, is orthogonal to state \(3\) on the second subsystem, and is orthogonal to state \(4\) on the first subsystem. Moreover, the edge between states \(3\) and \(4\) is not on the third subsystem. Therefore, by Theorem~\ref{thm:211-reduction}, the third party can identify state \(5\) from all the other states.
	
	The result is that four states \(\{1,2,3,4\}\) are left to be distinguished. For states \(\{1,2,3,4\}\), the vector of the numbers of pairwise orthogonal relations is \((3,2,0,1)\). This means that there is no orthogonal relation among states \(\{1,2,3,4\}\) on the third subsystem. Thus, states \(\{1,2,3,4\}\) can be regarded as four tripartite OPSs with the vector \((3,2,1)\) when the third subsystem is omitted. By Lemma~\ref{lem:four-tripartite-411-321}, states \(\{1,2,3,4\}\) can be perfectly distinguished by LOCC. Therefore, five quadripartite OPSs in case \((7\text{-}78)\) can be perfectly distinguished by LOCC.
	
	(11) In cases \((7\text{-}80)\), \((7\text{-}81)\), \((7\text{-}82)\), \((7\text{-}84)\), \((7\text{-}86)\), \((7\text{-}92)\), and \((7\text{-}95)\), state \(1\) is orthogonal to states \(\{2,3\}\) on the first subsystem, is orthogonal to state \(4\) on the second subsystem, and is orthogonal to state \(5\) on the third subsystem. Moreover, the edge between states \(4\) and \(5\) is not on the first subsystem. Therefore, by Theorem~\ref{thm:211-reduction}, the first party can identify state \(1\) from all the other states.
	
	The result is that four states \(\{2,3,4,5\}\) are left to be distinguished. For states \(\{2,3,4,5\}\), the vector of the numbers of pairwise orthogonal relations is \((2,2,1,1)\). By Lemma~\ref{lem:four-fourpartite-3111-2211}, states \(\{2,3,4,5\}\) can be perfectly distinguished by LOCC. Therefore, five quadripartite OPSs in cases \((7\text{-}80)\), \((7\text{-}81)\), \((7\text{-}82)\), \((7\text{-}84)\), \((7\text{-}86)\), \((7\text{-}92)\), and \((7\text{-}95)\) can be perfectly distinguished by LOCC.
	
	(12) In cases \((7\text{-}87)\), \((7\text{-}88)\), \((7\text{-}89)\), \((7\text{-}97)\), \((7\text{-}98)\), and \((7\text{-}117)\), state \(2\) is orthogonal to states \(\{1,5\}\) on the first subsystem and is orthogonal to states \(\{3,4\}\) on the second subsystem. Moreover, the edge between states \(3\) and \(4\) is not on the first subsystem. Therefore, by Theorem~\ref{thm:22-reduction}, the first party can identify state \(2\) from all the other states.
	
	The result is that four states \(\{1,3,4,5\}\) are left to be distinguished. For states \(\{1,3,4,5\}\), the vector of the numbers of pairwise orthogonal relations is \((2,1,2,1)\). Hence these four states belong to the quadripartite type \((2,2,1,1)\). By Lemma~\ref{lem:four-fourpartite-3111-2211}, states \(\{1,3,4,5\}\) can be perfectly distinguished by LOCC. Therefore, five quadripartite OPSs in these cases can be perfectly distinguished by LOCC.
	
	(13) In case \((7\text{-}90)\), state \(5\) is orthogonal to states \(\{1,2\}\) on the third subsystem and is orthogonal to states \(\{3,4\}\) on the first subsystem. Moreover, the edge between states \(3\) and \(4\) is not on the third subsystem. Therefore, by Theorem~\ref{thm:22-reduction}, the third party can identify state \(5\) from all the other states. 
	
	The result is that four states \(\{1,2,3,4\}\) are left to be distinguished. For states \(\{1,2,3,4\}\), the vector of the numbers of pairwise orthogonal relations is \((2,3,0,1)\). This means that there is no orthogonal relation among states \(\{1,2,3,4\}\) on the third subsystem. Thus, states \(\{1,2,3,4\}\) can be regarded as four tripartite OPSs with the vector \((3,2,1)\) when the third subsystem is omitted. By Lemma~\ref{lem:four-tripartite-411-321}, states \(\{1,2,3,4\}\) can be perfectly distinguished by LOCC. Therefore, five quadripartite OPSs in case \((7\text{-}90)\) can be perfectly distinguished by LOCC.
	
	(14) In cases \((7\text{-}91)\), \((7\text{-}99)\), and \((7\text{-}104)\), state \(4\) is orthogonal to states \(\{1,2\}\) on the second subsystem, is orthogonal to state \(5\) on the first subsystem, and is orthogonal to state \(3\) on the third subsystem. Moreover, the edge between states \(3\) and \(5\) is not on the second subsystem. Therefore, by Theorem~\ref{thm:211-reduction}, the second party can identify state \(4\) from all the other states. 
	
	The result is that four states \(\{1,2,3,5\}\) are left to be distinguished. For states \(\{1,2,3,5\}\), the vector of the numbers of pairwise orthogonal relations is \((3,1,1,1)\). By Lemma~\ref{lem:four-fourpartite-3111-2211}, states \(\{1,2,3,5\}\) can be perfectly distinguished by LOCC. Therefore, five quadripartite OPSs in these cases can be perfectly distinguished by LOCC.
	
	(15) In case \((7\text{-}93)\), state \(2\) is orthogonal to states \(\{1,5\}\) on the first subsystem, is orthogonal to state \(4\) on the third subsystem, and is orthogonal to state \(3\) on the second subsystem. Moreover, the edge between states \(3\) and \(4\) is not on the first subsystem. Therefore, by Theorem~\ref{thm:211-reduction}, the first party can identify state \(2\) from all the other states. 
	
	The result is that four states \(\{1,3,4,5\}\) are left to be distinguished. For states \(\{1,3,4,5\}\), the vector of the numbers of pairwise orthogonal relations is \((2,2,1,1)\). By Lemma~\ref{lem:four-fourpartite-3111-2211}, states \(\{1,3,4,5\}\) can be perfectly distinguished by LOCC. Therefore, five quadripartite OPSs in case \((7\text{-}93)\) can be perfectly distinguished by LOCC. 
	
	(16) In case \((7\text{-}96)\), state \(2\) is orthogonal to states \(\{4,5\}\) on the second subsystem, is orthogonal to state \(3\) on the third subsystem, and is orthogonal to state \(1\) on the first subsystem. Moreover, the edge between states \(1\) and \(3\) is not on the second subsystem. Therefore, by Theorem~\ref{thm:211-reduction}, the second party can identify state \(2\) from all the other states. 
	
	The result is that four states \(\{1,3,4,5\}\) are left to be distinguished. For states \(\{1,3,4,5\}\), the vector of the numbers of pairwise orthogonal relations is \((3,1,1,1)\). By Lemma~\ref{lem:four-fourpartite-3111-2211}, states \(\{1,3,4,5\}\) can be perfectly distinguished by LOCC. Therefore, five quadripartite OPSs in case \((7\text{-}96)\) can be perfectly distinguished by LOCC. 
	
	(17) In cases \((7\text{-}100)\) and \((7\text{-}101)\), state \(2\) is orthogonal to states \(\{1,5\}\) on the first subsystem, is orthogonal to state \(4\) on the third subsystem, and is orthogonal to state \(3\) on the second subsystem. Moreover, the edge between states \(3\) and \(4\) is not on the first subsystem. Therefore, by Theorem~\ref{thm:211-reduction}, the first party can identify state \(2\) from all the other states. 
	
	The result is that four states \(\{1,3,4,5\}\) are left to be distinguished. For states \(\{1,3,4,5\}\), the vector of the numbers of pairwise orthogonal relations is \((2,2,1,1)\). By Lemma~\ref{lem:four-fourpartite-3111-2211}, states \(\{1,3,4,5\}\) can be perfectly distinguished by LOCC. Therefore, five quadripartite OPSs in cases \((7\text{-}100)\) and \((7\text{-}101)\) can be perfectly distinguished by LOCC. 
	
	(18) In cases \((7\text{-}102)\) and \((7\text{-}103)\), state \(2\) is orthogonal to states \(\{1,5\}\) on the first subsystem, is orthogonal to state \(4\) on the second subsystem, and is orthogonal to state \(3\) on the third subsystem. Moreover, the edge between states \(3\) and \(4\) is not on the first subsystem. Therefore, by Theorem~\ref{thm:211-reduction}, the first party can identify state \(2\) from all the other states. 
	
	The result is that four states \(\{1,3,4,5\}\) are left to be distinguished. For states \(\{1,3,4,5\}\), the vector of the numbers of pairwise orthogonal relations is \((2,2,1,1)\). By Lemma~\ref{lem:four-fourpartite-3111-2211}, states \(\{1,3,4,5\}\) can be perfectly distinguished by LOCC. Therefore, five quadripartite OPSs in cases \((7\text{-}102)\) and \((7\text{-}103)\) can be perfectly distinguished by LOCC.
	
	(19) In cases \((7\text{-}105)\)--\((7\text{-}109)\), state \(1\) is orthogonal to states \(\{4,5\}\) on the third subsystem and is orthogonal to states \(\{2,3\}\) on the first subsystem. Moreover, the edge between states \(2\) and \(3\) is not on the third subsystem. Therefore, by Theorem~\ref{thm:22-reduction}, the third party can identify state \(1\) from all the other states.
	
	The result is that four states \(\{2,3,4,5\}\) are left to be distinguished. For states \(\{2,3,4,5\}\), the vector of the numbers of pairwise orthogonal relations is \((2,3,0,1)\). This means that there is no orthogonal relation among states \(\{2,3,4,5\}\) on the third subsystem. Thus, states \(\{2,3,4,5\}\) can be regarded as four tripartite OPSs with the vector \((3,2,1)\) when the third subsystem is omitted. By Lemma~\ref{lem:four-tripartite-411-321}, states \(\{2,3,4,5\}\) can be perfectly distinguished by LOCC. Therefore, five quadripartite OPSs in cases \((7\text{-}105)\)--\((7\text{-}109)\) can be perfectly distinguished by LOCC.
	
	(20) In cases \((7\text{-}110)\), \((7\text{-}111)\), \((7\text{-}112)\), \((7\text{-}114)\), and \((7\text{-}116)\), state \(1\) is orthogonal to states \(\{2,3\}\) on the first subsystem, is orthogonal to state \(4\) on the third subsystem, and is orthogonal to state \(5\) on the fourth subsystem. Moreover, the edge between states \(4\) and \(5\) is not on the first subsystem. Therefore, by Theorem~\ref{thm:211-reduction}, the first party can identify state \(1\) from all the other states.
	
	The result is that four states \(\{2,3,4,5\}\) are left to be distinguished. For states \(\{2,3,4,5\}\), the vector of the numbers of pairwise orthogonal relations is \((2,3,1,0)\). This means that there is no orthogonal relation among states \(\{2,3,4,5\}\) on the fourth subsystem. Thus, states \(\{2,3,4,5\}\) can be regarded as four tripartite OPSs with the vector \((3,2,1)\) when the fourth subsystem is omitted. By Lemma~\ref{lem:four-tripartite-411-321}, states \(\{2,3,4,5\}\) can be perfectly distinguished by LOCC. Therefore, five quadripartite OPSs in these cases can be perfectly distinguished by LOCC.
	
	(21) In case \((7\text{-}113)\), state \(5\) is orthogonal to states \(\{2,3\}\) on the second subsystem, is orthogonal to state \(4\) on the first subsystem, and is orthogonal to state \(1\) on the fourth subsystem. Moreover, the edge between states \(1\) and \(4\) is not on the second subsystem. Therefore, by Theorem~\ref{thm:211-reduction}, the second party can identify state \(5\) from all the other states.
	
	The result is that four states \(\{1,2,3,4\}\) are left to be distinguished. For states \(\{1,2,3,4\}\), the vector of the numbers of pairwise orthogonal relations is \((3,1,2,0)\). This means that there is no orthogonal relation among states \(\{1,2,3,4\}\) on the fourth subsystem. Thus, states \(\{1,2,3,4\}\) can be regarded as four tripartite OPSs with the vector \((3,2,1)\) when the fourth subsystem is omitted. By Lemma~\ref{lem:four-tripartite-411-321}, states \(\{1,2,3,4\}\) can be perfectly distinguished by LOCC. Therefore, five quadripartite OPSs in case \((7\text{-}113)\) can be perfectly distinguished by LOCC.
	
	(22) In case \((7\text{-}118)\), state \(3\) is orthogonal to states \(\{1,5\}\) on the first subsystem, is orthogonal to state \(2\) on the second subsystem, and is orthogonal to state \(4\) on the third subsystem. Moreover, the edge between states \(2\) and \(4\) is not on the first subsystem. Therefore, by Theorem~\ref{thm:211-reduction}, the first party can identify state \(3\) from all the other states. 
	
	The result is that four states \(\{1,2,4,5\}\) are left to be distinguished. For states \(\{1,2,4,5\}\), the vector of the numbers of pairwise orthogonal relations is \((2,2,1,1)\). By Lemma~\ref{lem:four-fourpartite-3111-2211}, states \(\{1,2,4,5\}\) can be perfectly distinguished by LOCC. Therefore, five quadripartite OPSs in case \((7\text{-}118)\) can be perfectly distinguished by LOCC.	

	\subsection*{7. Category \((4,2,2,2)\)}
	
	The orthogonality graphs of five quadripartite OPSs with the vector of the numbers of pairwise orthogonal relations \((4,2,2,2)\) are shown in Fig.~\ref{fig:4222}. There are forty different cases in this category, labeled \((8\text{-}1)\)--\((8\text{-}40)\).
	
	(1) In cases \((8\text{-}1)\) and \((8\text{-}2)\), state \(1\) is orthogonal to states \(\{2,3,4,5\}\) on the first subsystem. This means that the degree of state \(1\) on the first subsystem is \(4\). According to Lemma~\ref{lem:degree4-removal}, the first party can identify state \(1\) from all the other states. The result is that four states \(\{2,3,4,5\}\) are left to be distinguished.
	
	For states \(\{2,3,4,5\}\), the vector of the numbers of pairwise orthogonal relations is \((0,2,2,2)\). This means that there is no orthogonal relation among states \(\{2,3,4,5\}\) on the first subsystem. Thus, these four states can be regarded as four tripartite OPSs with the vector \((2,2,2)\) when the first subsystem is omitted. By Lemma~\ref{lem:four-tripartite-222-dist}, states \(\{2,3,4,5\}\) can be perfectly distinguished by LOCC. Therefore, five quadripartite OPSs in cases \((8\text{-}1)\) and \((8\text{-}2)\) can be perfectly distinguished by LOCC.
	
	(2) In cases \((8\text{-}4)\), \((8\text{-}5)\), \((8\text{-}6)\), \((8\text{-}7)\), \((8\text{-}8)\), \((8\text{-}9)\), and \((8\text{-}10)\), state \(1\) is orthogonal to states \(\{2,3,4\}\) on the first subsystem and is orthogonal to state \(5\) on the second subsystem. Hence the degree of state \(1\) is \(3\) on the first subsystem and \(1\) on the second subsystem. According to Lemma~\ref{lem:31-removal}, state \(1\) can be identified by LOCC. The result is that four states \(\{2,3,4,5\}\) are left to be distinguished.
	
	For these cases, the vector of the numbers of pairwise orthogonal relations of states \(\{2,3,4,5\}\) is \((1,1,2,2)\), up to a permutation of parties. Hence these four states can be regarded as four quadripartite OPSs with the vector \((2,2,1,1)\). By Lemma~\ref{lem:four-fourpartite-3111-2211}, states \(\{2,3,4,5\}\) can be perfectly distinguished by LOCC. Therefore, five quadripartite OPSs in cases \((8\text{-}4)\)--\((8\text{-}10)\) can be perfectly distinguished by LOCC.
	
	(3) In cases \((8\text{-}11)\), \((8\text{-}12)\), \((8\text{-}13)\), and \((8\text{-}14)\), state \(1\) is orthogonal to states \(\{2,3,4\}\) on the first subsystem and is orthogonal to state \(5\) on the third subsystem. Hence the degree of state \(1\) is \(3\) on the first subsystem and \(1\) on the third subsystem. According to Lemma~\ref{lem:31-removal}, state \(1\) can be identified by LOCC. The result is that four states \(\{2,3,4,5\}\) are left to be distinguished.
	
	For these cases, the vector of the numbers of pairwise orthogonal relations of states \(\{2,3,4,5\}\) is \((1,2,1,2)\), up to a permutation of parties. Hence these four states can be regarded as four quadripartite OPSs with the vector \((2,2,1,1)\). By Lemma~\ref{lem:four-fourpartite-3111-2211}, states \(\{2,3,4,5\}\) can be perfectly distinguished by LOCC. Therefore, five quadripartite OPSs in cases \((8\text{-}11)\)--\((8\text{-}14)\) can be perfectly distinguished by LOCC.
	
	(4) In cases \((8\text{-}15)\), \((8\text{-}16)\), \((8\text{-}17)\), \((8\text{-}18)\), \((8\text{-}19)\), \((8\text{-}20)\), and \((8\text{-}21)\), state \(1\) is orthogonal to states \(\{2,3,4\}\) on the first subsystem and is orthogonal to state \(5\) on the fourth subsystem. Hence the degree of state \(1\) is \(3\) on the first subsystem and \(1\) on the fourth subsystem. According to Lemma~\ref{lem:31-removal}, state \(1\) can be identified by LOCC. The result is that four states \(\{2,3,4,5\}\) are left to be distinguished.
	
	For these cases, the vector of the numbers of pairwise orthogonal relations of states \(\{2,3,4,5\}\) is \((1,2,2,1)\), up to a permutation of parties. Hence these four states can be regarded as four quadripartite OPSs with the vector \((2,2,1,1)\). By Lemma~\ref{lem:four-fourpartite-3111-2211}, states \(\{2,3,4,5\}\) can be perfectly distinguished by LOCC. Therefore, five quadripartite OPSs in cases \((8\text{-}15)\)--\((8\text{-}21)\) can be perfectly distinguished by LOCC.
	
	(5) In case \((8\text{-}22)\), state \(1\) is orthogonal to states \(\{4,5\}\) on the second subsystem and is orthogonal to states \(\{2,3\}\) on the first subsystem. Moreover, the edge between states \(2\) and \(3\) is not on the second subsystem. Therefore, by Theorem~\ref{thm:22-reduction}, the second party can identify state \(1\) from all the other states.
	
	The result is that four states \(\{2,3,4,5\}\) are left to be distinguished. For states \(\{2,3,4,5\}\), the vector of the numbers of pairwise orthogonal relations is \((2,0,2,2)\). Thus, these four states can be regarded as four tripartite OPSs with the vector \((2,2,2)\) when the second subsystem is omitted. By Lemma~\ref{lem:four-tripartite-222-dist}, states \(\{2,3,4,5\}\) can be perfectly distinguished by LOCC. Therefore, five quadripartite OPSs in case \((8\text{-}22)\) can be perfectly distinguished by LOCC.
	
	(6) In cases \((8\text{-}24)\), \((8\text{-}29)\), and \((8\text{-}38)\), state \(5\) is orthogonal to states \(\{2,3\}\) on the fourth subsystem, is orthogonal to state \(4\) on the first subsystem, and is orthogonal to state \(1\) on the third subsystem. Moreover, the edge between states \(1\) and \(4\) is not on the fourth subsystem. Therefore, according to Theorem~\ref{thm:211-reduction}, the fourth party can identify state \(5\) from all the other states.
	
	The result is that four states \(\{1,2,3,4\}\) are left to be distinguished. For states \(\{1,2,3,4\}\), the vector of the numbers of pairwise orthogonal relations is \((3,2,1,0)\). This means that there is no orthogonal relation among states \(\{1,2,3,4\}\) on the fourth subsystem. Thus, states \(\{1,2,3,4\}\) can be regarded as four tripartite OPSs with the vector \((3,2,1)\) when the fourth subsystem is omitted. By Lemma~\ref{lem:four-tripartite-411-321}, states \(\{1,2,3,4\}\) can be perfectly distinguished by LOCC. Therefore, five quadripartite OPSs in cases \((8\text{-}24)\), \((8\text{-}29)\), and \((8\text{-}38)\) can be perfectly distinguished by LOCC.
	
	(7) In cases \((8\text{-}26)\), \((8\text{-}27)\), \((8\text{-}28)\), and \((8\text{-}30)\), state \(1\) is orthogonal to states \(\{2,3\}\) on the first subsystem, is orthogonal to state \(4\) on the second subsystem, and is orthogonal to state \(5\) on the third subsystem. Moreover, the edge between states \(4\) and \(5\) is not on the first subsystem. Therefore, according to Theorem~\ref{thm:211-reduction}, the first party can identify state \(1\) from all the other states.
	
	The result is that four states \(\{2,3,4,5\}\) are left to be distinguished. For these cases, the vector of the numbers of pairwise orthogonal relations of states \(\{2,3,4,5\}\) is \((2,1,1,2)\), up to a permutation of parties. Hence these four states can be regarded as four quadripartite OPSs with the vector \((2,2,1,1)\). By Lemma~\ref{lem:four-fourpartite-3111-2211}, states \(\{2,3,4,5\}\) can be perfectly distinguished by LOCC. Therefore, five quadripartite OPSs in cases \((8\text{-}26)\), \((8\text{-}27)\), \((8\text{-}28)\), and \((8\text{-}30)\) can be perfectly distinguished by LOCC.
	
	(8) In case \((8\text{-}31)\), state \(3\) is orthogonal to states \(\{1,5\}\) on the first subsystem, is orthogonal to state \(2\) on the second subsystem, and is orthogonal to state \(4\) on the fourth subsystem. Moreover, the edge between states \(2\) and \(4\) is not on the first subsystem. Therefore, according to Theorem~\ref{thm:211-reduction}, the first party can identify state \(3\) from all the other states.
	
	The result is that four states \(\{1,2,4,5\}\) are left to be distinguished. For states \(\{1,2,4,5\}\), the vector of the numbers of pairwise orthogonal relations is \((2,1,2,1)\), up to a permutation of parties. Hence these four states can be regarded as four quadripartite OPSs with the vector \((2,2,1,1)\). By Lemma~\ref{lem:four-fourpartite-3111-2211}, states \(\{1,2,4,5\}\) can be perfectly distinguished by LOCC. Therefore, five quadripartite OPSs in case \((8\text{-}31)\) can be perfectly distinguished by LOCC.
	
	(9) In cases \((8\text{-}32)\) and \((8\text{-}40)\), state \(5\) is orthogonal to states \(\{1,2\}\) on the fourth subsystem, is orthogonal to state \(4\) on the first subsystem, and is orthogonal to state \(3\) on the third subsystem. Moreover, the edge between states \(3\) and \(4\) is not on the fourth subsystem. Therefore, according to Theorem~\ref{thm:211-reduction}, the fourth party can identify state \(5\) from all the other states.
	
	The result is that four states \(\{1,2,3,4\}\) are left to be distinguished. For states \(\{1,2,3,4\}\), the vector of the numbers of pairwise orthogonal relations is \((3,2,1,0)\). This means that there is no orthogonal relation among states \(\{1,2,3,4\}\) on the fourth subsystem. Thus, states \(\{1,2,3,4\}\) can be regarded as four tripartite OPSs with the vector \((3,2,1)\) when the fourth subsystem is omitted. By Lemma~\ref{lem:four-tripartite-411-321}, states \(\{1,2,3,4\}\) can be perfectly distinguished by LOCC. Therefore, five quadripartite OPSs in cases \((8\text{-}32)\) and \((8\text{-}40)\) can be perfectly distinguished by LOCC.
	
	(10) In cases \((8\text{-}33)\) and \((8\text{-}34)\), state \(1\) is orthogonal to states \(\{2,3\}\) on the first subsystem, is orthogonal to state \(4\) on the second subsystem, and is orthogonal to state \(5\) on the fourth subsystem. Moreover, the edge between states \(4\) and \(5\) is not on the first subsystem. Therefore, according to Theorem~\ref{thm:211-reduction}, the first party can identify state \(1\) from all the other states.
	
	The result is that four states \(\{2,3,4,5\}\) are left to be distinguished. For states \(\{2,3,4,5\}\), the vector of the numbers of pairwise orthogonal relations is \((2,1,2,1)\), up to a permutation of parties. Hence these four states can be regarded as four quadripartite OPSs with the vector \((2,2,1,1)\). By Lemma~\ref{lem:four-fourpartite-3111-2211}, states \(\{2,3,4,5\}\) can be perfectly distinguished by LOCC. Therefore, five quadripartite OPSs in cases \((8\text{-}33)\) and \((8\text{-}34)\) can be perfectly distinguished by LOCC.
	
	(11) In case \((8\text{-}35)\), state \(5\) is orthogonal to states \(\{1,3\}\) on the fourth subsystem and is orthogonal to states \(\{2,4\}\) on the first subsystem. Moreover, the edge between states \(2\) and \(4\) is not on the fourth subsystem. Therefore, according to Theorem~\ref{thm:22-reduction}, the fourth party can identify state \(5\) from all the other states.
	
	The result is that four states \(\{1,2,3,4\}\) are left to be distinguished. For states \(\{1,2,3,4\}\), the vector of the numbers of pairwise orthogonal relations is \((2,2,2,0)\). This means that there is no orthogonal relation among states \(\{1,2,3,4\}\) on the fourth subsystem. Thus, states \(\{1,2,3,4\}\) can be regarded as four tripartite OPSs with the vector \((2,2,2)\) when the fourth subsystem is omitted. By Lemma~\ref{lem:four-tripartite-222-dist}, states \(\{1,2,3,4\}\) can be perfectly distinguished by LOCC. Therefore, five quadripartite OPSs in case \((8\text{-}35)\) can be perfectly distinguished by LOCC.
	
	(12) In cases \((8\text{-}36)\) and \((8\text{-}37)\), state \(1\) is orthogonal to states \(\{2,3\}\) on the first subsystem and is orthogonal to states \(\{4,5\}\) on the third subsystem. Moreover, the edge between states \(4\) and \(5\) is not on the first subsystem. Therefore, according to Theorem~\ref{thm:22-reduction}, the first party can identify state \(1\) from all the other states.
	
	The result is that four states \(\{2,3,4,5\}\) are left to be distinguished. For these two cases, the vector of the numbers of pairwise orthogonal relations of states \(\{2,3,4,5\}\) is \((2,2,0,2)\). This means that there is no orthogonal relation among states \(\{2,3,4,5\}\) on the third subsystem. Thus, states \(\{2,3,4,5\}\) can be regarded as four tripartite OPSs with the vector \((2,2,2)\) when the third subsystem is omitted. By Lemma~\ref{lem:four-tripartite-222-dist}, states \(\{2,3,4,5\}\) can be perfectly distinguished by LOCC. Therefore, five quadripartite OPSs in cases \((8\text{-}36)\) and \((8\text{-}37)\) can be perfectly distinguished by LOCC.
	
	\subsection*{8. Category \((3,3,3,1)\)}
	
	The orthogonality graphs of five quadripartite OPSs with the vector of the numbers of pairwise orthogonal relations \((3,3,3,1)\) are shown in Fig.~\ref{fig:3331}. There are thirty-three different cases in this category, labeled \((9\text{-}1)\)--\((9\text{-}33)\).
	
	(1) In cases \((9\text{-}1)\), \((9\text{-}2)\), \((9\text{-}3)\), and \((9\text{-}4)\), state \(1\) is orthogonal to states \(\{2,3,4\}\) on the first subsystem and is orthogonal to state \(5\) on the second subsystem. Hence the degree of state \(1\) is \(3\) on the first subsystem and \(1\) on the second subsystem. According to Lemma~\ref{lem:31-removal}, state \(1\) can be identified by LOCC. The result is that four states \(\{2,3,4,5\}\) are left to be distinguished.
	
	For these cases, the vector of the numbers of pairwise orthogonal relations of states \(\{2,3,4,5\}\) is \((0,2,3,1)\). This means that there is no orthogonal relation among states \(\{2,3,4,5\}\) on the first subsystem. Thus, states \(\{2,3,4,5\}\) can be regarded as four tripartite OPSs with the vector \((3,2,1)\), up to a permutation of parties, when the first subsystem is omitted. By Lemma~\ref{lem:four-tripartite-411-321}, states \(\{2,3,4,5\}\) can be perfectly distinguished by LOCC. Therefore, five quadripartite OPSs in cases \((9\text{-}1)\)--\((9\text{-}4)\) can be perfectly distinguished by LOCC.
	
	(2) In cases \((9\text{-}5)\), \((9\text{-}6)\), \((9\text{-}7)\), and \((9\text{-}8)\), state \(1\) is orthogonal to states \(\{2,3,4\}\) on the first subsystem and is orthogonal to state \(5\) on the third subsystem. Hence the degree of state \(1\) is \(3\) on the first subsystem and \(1\) on the third subsystem. According to Lemma~\ref{lem:31-removal}, state \(1\) can be identified by LOCC. The result is that four states \(\{2,3,4,5\}\) are left to be distinguished.
	
	For these cases, the vector of the numbers of pairwise orthogonal relations of states \(\{2,3,4,5\}\) is \((0,3,2,1)\). This means that there is no orthogonal relation among states \(\{2,3,4,5\}\) on the first subsystem. Thus, states \(\{2,3,4,5\}\) can be regarded as four tripartite OPSs with the vector \((3,2,1)\), up to a permutation of parties, when the first subsystem is omitted. By Lemma~\ref{lem:four-tripartite-411-321}, states \(\{2,3,4,5\}\) can be perfectly distinguished by LOCC. Therefore, five quadripartite OPSs in cases \((9\text{-}5)\)--\((9\text{-}8)\) can be perfectly distinguished by LOCC.
	
	(3) In case \((9\text{-}9)\), state \(1\) is orthogonal to states \(\{2,3,4\}\) on the first subsystem and is orthogonal to state \(5\) on the fourth subsystem. Hence the degree of state \(1\) is \(3\) on the first subsystem and \(1\) on the fourth subsystem. According to Lemma~\ref{lem:31-removal}, state \(1\) can be identified by LOCC. The result is that four states \(\{2,3,4,5\}\) are left to be distinguished.
	
	For states \(\{2,3,4,5\}\), the vector of the numbers of pairwise orthogonal relations is \((0,3,3,0)\). This means that there is no orthogonal relation among states \(\{2,3,4,5\}\) on the first and fourth subsystems. Thus, states \(\{2,3,4,5\}\) can be regarded as four bipartite OPSs when the first and fourth subsystems are omitted. By Lemma~\ref{lem:four-bipartite-ops}, states \(\{2,3,4,5\}\) can be perfectly distinguished by LOCC. Therefore, five quadripartite OPSs in case \((9\text{-}9)\) can be perfectly distinguished by LOCC.
	
	(4) In cases \((9\text{-}10)\), \((9\text{-}11)\), \((9\text{-}12)\), \((9\text{-}13)\), and \((9\text{-}14)\), state \(1\) is orthogonal to states \(\{2,3\}\) on the first subsystem and is orthogonal to states \(\{4,5\}\) on the second subsystem. Moreover, the edge between states \(4\) and \(5\) is not on the first subsystem. Therefore, by Theorem~\ref{thm:22-reduction}, the first party can identify state \(1\) from all the other states.
	
	The result is that four states \(\{2,3,4,5\}\) are left to be distinguished. For these cases, the vector of the numbers of pairwise orthogonal relations of states \(\{2,3,4,5\}\) is \((1,1,3,1)\), up to a permutation of parties. Hence these four states can be regarded as four quadripartite OPSs with the vector \((3,1,1,1)\). By Lemma~\ref{lem:four-fourpartite-3111-2211}, states \(\{2,3,4,5\}\) can be perfectly distinguished by LOCC. Therefore, five quadripartite OPSs in cases \((9\text{-}10)\)--\((9\text{-}14)\) can be perfectly distinguished by LOCC.
	
	(5) In case \((9\text{-}15)\), state \(2\) is orthogonal to states \(\{4,5\}\) on the third subsystem, is orthogonal to state \(1\) on the first subsystem, and is orthogonal to state \(3\) on the second subsystem. Moreover, the edge between states \(1\) and \(3\) is not on the third subsystem. Therefore, by Theorem~\ref{thm:211-reduction}, the third party can identify state \(2\) from all the other states.
	
	The result is that four states \(\{1,3,4,5\}\) are left to be distinguished. For states \(\{1,3,4,5\}\), the vector of the numbers of pairwise orthogonal relations is \((2,2,1,1)\). By Lemma~\ref{lem:four-fourpartite-3111-2211}, states \(\{1,3,4,5\}\) can be perfectly distinguished by LOCC. Therefore, five quadripartite OPSs in case \((9\text{-}15)\) can be perfectly distinguished by LOCC.
	
	(6) In cases \((9\text{-}16)\)--\((9\text{-}29)\), state \(1\) is orthogonal to states \(\{2,3\}\) on the first subsystem, is orthogonal to state \(4\) on the second subsystem, and is orthogonal to state \(5\) on the third subsystem. Moreover, the edge between states \(4\) and \(5\) is not on the first subsystem. Therefore, by Theorem~\ref{thm:211-reduction}, the first party can identify state \(1\) from all the other states.
	
	The result is that four states \(\{2,3,4,5\}\) are left to be distinguished. For these cases, the vector of the numbers of pairwise orthogonal relations of states \(\{2,3,4,5\}\) is \((1,2,2,1)\), up to a permutation of parties. Hence these four states can be regarded as four quadripartite OPSs with the vector \((2,2,1,1)\). By Lemma~\ref{lem:four-fourpartite-3111-2211}, states \(\{2,3,4,5\}\) can be perfectly distinguished by LOCC. Therefore, five quadripartite OPSs in cases \((9\text{-}16)\)--\((9\text{-}29)\) can be perfectly distinguished by LOCC.
	
	(7) In case \((9\text{-}33)\), state \(1\) is orthogonal to states \(\{2,3\}\) on the first subsystem and is orthogonal to states \(\{4,5\}\) on the third subsystem. Moreover, the edge between states \(4\) and \(5\) is not on the first subsystem. Therefore, by Theorem~\ref{thm:22-reduction}, the first party can identify state \(1\) from all the other states.
	
	The result is that four states \(\{2,3,4,5\}\) are left to be distinguished. For states \(\{2,3,4,5\}\), the vector of the numbers of pairwise orthogonal relations is \((1,3,1,1)\), up to a permutation of parties. Hence these four states can be regarded as four quadripartite OPSs with the vector \((3,1,1,1)\). By Lemma~\ref{lem:four-fourpartite-3111-2211}, states \(\{2,3,4,5\}\) can be perfectly distinguished by LOCC. Therefore, five quadripartite OPSs in case \((9\text{-}33)\) can be perfectly distinguished by LOCC.
	
	Combining the above cases, the listed graphs in the category \((3,3,3,1)\) can be perfectly distinguished by LOCC.
	
	\subsection*{9. Category \((3,3,2,2)\)}
	
	The orthogonality graphs of five quadripartite OPSs with the vector of the numbers of pairwise orthogonal relations \((3,3,2,2)\) are shown in Fig.~\ref{fig:3322-part2}. There are sixty-seven different cases in this category, labeled \((10\text{-}1)\)--\((10\text{-}67)\).
	
	(1) In cases \((10\text{-}1)\)--\((10\text{-}7)\), state \(1\) is orthogonal to states \(\{2,3,4\}\) on the first subsystem and is orthogonal to state \(5\) on the second subsystem. Hence the degree of state \(1\) is \(3\) on the first subsystem and \(1\) on the second subsystem. According to Lemma~\ref{lem:31-removal}, state \(1\) can be identified by LOCC. The result is that four states \(\{2,3,4,5\}\) are left to be distinguished.
	
	For these cases, the vector of the numbers of pairwise orthogonal relations of states \(\{2,3,4,5\}\) is \((0,2,2,2)\). This means that there is no orthogonal relation among states \(\{2,3,4,5\}\) on the first subsystem. Thus, states \(\{2,3,4,5\}\) can be regarded as four tripartite OPSs with the vector \((2,2,2)\) when the first subsystem is omitted. By Lemma~\ref{lem:four-tripartite-222-dist}, states \(\{2,3,4,5\}\) can be perfectly distinguished by LOCC. Therefore, five quadripartite OPSs in cases \((10\text{-}1)\)--\((10\text{-}7)\) can be perfectly distinguished by LOCC.
	
	(2) In cases \((10\text{-}9)\)--\((10\text{-}12)\), state \(1\) is orthogonal to states \(\{2,3,4\}\) on the first subsystem and is orthogonal to state \(5\) on the third subsystem. Hence the degree of state \(1\) is \(3\) on the first subsystem and \(1\) on the third subsystem. According to Lemma~\ref{lem:31-removal}, state \(1\) can be identified by LOCC. The result is that four states \(\{2,3,4,5\}\) are left to be distinguished.
	
	For these cases, the vector of the numbers of pairwise orthogonal relations of states \(\{2,3,4,5\}\) is \((0,3,1,2)\). This means that there is no orthogonal relation among states \(\{2,3,4,5\}\) on the first subsystem. Thus, states \(\{2,3,4,5\}\) can be regarded as four tripartite OPSs with the vector \((3,1,2)\) when the first subsystem is omitted. By Lemma~\ref{lem:four-tripartite-411-321}, states \(\{2,3,4,5\}\) can be perfectly distinguished by LOCC. Therefore, five quadripartite OPSs in cases \((10\text{-}9)\)--\((10\text{-}12)\) can be perfectly distinguished by LOCC.
	
	(3) In cases \((10\text{-}13)\)--\((10\text{-}17)\), state \(1\) is orthogonal to states \(\{2,3,4\}\) on the first subsystem and is orthogonal to state \(5\) on the fourth subsystem. Hence the degree of state \(1\) is \(3\) on the first subsystem and \(1\) on the fourth subsystem. According to Lemma~\ref{lem:31-removal}, state \(1\) can be identified by LOCC. The result is that four states \(\{2,3,4,5\}\) are left to be distinguished.
	
	For these cases, the vector of the numbers of pairwise orthogonal relations of states \(\{2,3,4,5\}\) is \((0,3,2,1)\). This means that there is no orthogonal relation among states \(\{2,3,4,5\}\) on the first subsystem. Thus, states \(\{2,3,4,5\}\) can be regarded as four tripartite OPSs with the vector \((3,2,1)\) when the first subsystem is omitted. By Lemma~\ref{lem:four-tripartite-411-321}, states \(\{2,3,4,5\}\) can be perfectly distinguished by LOCC. Therefore, five quadripartite OPSs in cases \((10\text{-}13)\)--\((10\text{-}17)\) can be perfectly distinguished by LOCC.
	
	(4) In cases \((10\text{-}18)\)--\((10\text{-}25)\), state \(1\) is orthogonal to states \(\{2,3\}\) on the first subsystem and is orthogonal to states \(\{4,5\}\) on the second subsystem. Moreover, the edge between states \(4\) and \(5\) is not on the first subsystem. Therefore, by Theorem~\ref{thm:22-reduction}, the first party can identify state \(1\) from all the other states.
	
	The result is that four states \(\{2,3,4,5\}\) are left to be distinguished. For these cases, the vector of the numbers of pairwise orthogonal relations of states \(\{2,3,4,5\}\) is \((1,1,2,2)\). Hence these four states belong to the quadripartite type \((2,2,1,1)\). By Lemma~\ref{lem:four-fourpartite-3111-2211}, states \(\{2,3,4,5\}\) can be perfectly distinguished by LOCC. Therefore, five quadripartite OPSs in cases \((10\text{-}18)\)--\((10\text{-}25)\) can be perfectly distinguished by LOCC.
	
	(5) In cases \((10\text{-}26)\) and \((10\text{-}56)\), state \(4\) is orthogonal to states \(\{2,3\}\) on the third subsystem, is orthogonal to state \(1\) on the second subsystem, and is orthogonal to state \(5\) on the first subsystem. Moreover, the edge between states \(1\) and \(5\) is not on the third subsystem. Therefore, by Theorem~\ref{thm:211-reduction}, the third party can identify state \(4\) from all the other states. 
	
	The result is that four states \(\{1,2,3,5\}\) are left to be distinguished. For states \(\{1,2,3,5\}\), the vector of the numbers of pairwise orthogonal relations is \((2,2,0,2)\). This means that there is no orthogonal relation among states \(\{1,2,3,5\}\) on the third subsystem. Thus, states \(\{1,2,3,5\}\) can be regarded as four tripartite OPSs with the vector \((2,2,2)\) when the third subsystem is omitted. By Lemma~\ref{lem:four-tripartite-222-dist}, states \(\{1,2,3,5\}\) can be perfectly distinguished by LOCC. Therefore, five quadripartite OPSs in cases \((10\text{-}26)\) and \((10\text{-}56)\) can be perfectly distinguished by LOCC.
	
	(6) In cases \((10\text{-}28)\)--\((10\text{-}38)\), \((10\text{-}40)\), and \((10\text{-}42)\), state \(1\) is orthogonal to states \(\{2,3\}\) on the first subsystem, is orthogonal to state \(4\) on the second subsystem, and is orthogonal to state \(5\) on the third subsystem. Moreover, the edge between states \(4\) and \(5\) is not on the first subsystem. Therefore, by Theorem~\ref{thm:211-reduction}, the first party can identify state \(1\) from all the other states.
	
	The result is that four states \(\{2,3,4,5\}\) are left to be distinguished. For these cases, the vector of the numbers of pairwise orthogonal relations of states \(\{2,3,4,5\}\) is \((1,2,1,2)\). Hence these four states belong to the quadripartite type \((2,2,1,1)\). By Lemma~\ref{lem:four-fourpartite-3111-2211}, states \(\{2,3,4,5\}\) can be perfectly distinguished by LOCC. Therefore, five quadripartite OPSs in cases \((10\text{-}28)\)--\((10\text{-}38)\), \((10\text{-}40)\), and \((10\text{-}42)\) can be perfectly distinguished by LOCC.
	
	(7) In cases \((10\text{-}39)\), \((10\text{-}54)\), and \((10\text{-}67)\), state \(2\) is orthogonal to states \(\{3,4\}\) on the second subsystem, is orthogonal to state \(1\) on the first subsystem, and is orthogonal to state \(5\) on the fourth subsystem. Moreover, the edge between states \(1\) and \(5\) is not on the second subsystem. Therefore, by Theorem~\ref{thm:211-reduction}, the second party can identify state \(2\) from all the other states. 
	
	The result is that four states \(\{1,3,4,5\}\) are left to be distinguished. For these cases, the vector of the numbers of pairwise orthogonal relations of states \(\{1,3,4,5\}\) is \((2,1,2,1)\). Hence these four states belong to the quadripartite type \((2,2,1,1)\). By Lemma~\ref{lem:four-fourpartite-3111-2211}, states \(\{1,3,4,5\}\) can be perfectly distinguished by LOCC. Therefore, five quadripartite OPSs in cases \((10\text{-}39)\), \((10\text{-}54)\), and \((10\text{-}67)\) can be perfectly distinguished by LOCC.
	
	(8) In case \((10\text{-}44)\), state \(4\) is orthogonal to states \(\{2,3\}\) on the fourth subsystem, is orthogonal to state \(1\) on the second subsystem, and is orthogonal to state \(5\) on the first subsystem. Moreover, the edge between states \(1\) and \(5\) is not on the fourth subsystem. Therefore, by Theorem~\ref{thm:211-reduction}, the fourth party can identify state \(4\) from all the other states.
	
	The result is that four states \(\{1,2,3,5\}\) are left to be distinguished. For states \(\{1,2,3,5\}\), the vector of the numbers of pairwise orthogonal relations is \((2,2,2,0)\). This means that there is no orthogonal relation among states \(\{1,2,3,5\}\) on the fourth subsystem. Thus, states \(\{1,2,3,5\}\) can be regarded as four tripartite OPSs with the vector \((2,2,2)\) when the fourth subsystem is omitted. By Lemma~\ref{lem:four-tripartite-222-dist}, states \(\{1,2,3,5\}\) can be perfectly distinguished by LOCC. Therefore, five quadripartite OPSs in case \((10\text{-}44)\) can be perfectly distinguished by LOCC.
	
	(9) In cases \((10\text{-}45)\)--\((10\text{-}51)\), \((10\text{-}53)\), \((10\text{-}55)\), \((10\text{-}57)\), and \((10\text{-}58)\), state \(1\) is orthogonal to states \(\{2,3\}\) on the first subsystem, is orthogonal to state \(4\) on the second subsystem, and is orthogonal to state \(5\) on the fourth subsystem. Moreover, the edge between states \(4\) and \(5\) is not on the first subsystem. Therefore, by Theorem~\ref{thm:211-reduction}, the first party can identify state \(1\) from all the other states.
	
	The result is that four states \(\{2,3,4,5\}\) are left to be distinguished. For these cases, the vector of the numbers of pairwise orthogonal relations of states \(\{2,3,4,5\}\) is \((1,2,2,1)\). Hence these four states belong to the quadripartite type \((2,2,1,1)\). By Lemma~\ref{lem:four-fourpartite-3111-2211}, states \(\{2,3,4,5\}\) can be perfectly distinguished by LOCC. Therefore, five quadripartite OPSs in these cases can be perfectly distinguished by LOCC.
	
	(10) In cases \((10\text{-}52)\) and \((10\text{-}66)\), state \(2\) is orthogonal to states \(\{3,4\}\) on the second subsystem, is orthogonal to state \(1\) on the first subsystem, and is orthogonal to state \(5\) on the third subsystem. Moreover, the edge between states \(1\) and \(5\) is not on the second subsystem. Therefore, by Theorem~\ref{thm:211-reduction}, the second party can identify state \(2\) from all the other states. 
	
	The result is that four states \(\{1,3,4,5\}\) are left to be distinguished. For these cases, the vector of the numbers of pairwise orthogonal relations of states \(\{1,3,4,5\}\) is \((2,1,1,2)\). Hence these four states belong to the quadripartite type \((2,2,1,1)\). By Lemma~\ref{lem:four-fourpartite-3111-2211}, states \(\{1,3,4,5\}\) can be perfectly distinguished by LOCC. Therefore, five quadripartite OPSs in cases \((10\text{-}52)\) and \((10\text{-}66)\) can be perfectly distinguished by LOCC.
	
	(11) In cases \((10\text{-}60)\) and \((10\text{-}61)\), state \(1\) is orthogonal to states \(\{4,5\}\) on the third subsystem and is orthogonal to states \(\{2,3\}\) on the first subsystem. Moreover, the edge between states \(2\) and \(3\) is not on the third subsystem. Therefore, by Theorem~\ref{thm:22-reduction}, the third party can identify state \(1\) from all the other states.
	
	The result is that four states \(\{2,3,4,5\}\) are left to be distinguished. For these cases, the vector of the numbers of pairwise orthogonal relations of states \(\{2,3,4,5\}\) is \((1,3,0,2)\). This means that there is no orthogonal relation among states \(\{2,3,4,5\}\) on the third subsystem. Thus, states \(\{2,3,4,5\}\) can be regarded as four tripartite OPSs with the vector \((3,2,1)\) when the third subsystem is omitted. By Lemma~\ref{lem:four-tripartite-411-321}, states \(\{2,3,4,5\}\) can be perfectly distinguished by LOCC. Therefore, five quadripartite OPSs in cases \((10\text{-}60)\) and \((10\text{-}61)\) can be perfectly distinguished by LOCC.
	
	(12) In cases \((10\text{-}62)\)--\((10\text{-}65)\), state \(1\) is orthogonal to states \(\{2,3\}\) on the first subsystem, is orthogonal to state \(4\) on the third subsystem, and is orthogonal to state \(5\) on the fourth subsystem. Moreover, the edge between states \(4\) and \(5\) is not on the first subsystem. Therefore, by Theorem~\ref{thm:211-reduction}, the first party can identify state \(1\) from all the other states.
	
	The result is that four states \(\{2,3,4,5\}\) are left to be distinguished. For these cases, the vector of the numbers of pairwise orthogonal relations of states \(\{2,3,4,5\}\) is \((1,3,1,1)\). Hence these four states belong to the quadripartite type \((3,1,1,1)\). By Lemma~\ref{lem:four-fourpartite-3111-2211}, states \(\{2,3,4,5\}\) can be perfectly distinguished by LOCC. Therefore, five quadripartite OPSs in cases \((10\text{-}62)\)--\((10\text{-}65)\) can be perfectly distinguished by LOCC.

	\bibliography{references}
	
\end{document}